\documentclass[letter,12pt]{amsart}
\usepackage{wrapfig}
\usepackage{amsfonts}
\usepackage{amssymb}
\usepackage{amsmath,tabu}
\usepackage{amsthm}
\usepackage{comment}
\usepackage{graphics}
\usepackage{epstopdf}
\usepackage{pst-all,psfrag}
\usepackage[unicode]{hyperref}
\usepackage{mathtools}
\usepackage{subcaption}
\usepackage{multicol}
\usepackage{setspace}
\usepackage{bm}

\usepackage{natbib}

\usepackage{verbatim}
\usepackage{eurosym}
\usepackage[utf8]{inputenc}
\usepackage[english]{babel}
\usepackage[margin=1in]{geometry}
\usepackage{apptools}
\usepackage{tabularx}
\usepackage{enumitem}
\usepackage{bbm}

\usepackage{arydshln}

\newtheorem{theorem}{Theorem}
\newtheorem{proposition}[theorem]{Proposition}

\newtheorem{lemma}[theorem]{Lemma}

\newtheorem{task}[theorem]{Task}
\newtheorem{procedure}[theorem]{Procedure}

\newtheorem{corollary}[theorem]{Corollary}

\newtheorem{definition}[theorem]{Definition}

\theoremstyle{remark}
\newtheorem{remark}[theorem]{Remark}

\numberwithin{theorem}{section}

\newcommand{\eps}{\varepsilon}

\newcommand{\U}{\mathbf U}
\newcommand{\A}{\mathbf A}
\newcommand{\B}{\mathbf B}

\newcommand{\V}{\mathbf V}
\newcommand{\W}{\mathbf W}
\newcommand{\X}{\mathbf X}
\newcommand{\Y}{\mathbf Y}
\newcommand{\Z}{\mathbf Z}
\renewcommand{\v}{\mathbf v}
\renewcommand{\u}{\mathbf u}
\newcommand{\w}{\mathbf w}
\newcommand{\x}{\mathbf x}
\newcommand{\y}{\mathbf y}
\newcommand{\ba}{\boldsymbol\alpha}
\newcommand{\bb}{\boldsymbol\beta}

\newcommand{\s}{\mathfrak s}

\newcommand{\J}{\mathbf J}
\newcommand{\WL}{\bm{\mathcal{W}}}

\renewcommand{\aa}{\mathfrak a}

\newcommand{\ii}{\mathbf i}
\newcommand{\dd}{\mathrm d}

\newcommand{\E}{\mathbb E}

\newcommand{\T}{\mathsf T}

\newcommand{\la}{\langle}
\newcommand{\ra}{\rangle}

\begin{document}
\title{Canonical Correlation Analysis: review}

\author{Anna Bykhovskaya}
\address[Anna Bykhovskaya]{Duke University}
\email{anna.bykhovskaya@duke.edu}

\author{Vadim Gorin}
\address[Vadim Gorin]{University of California at Berkeley}
\email{vadicgor@gmail.com}


\date{\today}

\maketitle


\section*{Preface}

For over a century canonical correlations, variables, and related concepts have been studied across various fields, with contributions dating back to \cite{jordan1875essai} and \cite{harold1936relations}. This text surveys the evolution of canonical correlation analysis, a fundamental statistical tool, beginning with its foundational theorems and progressing to recent developments and open research problems. Along the way we introduce and review methods, notions, and fundamental concepts from linear algebra, random matrix theory, and high-dimensional statistics, placing particular emphasis on rigorous mathematical treatment.


The survey is intended for technically proficient graduate students and other researchers with an interest in this area.
The content is organized into five chapters, supplemented by six sets of exercises found in Chapter \ref{Section_exercises}. These exercises introduce additional material, reinforce key concepts, and serve to bridge ideas across chapters. We recommend the following sequence: first, solve Problem Set 0, then proceed with Chapter 1, solve Problem Set 1, and so on through the text.


The review grew out of lectures delivered by Gorin at the 2024 Random Matrix Theory Summer School at the University of Michigan, and from the topics class taught by Bykhovskaya at Duke in the Fall 2024. We are grateful to the organizers and participants of these events for their valuable contributions.
 Gorin's work was partially supported by NSF grant DMS - 2246449.


\newpage

\begin{singlespace}

\tableofcontents

\end{singlespace}

\section{Introduction and basic definitions}

\subsection{Opening remarks}

Broadly speaking, \emph{Canonical Correlation Analysis} (CCA) is a method for identifying common factors between two multi-dimensional objects. It can be seen as a companion to the \emph{Principal Component Analysis} (PCA):

\begin{itemize}
 \item PCA is used to find a signal (or factor) in a single large matrix with a large amount of noise. That is, PCA aims to explain and reduce the dimensionality of a single $K\times S$ matrix $X$.
 \item CCA is used to find a \emph{common} signal among two large matrices with a large amount of noise. That is, CCA aims to explain and reduce the dimensionality of the relationship between two matrices: $K\times S$ matrix $X$ and $M\times S$ matrix $Y$.
\end{itemize}

Empirically, CCA is widely applied in testing and inferring relationships between data sets: given two data sets $A$ and $B$, one can test for independence by examining canonical correlations between the row spaces of their respective matrices. If independence is rejected, CCA can then be used to identify the most interdependent components, or linear combinations, of $A$ and $B$. Additionally, applying CCA to spaces derived from the same data sets but through slightly varied procedures can uncover structural and temporal properties within the data.


Examples include: in genomics, finding shared patterns across assays from the same individuals; in neuroscience, linking brain activity with behavioral data; in ecology, correlating species characteristics with habitats; in econometrics, selecting common factors and performing cointegration analysis; and in finance, optimizing portfolio allocations. We will not cover all possible applications but provide references as starting points: \citet{gittins1985canonical,johansen_book}, \citet[Chapter 23]{simon1998assessing}, \citet{sherry2005conducting,witten2009extensions,breitung2013canonical,andreou2019inference,wang2020finding,zhuang2020technical, choi2021canonical,franchi2023estimating,firoozye2023canonical}.

\subsection{Three points of view on CCA} CCA can be approached from the probabilistic, statistical, and geometrical perspectives. The probabilistic (or \emph{population}) framework deals with two families of random variables and measures their dependence, while the statistics (or \emph{sample}) framework assumes that instead of observing the actual distributions (i.e., knowing the mean, variance, and so on) one observes realizations or samples from the distributions. This reflects real-world data scenarios, such as when researchers do not know the true distribution of stock returns, but have access to daily observations. Finally, the geometrical framework unifies both the probabilistic and statistical settings by interpreting random variables or their samples as vectors in an appropriately defined vector space.

Before introducing general definitions of CCA, it is helpful to first examine a one-dimensional case, where everything reduces to well-known notions of (sample) correlations and angles.


\begin{itemize}

\item {\bf Probability.}  Suppose that we are given two mean $0$ random variables $\mathsf x$ and $\mathsf y$. Linear dependence between $\mathsf x$ and $\mathsf y$ is measured by the \emph{correlation coefficient} $\rho$, where

    \begin{equation} \label{eq_cor_coefficient}
      \rho^2 = \frac{(\E \mathsf x \mathsf y)^2}{\E \mathsf x^2 \E \mathsf y^2}.
    \end{equation}
    We have $0\le \rho^2 \le 1$ and if $\rho^2$ is close to $1$, then $\mathsf x$ and $\mathsf y$ are highly dependent. In particular, $\rho^2=1$ implies that $\mathsf x = c\, \mathsf y$ for a deterministic constant $c$.

    \smallskip

\item {\bf Statistics.} Suppose that we are given two vectors of data $(x_1,x_2,\dots,x_S)$ and $(y_1,y_2,\dots,y_S)$ (e.g., samples of $(\mathsf x,\mathsf y)$ from the previous setting). Linear dependence between $(x_1,x_2,\dots,x_S)$ and $(y_1,y_2,\dots,y_S)$ is measured by the \emph{sample correlation coefficient} $\hat{\rho}$, where\footnote{For simplicity we assume here that the data is coming from a mean $0$ process. Otherwise, one first needs to subtract the means from the data.}
    \begin{equation} \label{eq_sample_cor_coefficient}
     \hat \rho^2= \frac{ \left( \sum\limits_{i=1}^S x_i y_i \right)^2}{\sum\limits_{i=1}^S x_i^2 \sum\limits_{i=1}^S y_i^2}.
    \end{equation}
    As before, $0\le \hat\rho^2 \le 1$ and if $\hat \rho^2$ is close to $1$, then two vectors are highly dependent. In particular, $\hat \rho^2=1$ implies that $x_i = c\, y_i$ for all $i=1,2,\dots,S$ for a constant $c$.

\item {\bf Geometry.} Suppose that we are given two lines passing through the origin in a linear vector space equipped with a scalar product $\la \cdot, \cdot \ra$. How can we describe their relative position? A numeric characteristic of this relationship is given by the \emph{angle} $\phi$ between these lines.\footnote{See Problem set $0$ in Chapter \ref{Section_exercises} for an explanation that this is the only characteristic.}. Suppose that $\u$ and $\v$ are vectors pointing in the direction of these two lines. Then $\phi$ can be computed through
    \begin{equation}\label{eq_angle}
     \cos^2\phi = \frac{\la \u,\v\ra^2}{\la \u,\u\ra \la \v,\v\ra}.
    \end{equation}
    Small $\phi$ or close to $1$ values of $\cos^2\phi$ mean that the lines almost coincide. In particular, $\cos^2\phi=1$ implies that the lines are the same.

\end{itemize}

\smallskip

Depending on the reader's preference, one may choose to work with one of the above frameworks. Yet, it is helpful to remember that \eqref{eq_cor_coefficient}, \eqref{eq_sample_cor_coefficient}, and \eqref{eq_angle} are three faces of the same mathematical object: both \eqref{eq_cor_coefficient} and \eqref{eq_sample_cor_coefficient} can be turned into the form \eqref{eq_angle} by properly specifying the vector space and the scalar product.

Canonical correlations provide a multi-dimensional extension of \eqref{eq_cor_coefficient}, \eqref{eq_sample_cor_coefficient}, and \eqref{eq_angle}. Instead of two random variables we consider two groups of random variables; rather than using two vectors of data, we work with two rectangular matrices of data; and instead of lines we examine linear subspaces.

The concepts behind CCA were originally introduced in \citet{jordan1875essai,harold1936relations} and remain widely used today. While we will cover all the key definitions, for additional reading, we recommend statistics textbooks such as \citet{thompson1984canonical,gittins1985canonical,anderson1958introduction}, and \citet{muirhead2009aspects}.

\subsection{Definitions: canonical correlations and variables} \label{Section_CCA_def}

In this subsection we provide formal definitions of canonical correlations and canonical variables and introduce different algorithms, which can be used to calculate them. We start with geometric framework via Theorem \ref{Theorem_canonical_bases} and then explain how to specialize it to probabilistic and statistics frameworks. We then present two alternative ways to calculate canonical correlations and canonical variables, which are commonly used perspectives in data science when thinking about CCA. The first approach can be viewed as a form of dimension reduction, where we iteratively capture as much of the correlation structure as possible by working with linear combinations of the data. It is also used to construct the bases of Theorem \ref{Theorem_canonical_bases}. The second approach relates CCA to the eigenvalues and eigenvectors of specific matrices. The proof of Theorem \ref{Theorem_canonical_bases} as well as some additional results are presented in Section \ref{Section_Th_CCA_proof}.

\subsubsection{Two bases}

\begin{theorem}\label{Theorem_canonical_bases}
Let $\W$ be a linear space\footnote{Throughout this text we work with real vector spaces, such as $\mathbb{R}^n$, though the theory also extends to complex spaces like $\mathbb{C}^n$, where some aspects are, in fact, simpler.}  with a scalar product $\la \cdot,\cdot\ra$. Suppose that $K\le M$ and let $\U$ and $\V$ be $K$-dimensional and $M$-dimensional subspaces of $\W$, respectively. Then there exist two orthonormal bases: vectors $\u_1, \u_2, \dots, \u_{K}$ span $\U$ and vectors $\v_1,\dots, \v_{M}$ span $\V$ -- such that for all meaningful indices $i$ and $j$ we have
\begin{equation}
\label{eq_scalar_products_table}
 \la \u_i, \u_j\ra = \delta_{i=j}, \qquad \la \v_i, \v_j\ra = \delta_{i=j},\qquad \la \u_i, \v_j\ra = c_i \delta_{i=j},
\end{equation}
where  $1\ge c_1\ge c_2 \ge \dots \ge c_{K}\ge 0$.
\end{theorem}
The numbers $c_1\ge c_2 \ge \dots \ge c_{K}$ are called \emph{canonical correlations} between subspaces $\U$ and $\V$; they are also cosines of the canonical angles between the subspaces. The vectors $\u_i$, $1\le i \le {K}$, and $\v_j$, $1\le j\le {M}$, are called \emph{canonical variables}; they split into ${K}$ pairs $(\u_i,\v_i)$ and $M-K$ singletons $\v_j$, $j\ge K$.

Theorem \ref{Theorem_canonical_bases} constructs an orthonormal basis in the ambient space, which encompasses $\U$ and $\V$. It can be treated as an analogue of the diagonalization procedure for symmetric matrices (i.e., the theorem which says that any symmetric matrix can be made diagonal by a choice of an appropriate basis of the space). As will be seen in Corollary \ref{Corollary_invariants}, canonical correlations $c_i$ are invariants of spaces $\U$ and $\V$ under orthogonal transformations of the space: they are present, no matter what orthogonal basis one chooses. The point of view of $c_i$ being invariants is explored, e.g., in \cite{halmos1969two} and \cite{bottcher2010gentle}.

\bigskip

Let us adapt Theorem \ref{Theorem_canonical_bases} to probability and statistics frameworks.

\textbf{Theorem \ref{Theorem_canonical_bases} in the probability framework.}
Suppose we are given a collection of mean $0$ random variables $\{\mathsf x^i\}_{i=1}^K$ and another collection $\{\mathsf y^j\}_{j=1}^M$. Then we take as $\U$ the space of all linear combinations $\alpha_1 \mathsf x^1+\alpha_2 \mathsf x^2+\dots+\alpha_K\mathsf x^K$ with \emph{deterministic} real coefficients $\alpha_i$. Similarly, $\V$ is the space of all linear combinations of $\mathsf y^j$ with deterministic real coefficients. The scalar product is given by expectation: $\la \mathsf u,\mathsf v\ra=\E \mathsf u \mathsf v$ for mean $0$ random variables $\mathsf u$ and $\mathsf v$. In this situation, the canonical variables are also random variables: $\u_i$ are linear combinations of $\mathsf x^1,\dots,\mathsf x^K$ and $\v_j$ are linear combinations of $\mathsf y^1,\dots,\mathsf y^M$ with deterministic coefficients. On the other hand, the canonical correlations $c_1,\dots,c_K$ are deterministic --- they describe the correlation structure between families $\{\mathsf x^i\}_{i=1}^K$ and $\{\mathsf y^j\}_{j=1}^M$. In particular, if $K=M=1$, then we are back to \eqref{eq_cor_coefficient} with $\rho=c_1$:
$$
 \u_1=\pm\frac{\mathsf x^1}{\sqrt{\E (\mathsf x^1)^2}}, \qquad \v_1=\pm\frac{\mathsf y^1}{\sqrt{\E (\mathsf y^1)^2}}, \qquad c_1= \frac{\left|\E \mathsf x^1\mathsf y^1\right|}{\sqrt{\E (\mathsf x^1)^2 \E (\mathsf y^1)^2}},
$$\
where the signs are chosen so that $\langle \u_1,\v_1\rangle \ge 0$.

\smallskip

For general $M\ge K$ Theorem \ref{Theorem_canonical_bases} gives the simplest possible form to which a correlation structure between two families of random variables  $\{\mathsf x^i\}_{i=1}^K$ and $\{\mathsf y^j\}_{j=1}^M$ can be reduced by linear transformations.
A useful analogy can be drawn with the theorem for Gaussian vectors, which states that any zero-mean Gaussian vector can be obtained by applying a linear transformation to a Gaussian vector with identity covariance matrix. Similarly, Theorem \ref{Theorem_canonical_bases} asserts that the $(K+M)\times (K+M)$ covariance matrix of the vector $(\mathsf x^1,\dots, \mathsf x^K, \mathsf y^1, \dots, \mathsf y^M)$ can be reduced by two linear transformations (one applied to $\{\mathsf x^i\}_{i=1}^K$ and another to $\{\mathsf y^j\}_{j=1}^M$) to a simple block form:

{\small
$$\left(
 \begin{array}{ccc:ccccc}
   1 & 0 & 0 & c_1 & 0 &0 &\dots & 0\\
   0 & \ddots  & 0 & 0 &  \ddots &0  & & 0\\
   0 & 0&  1 & 0 & 0  & c_K &\dots & 0 \\
   \hdashline
    c_1 & 0 & 0 &  1 & 0 & & & \\
    0 & \ddots & 0  & 0 & \ddots & & & \\
    0 & 0 & c_K & & &  &   & \\
     &   &   &    & & & \ddots  & 0\\
     0 &  \dots & 0  &    & & & 0 & 1
 \end{array}\right),
$$}
where the block sizes are $K\times K$, $K\times M$ in the top and $M\times K$, $M\times M$ in the bottom.

\smallskip

\textbf{Theorem \ref{Theorem_canonical_bases} in the statistics framework.}
Next, suppose we are given two rectangular matrices of data: $K\times S$ matrix $X=[X_{ij}]$ and $M\times S$ matrix $Y=[Y_{ij}]$. Then we take as $\U\subset \mathbb R^S$ the space of all linear combinations of $K$ rows of $X$ and $\V\subset\mathbb R^S$ is the space of all linear combinations of $M$ rows of $Y$. The scalar product is the standard scalar product in $\mathbb R^S$: $\la (x_1,\dots,x_S), (y_1,\dots,y_S)\ra = \sum_{i=1}^S x_i y_i$. In this situation the sample\footnote{In this context we add the term ``sample'' to indicate the statistical, or sampled, nature of the data.} canonical variables are $S$-dimensional vectors: $\u_i$ are linear combinations of rows of $X$ and $\v_j$ are linear combinations of rows of $Y$. The sample canonical correlations $c_1,\dots,c_K$ are numbers describing the cross-correlation structure between the rows of $X$ and $Y$. In particular, if $K=M=1$, then we are back to \eqref{eq_sample_cor_coefficient} with $\hat \rho=c_1$:
$$
 \u_1=\pm \frac{ (X_{11},X_{12},\dots,X_{1S})}{\sqrt{\sum\limits_{i=1}^S X_{1i}^2}},\qquad \v_1=\pm \frac{(Y_{11},Y_{12},\dots,Y_{1S})}{\sqrt{\sum\limits_{i=1}^S Y_{1i}^2}},\qquad c_1=\frac{\left|\sum\limits_{i=1}^S X_{1i}Y_{1i} \right| }{\sqrt{\sum\limits_{i=1}^S X_{1i}^2\sum\limits_{i=1}^S Y_{1i}^2}}.
$$


\subsubsection{Maximization problem}\label{subsection_maximiz_problem}

The first approach for finding the canonical correlations and variables of Theorem \ref{Theorem_canonical_bases} is to consider a function $f(\u,\v)=\la \u, \v\ra$, in which $\u$ varies over all unit vectors in $\U$ and $\v$ varies over all unit vectors in $\V$. This function represents (sample) correlation\footnote{Since correlation between two variables is unaffected by rescaling of the variables, we can restrict the space to be the unit sphere.} between $\u$ and $\v$.  Then $c_1$ is the maximum of $f$, which is achieved at $(\u_1,\v_1)$. After $(\u_1,\v_1,c_1)$ is found, we consider restriction of the same function $f(\u,\v)$ on the vectors $\u \in \U$ orthogonal to $\u_1$ and vectors $\v\in \V$ orthogonal to $\v_1$. Maximizing the restricted function, we get $c_2$, which is achieved at $(\u_2,\v_2)$. Next, we restrict $f$ to $\u\in \U$ orthogonal to both $\u_1$ and $\u_2$ and to $\v\in \V$ orthogonal to both $\v_1$ and $\v_2$, etc. Formally, step $i$ of the sequential maximization problem can be written as
\begin{equation}\label{eq_CCA_maximization}
\begin{split}
   &(\arg)\max\limits_{\u\in\U,\,\v\in\V} \la \u, \v\ra \\
    \text{such that }&\begin{cases}
    \la\u,\u\ra=\la\v,\v\ra=1,\\
    \la\u,\u_1\ra=\ldots=\la\u,\u_{i-1}\ra=0,\\
    \la\v,\v_1\ra=\ldots=\la\v,\v_{i-1}\ra=0.
    \end{cases}
\end{split}
 \end{equation}
Once all $(\u_i,\v_i,c_i)$, $1\le i \le {K}$, are identified, the remaining vectors $\v_j$, $j>{K}$, are an arbitrary orthogonal basis in the part of $\V$ orthogonal to $\v_1,\dots,\v_{K}$. This approach reduces the task of finding canonical correlations and variables to a series of iterative maximization problems, which can be efficiently solved using numerous numerical algorithms and also admit various generalizations.

If we wish to solve maximization problems \eqref{eq_CCA_maximization} simultaneously for all $1\le i \le K$, then we can view $(\u_i,\v_i,c_i)$, $1\le i \le {K}$, as critical points and corresponding values of $f$. Given the constraints $\la\u,\u\ra=\la\v,\v\ra=1$, these critical points can be found through the Lagrangian function:
$$
 \mathcal L(\u,\v)=\la \u, \v\ra+ a \la\u,\u\ra + b \la\v,\v\ra, \qquad \u\in\U, \quad \v\in \V.
$$
where $a$ and $b$ are Lagrange multipliers. Expanding $\u$ and $\v$ in arbitrary bases of $\U$ and $\V$, respectively (in the probabilistic setting it is natural to expand in $\{\mathsf x^i\}_{i=1}^K$ and  $\{\mathsf y^j\}_{j=1}^M$; while in the statistical setting it is natural to expand in rows of $X$ and $Y$), we differentiate $\mathcal{L}(\u, \v)$ with respect to all coefficients of the expansions and set the partial derivatives to zero. This yields a system of linear equations, whose solutions provide the pairs $(\pm \u_i, \pm \v_i)$.

\subsubsection{Eigenvalues and eigenvectors} \label{Subsection_ev}

Here is an alternative approach for finding the canonical correlations and variables of Theorem \ref{Theorem_canonical_bases}. For a subspace $\mathbf Q\subset \W$ let us denote through $P_{\mathbf Q}$ the orthogonal projector onto $\mathbf Q$.

 \begin{proposition} \label{Proposition_CCA_as_eigenvectors}
  In the notations of Theorem \ref{Theorem_canonical_bases}, the squares $c_i^2$, $1\le i \le {K}$, are eigenvalues of $P_{\U} P_{\V}$ and $\u_i$ are corresponding eigenvectors. Simultaneously, $c_j^2$ are eigenvalues of $P_{\V} P_{\U}$ and $\v_j$ are corresponding eigenvectors.
 \end{proposition}

In order to use Proposition \ref{Proposition_CCA_as_eigenvectors} as a numerical algorithm, one needs to calculate the matrices of the projectors. If we use the statistics point of view where $\U$ and $\V$ are rows of rectangular matrices $X$ and $Y$, then
\begin{equation}
\label{eq_x2}
 P_{\U} P_{\V}= X^\T (X X^\T)^{-1} X Y^\T (Y Y^\T)^{-1} Y, \qquad  P_{\V} P_{\U}=Y^\T (Y Y^\T)^{-1} Y X^\T (X X^\T)^{-1} X,
\end{equation}
where $(\cdot)^{\T}$ here and below denotes matrix transposition. Note that if the rows of $X$ and $Y$ are orthogonal, then  $X X^\T$ and $Y Y^\T$ are identity matrices, whose inversion is straightforward. This feature is sometimes used for developing more efficient numerical methods.

\begin{remark}
\label{Remark_projectors_abuse}
 With a minor abuse of notations, we often want to identify the spaces $\U$ and $\V$ with the linear span of the rows of matrices denoted by the same letters $\U$ and $\V$. In other words, we do not want to introduce another notation for $X$ and $Y$, but instead write directly $X=\U$ and $Y=\V$, so that \eqref{eq_x2} becomes:
\begin{equation}
\label{eq_x3}
 P_{\U} P_{\V}= \U^\T (\U \U^\T)^{-1} \U \V^\T (\V \V^\T)^{-1} \V, \qquad  P_{\V} P_{\U}=\V^\T (\V \V^\T)^{-1} \V \U^\T (\U \U^\T)^{-1} \U.
\end{equation}
\end{remark}

The product of projectors $P_{\U} P_{\V}$ in \eqref{eq_x2} is a $\dim(\W)\times \dim(\W)$ dimensional matrix, which can be hard to operate with if $\dim(\W)$ is large. In particular, if we are in the probability setting with random variables $\{\mathsf x^i\}_{i=1}^K$ and $\{\mathsf y^j\}_{j=1}^M$, it is natural to take as $\W$ the space of all square-integrable random variables (on some probability space), which makes $\W$ infinite-dimensional and, therefore, directly evaluating $P_{\U} P_{\V}$ becomes computationally challenging. However, as formally proved later in Proposition \ref{prop_CCA_matrix_form}, we can simplify the algorithm and deal with $K\times K$ and $M\times M$ matrices instead:

$\bullet$ In the statistics framework the squared sample canonical correlations can be equivalently found as eigenvalues of either of the matrices
\begin{equation}
\label{eq_x16}
  (\U \U^\T)^{-1} \U \V^\T (\V \V^\T)^{-1} \V \U^\T, \qquad\text{or} \qquad  \V \U^\T (\U \U^\T)^{-1} \U \V^\T (\V \V^\T)^{-1},
\end{equation}
where we note that the first matrix is $K\times K$ and the second one is $M\times M$. If $\ba$ is an eigenvector of the first matrix and $\bb$ is an eigenvector of the second matrix with the same eigenvalue $c$, then, up to normalization of vectors, the triplets $(\U^\T\ba, \V^\T\bb, c)$ represent sample canonical variables and correlations between $\U$ and $\V$. Sometimes, $\ba$ and $\bb$ are referred to as \emph{canonical vectors}, which should not be confused with canonical variables $\U^\T\ba$ and  $\V^\T\bb$.

\medskip

$\bullet$ Similarly, in the probability framework the squared (populational) canonical correlations can be found as eigenvalues of either of the following matrices, in which we used the vector notations $\u=(\mathsf x^1,\mathsf x^2,\dots,\mathsf x^K)^{\T}$, $\v=(\mathsf y^1,\mathsf y^2,\dots,\mathsf y^M)^{\T}$,
\begin{equation}
\label{eq_x17}
 (\E \u \u^\T)^{-1} \E \u \v^\T (\E \v\v^\T)^{-1} \E \v \u^{\T} \qquad\text{or} \qquad  \E \v \u^{\T} (\E \u \u^\T)^{-1} \E \u \v^\T (\E \v\v^\T)^{-1}.
\end{equation}
If $\ba$ is an eigenvector of the first matrix and $\bb$ is an eigenvector of the second matrix with the same eigenvalue $c$, then, up to normalization of vectors, the triplets $(\u^\T\ba, \v^\T\bb, c)$ represent canonical variables and correlations between the families $\{\mathsf x^i\}_{i=1}^K$ and $\{\mathsf y^j\}_{j=1}^M$.

\subsection{Proof of Theorem \ref{Theorem_canonical_bases} and corollaries} \label{Section_Th_CCA_proof}
This section presents proofs of all the previously mentioned facts, along with some additional results. We start by showing that the iterative maximization algorithm \eqref{eq_CCA_maximization} produces the desired bases $\u_1, \u_2, \dots, \u_{K}$ and $\v_1,\dots, \v_{M}$.

Let $U^1=\{\u\in\U\mid \la \u,\u\ra =1\}$, $V^1=\{\v\in\V\mid \la \v,\v\ra=1\}$ be two unit spheres in $K$ and $M$ dimensional spaces, respectively. Consider the function $f(\u,\v)=\la \u,\v\ra$ restricted to these spheres. If $f$ is identical zero, then we are done: the canonical correlations $c_i$ are $0$ and we can choose an arbitrary basis of $\U$ as $\u_1,\dots,\u_K$ and an arbitrary basis of $\V$ ad $\v_1,\dots,\v_M$.

If $f(\u,\v)$ is not identical zero, then $f(-\u,\v)=-f(\u,\v)$ implies that $f$ takes positive values at some points. Since $f$ is a continuous function and spheres are compact, $f$ achieves its maximum at a point. Let $0<c_1\le 1$ denote the maximal value of $f$, and let $(\u_1,\v_1)$ denote the point where it is achieved. By definition $\la \u_1,\u_1\ra=\la \v_1,\v_1\ra=1$ and $\la \u_1,\v_1\ra=c_1$, fitting into \eqref{eq_scalar_products_table}.

Further, set $U^2=\{\u\in U^1\mid \la \u,\u_1\ra =0\}$, $V^2=\{v\in V^1\mid \la \v,\v_1\ra=0\}$. These are $K-1$ and $M-1$ dimensional spheres, respectively, with an important property:
\begin{lemma} \label{Lemma_orthogonality}
 All vectors in $U^2$ are orthogonal to $\v_1$. All vectors in $V^2$ are orthogonal to $\u_1$.
\end{lemma}
\begin{proof}
 We argue by contradiction. Suppose that for $\u\in U^2$ we have $\la \u,\v_1\ra=d>0$. For any $0<\alpha<1$ consider a unit vector $\u \alpha+ \sqrt{1-\alpha^2} \u_1\in U^1$. For small positive $\alpha$ we have
 $$
  \la \u \alpha+ \sqrt{1-\alpha^2} \u_1, \v_1\ra = \alpha d +  \sqrt{1-\alpha^2}\, c_1= c_1 + \alpha d - \frac{\alpha^2}{2} c_1 + o(\alpha^2)> c_1,
 $$
 contradicting that $c_1$ is a maximum of $f$. The argument for $V^2$ is the same.
\end{proof}

Next, we consider the restriction of $f$ onto $U^2\times V^2$. If it is identical zero, then $c_2=c_3=\dots=0$, and we can complement $\u_1$ to an orthonormal basis of $\U$ and complement $\v_1$ to an orthonormal basis of $\V$, getting the desired \eqref{eq_scalar_products_table}.

Otherwise, let $c_2$ denote the maximum of $f$ on $U^2\times V^2$, achieved at $(\u_2,\v_2)$. By definitions and Lemma \ref{Lemma_orthogonality}, we have:
$$
 \la \u_2,\u_2\ra =\la \v_2,\v_2\ra=1,\quad \la \u_2,\u_1\ra=\la \u_2,\v_1\ra=\la \v_2,\u_1\ra= \la \v_2,\v_1\ra =0,\quad \la \u_2,\v_2\ra = c_2,
$$
fitting into \eqref{eq_scalar_products_table}. We can continue in the same way by defining $K-2$ and $M-2$ dimensional spheres $
 U^3=\{\u\in U^2\mid \la \u,\u_2\ra =0\}$, $V^3=\{v\in V^2\mid \la \v,\v_2\ra=0\}$. Then we use the same argument as in Lemma \ref{Lemma_orthogonality} to show that $U^3$ is orthogonal to $\v_1$ and $\v_2$ and $V^3$ is orthogonal to $\u_1$ and $\u_2$, maximize $f$ over $U^3\times V^3$, etc.

 Repeating the argument $K$ times, we exhaust $\U$ and construct its orthonormal basis $\u_1,\dots,\u_K$ and corresponding partial basis $\v_1,\dots,\v_K$ of $\V$. Yet again by a version of Lemma \ref{Lemma_orthogonality}, the complement of $\v_1,\dots,\v_K$ in $\V$ is orthogonal to all of $\U$. Hence, we can extend $\v_1,\dots,\v_K$ to an orthonormal basis of $\V$ in an arbitrary way, and the scalar products \eqref{eq_scalar_products_table} will hold. This finishes the proof of Theorem \ref{Theorem_canonical_bases}.

 \bigskip

 Let us now explain the validity of the second algorithm outlined in Section \ref{Subsection_ev}.

 \begin{proof}[Proof of Proposition \ref{Proposition_CCA_as_eigenvectors}] The operator $P_\U P_\V$ is a product of rank $K$ and rank $M$ operators, hence, its rank is at most $\min(K,M)=K$. Therefore, it has at most $K$ eigenvectors with non-zero eigenvalues. Let us show that these eigenvectors are $\u_1,\dots,\u_K$ and corresponding eigenvalues are $c_i^2$. Indeed, using the table \eqref{eq_scalar_products_table}, we have for any $1\le i\le K$:
 $$
  P_\V \u_i = c_i \v_i, \qquad P_\U P_\V \u_i = P_\U (c_i \v_i) = c_i^2 \u_i.
 $$
 Similarly, $\v_i$ are eigenvectors of $P_\V P_\U$.
 \end{proof}

\begin{corollary} \label{Corollary_invariants}
 Canonical correlations $c_i$ are invariants of subspaces $\U\subset \W$ and $\V\subset \W$ under orthogonal transformations of the ambient space $\W$
\end{corollary}
\begin{proof}
 Let us transform $\W$ by acting with an orthogonal matrix $O$ on all vectors. Then in the notations of Remark \ref{Remark_projectors_abuse}, $\U$ gets transformed into $\U O^\T$ and $\V$ gets transformed into $\V O^\T$. Thus, the projector $P_\U P_\V$ is transformed into $O P_{\U} O^\T O P_{\V} O^\T$. The latter matrix is the same as $O P_{\U} P_{\V} O^\T$ and its eigenvalues coincide with those of $P_\U P_\V$. Hence, by Proposition \ref{Proposition_CCA_as_eigenvectors}, the canonical correlations are preserved.
\end{proof}

\begin{proposition}\label{prop_CCA_matrix_form}
 Assume that all the matrices inverted in \eqref{eq_x16}, \eqref{eq_x17} are non-denerate. Then canonical correlations can be equivalently computed as eigenvalues of matrices \eqref{eq_x16} and \eqref{eq_x17} in the statistics and probability frameworks, respectively, and corresponding variables can be computed by the procedures outlined immediately after these formulas.
\end{proposition}
\begin{proof}
 For the statistics setting we use Proposition \ref{Proposition_CCA_as_eigenvectors} and (as in Remark \ref{Remark_projectors_abuse}) rewrite the product of projectors as
 $$
  P_{\U} P_{\V}= \U^\T (\U \U^\T)^{-1} \U \V^\T (\V \V^\T)^{-1} \V.
 $$
 Directly from the above formula we see that $\u_i$ is an eigenvector of the above matrix with eigenvalue $c_i^2>0$, if and only if $\u_i = \U^{\T} \ba_i$ and $\ba_i$ is an eigenvector of $(\U \U^\T)^{-1} \U \V^\T (\V \V^\T)^{-1} \V \U^\T$ with the same eigenvalue $c_i^2$. $\ba_i$ can be also expressed as $(c_i)^{-2}(\U \U^\T)^{-1} \U \V^\T (\V \V^\T)^{-1} \V\u_i$.


 For the probability setting, let us write down the projectors used in Proposition \ref{Proposition_CCA_as_eigenvectors}. The projector $P_\U$ on the space spanned by $\xi_1,\dots,\xi_K$ acts on a random variable $\zeta$ via
 $$
   P_\U \zeta = \u^\T \cdot [\E \u \u^\T]^{-1}  \cdot \E \u \zeta, \qquad \u=(\xi_1,\dots,\xi_K)^\T.
 $$
 In this formula random column-vector $\u^T$ of size $1\times K$ is being multiplied by a deterministic vector of weights, obtained by multiplying the $K\times K$ inverse matrix by $K\times 1$ vector $\E \u \zeta$.
 In order to verify the formula, one directly checks that each $\xi_l$ is mapped to itself and each random variable $\zeta$ with $\E \xi_l \zeta=0$ for all $1\le \l \le K$ is mapped to zero.
 Also the projector $P_\V$ on the space spanned by $\eta_1,\dots, \eta_M$ acts on a random variable $\zeta$ via
 $$
   P_\V \zeta = \v^\T \cdot [\E \v \v^\T]^{-1}\cdot \E \v \zeta, \qquad \v=(\eta_1,\eta_2,\dots,\eta_M)^\T.
 $$
 Hence,
 $$
   P_\U P_\V \zeta =   \u^\T \cdot [\E \u \u^\T]^{-1}  \cdot \E \u \v^\T \cdot [\E \v \v^\T]^{-1}\cdot \E \v \zeta.
 $$
 Directly from the above formula, we see that $\u_i$ is an eigenvector of the above matrix with eigenvalue $c_i^2>0$ if and only if $\u_i= \u^\T \ba$ and $\ba=(\alpha_1,\dots,\alpha_K)^\T$ is an eigenvector of the matrix $(\E \u \u^\T)^{-1} \E \u \v^\T (\E \v\v^\T)^{-1} \E \v \u^{\T} $ with the same eigenvalue $c_i^2$.
\end{proof}

 We end this section with a discussion of uniqueness of the bases satisfying \eqref{eq_scalar_products_table} in Theorem~\ref{Theorem_canonical_bases}. If all $c_i$ are distinct, then Proposition \ref{Proposition_CCA_as_eigenvectors} implies that the only freedom in choosing $\u_1,\dots,\u_K$ and $\v_1,\dots,\v_K$ is in simultaneous multiplication of $\u_i$ and $\v_i$ by $-1$. For $\v_{K+1},\dots,\v_M$ we always have more options: any $M-K$ orthonormal vectors complementing $\v_1,\dots,\v_K$ to a basis of $\V$ would work.

 If some $c_i$ coincide, then the situation is more delicate: eigenvectors in Proposition \ref{Proposition_CCA_as_eigenvectors} are no longer unique, only the linear spaces spanned by eigenvectors with the same eigenvalue are uniquely determined. In particular, imagine that $c_1$ has multiplicity two, so that for four unit vectors $\u_1,\u_2\in \U$ ,$\v_1,\v_2\in \V$ we have:
 $$
  \la \u_1,\v_1\ra = \la \u_2,\v_2\ra = c_1, \qquad  \la \u_1,\u_2\ra = \la \u_1,\v_2\ra = \la \u_2,\v_1\ra = \la \v_1,\v_2\ra = 0.
 $$
 Then for any $-1\le \alpha\le 1$, we also have
 $$
  \la \alpha \u_1 + \sqrt{1-\alpha^2} \u_2, \alpha \v_1 + \sqrt{1-\alpha^2} \v_2 \ra 
  = \alpha^2 c_1 +(1-\alpha^2)c_1=c_1.
 $$
Hence the vectors $(\alpha \u_1 + \sqrt{1-\alpha^2} \u_2, \alpha \v_1 + \sqrt{1-\alpha^2} \v_2)$ also deliver the maximum of the function $f(\u,\v)$ and (by the algorithm in subsection \ref{subsection_maximiz_problem}) can be included as the first vectors in a system satisfying \eqref{eq_scalar_products_table}.

Two extreme examples are when the subspaces $\U$ and $\V$ are either orthogonal or coincide. In both situations all $c_i$ are equal.
\begin{itemize}
 \item If subspaces $\U$ and $\V$ are orthogonal, then all $c_i$ are equal to $0$, $\u_1,\dots,\u_K$ is an arbitrary orthonormal basis of $\U$ and $\v_1,\dots,\v_M$ is an arbitrary orthogonal basis of $\V$.
 \item If $K=M$ and $\U=\V$, then all $c_i$ are equal to $1$, $\u_1,\dots,\u_K$ is an arbitrary orthonormal basis of $\U$ and $\v_i=\u_i$, $1\le i \le K$.
\end{itemize}

\subsection{Questions and applications}
There are many open questions that continue to drive research and applications of CCA. 
On the theoretical side, CCA-related questions in probability and statistics often center on the distributional properties -- both exact and asymptotic -- of canonical correlations and variables. In this review we focus on some of these topics. Specifically, we highlight that:



\begin{enumerate}
\item The histograms of canonical correlations frequently exhibit asymptotic non-randomness, converging to explicit deterministic \emph{limit shapes}.
\item The asymptotic distributions of individual canonical correlations, such as the largest one, are \emph{universal}, meaning that these limiting laws depend only on the high-level specification of the models, but not on the particular details.
\end{enumerate}

Returning to empirical applications, these two features allow us to test mathematical theorems about CCA on real-world data sets, yielding several insights. First, we often observe a remarkable alignment between theoretical predictions and empirical results in various data sets, highlighting the applicability of our modeling approach to real-world data. Second, we can leverage these theorems to derive structural insights from the data sets.

As a preview, we present three plots generated from financial data sets. Figure \ref{Fig_SP_Wachter_data} shows the histogram of sample canonical correlations for weekly data from S\&P 100 stocks over a ten-year period, correlating the logarithms of stock prices with their time increments. Figure \ref{Fig_cryptovar1} provides a similar plot but for daily data on 25 cryptocurrencies over two years. Figure \ref{Fig_two_stocks_hist} presents squared sample canonical correlations for weekly returns over ten years for two groups of 80 stocks: largest ``cyclical'' versus largest ``non-cyclical (defensive)'' stocks. In each case the empirical histograms are accompanied by (closely matching) theoretical curves, that represent the Wachter distribution. The relevance of the Wachter distribution to CCA and these data sets in particular is discussed in Chapters 3 and 5. The canonical correlations appearing to the right of the support of Wachter distribution (the flat region of the theoretical curves) can be interpreted as signals in the data, which we elaborate on in Chapter 4. For further details on these figures and the corresponding theoretical results we refer to \cite{BG1}, \cite{BG2}, and \cite{BG3}, respectively.

\begin{figure}[p]
\includegraphics[width=0.48\linewidth]{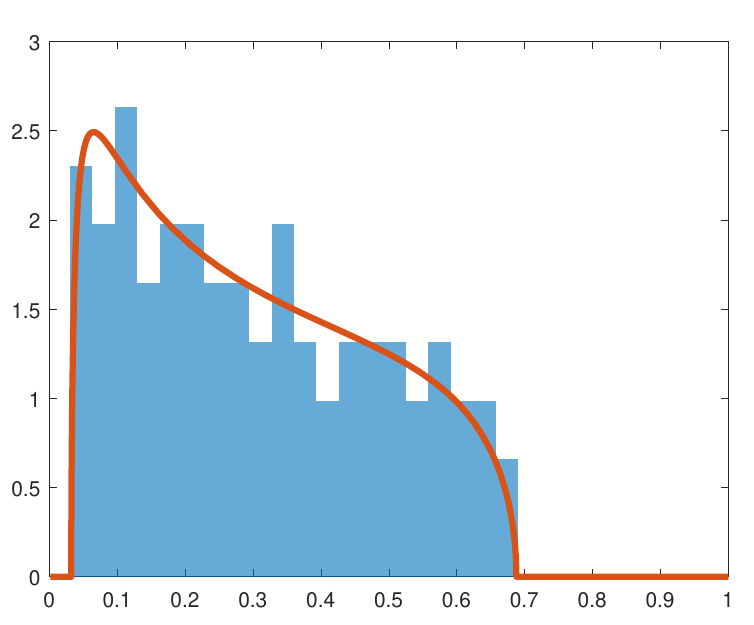}
\caption{S$\&$P 100 stocks log-prices vs time increments and Wachter.}
\label{Fig_SP_Wachter_data}
\end{figure}
\begin{figure}[p]
  \includegraphics[width=0.52\linewidth]{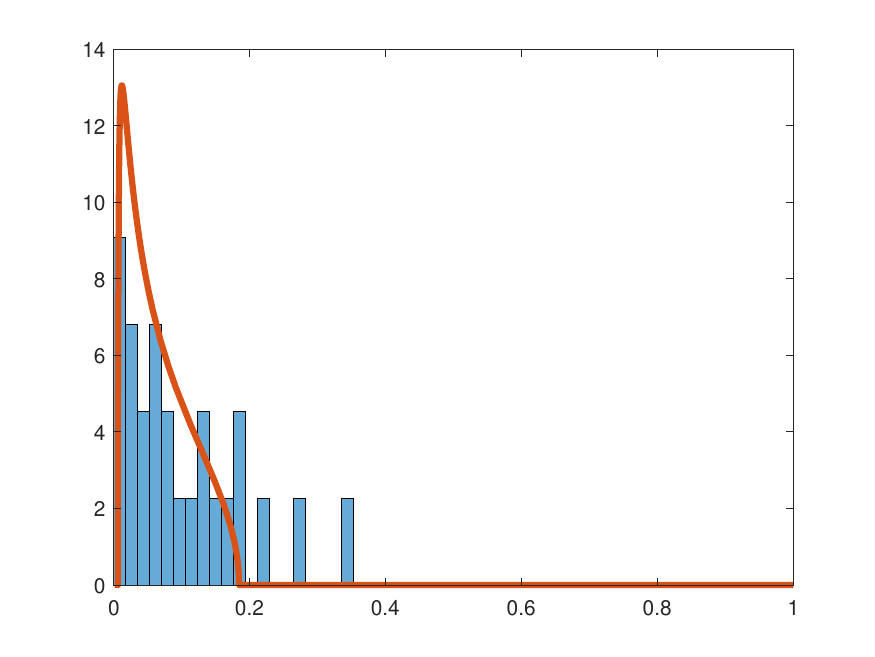}
 \caption{Cryptocurrency log-prices vs time increments and Wachter.}
   \label{Fig_cryptovar1}
\end{figure}
\begin{figure}[p]
  \centering
  \includegraphics[width=0.52\linewidth]{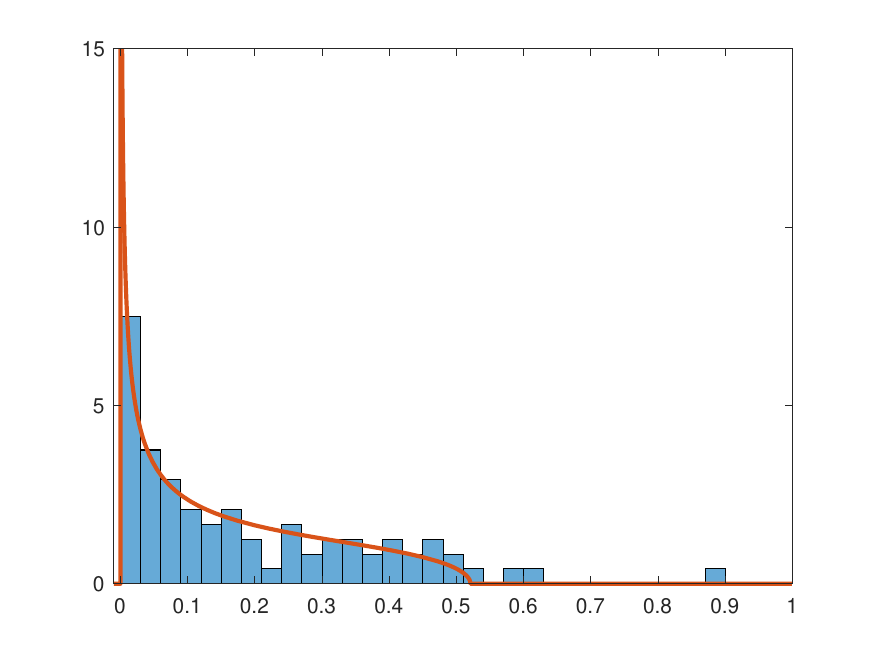}
  \caption{Cyclical vs non-cyclical stock returns and Wachter.}
  \label{Fig_two_stocks_hist}
\end{figure}

 \newpage

\section{IID setting}

\subsection{CCA and Maximum Likelihood Estimation} \label{Section_CCA_as estimation}
This chapter develops a connection between statistical and probability frameworks. We start from the population model with two random vectors: $\u=(u^1,\dots,u^K)^{\T}$ and $\v=(v^1,\dots,v^M)^\T$, $M\ge K$. We do not have direct access to these vectors and we do not know anything about them, except that they have mean $0$. Instead, we are observing $S$ i.i.d.\ samples of these vectors. Formally, this means that we have two matrices: $\U$ of size $K\times S$ and $\V$ of size $M\times S$. The combined $(K+M)\times S$ matrix $\W=\begin{pmatrix}\U\\ \V \end{pmatrix}$ has i.i.d.\ columns, with each column following the distribution of $\w=\begin{pmatrix}\u\\ \v \end{pmatrix}$.

\begin{task} How can the correlation structure between $\u$ and $\v$ be inferred from $\U$ and $\V$? \end{task}

Let us clarify the meaning of the words ``correlation structure''. There are three covariance matrices in play: $K\times K$ matrix of auto-covariances of $\u$,  $\Lambda_{uu}[i,j]=\E u^i u^j$, and $M\times M$ matrix of auto-covariances of $\v$, $\Lambda_{vv}[i,j]=\E v^i v^j$, and $K\times M$ matrix of cross-covariances $\Lambda_{uv}[i,j]=\E u^i v^j$; we also denote $\Lambda_{vu}=\Lambda_{uv}^\T$. Theorem \ref{Theorem_canonical_bases} implies that there exists $K\times K$ matrix $A$ and $M\times M$ matrix $B$, such that
\begin{equation} \label{eq_x1}
 \begin{pmatrix} A & 0_{K\times M}\\ 0_{M\times K} & B \end{pmatrix} \begin{pmatrix} \Lambda_{uu} & \Lambda_{uv}\\ \Lambda_{vu} & \Lambda_{vv}\end{pmatrix} \begin{pmatrix} A^\T & 0_{K\times M}\\ 0_{M\times K} & B^\T \end{pmatrix}= \begin{pmatrix} I_{K} & \mathrm{diag}(c_1,\dots,c_K)\\ \mathrm{diag}(c_1,\dots,c_K) & I_{M} \end{pmatrix},
\end{equation}
where $I_K$ and $I_M$ are identity matrices of $K\times K$ and $M\times M$ dimensions, respectively, and $\mathrm{diag}(c_1,\dots,c_K)$ is a rectangular matrix with $c_1,c_2,\dots,c_K$ on the main diagonal and $0$ everywhere else. Here $1\ge c_1\ge c_2\ge\dots\ge c_K\ge 0$ are canonical correlations between the spaces spanned by coordinates of $\u$ and coordinates of $\v$, $A \u$ is the random vector of $K$ canonical variables $(\u_1,\dots,\u_K)^{\T}$ and $B\v$ is the vector of $M$ canonical variables $(\v_1,\dots,\v_M)^{\T}$.

From \eqref{eq_x1} we conclude that knowing the covariance matrices $\Lambda_{uu}$, $\Lambda_{vv}$, $\Lambda_{uv}$ is equivalent to knowing the canonical correlations and variables between the coordinates of $\u$ and $\v$. The advantage of \eqref{eq_x1} is its clear separation of the internal correlation structures of $\u$ and $\v$, captured by $A$ and $B$, from the strength of the correlations between $\u$ and $\v$, represented $c_1,\dots,c_K$. Hence, our task can be reformulated as:

\begin{task} \label{Task_CCA_iid} How can we estimate the canonical correlations and variables between $\u$ and $\v$ from observing $\U$ and $\V$? \end{task}

There are various statistical methods for estimating parameters, with one of the most classical being the \emph{Maximum Likelihood Estimator} (MLE). The MLE procedure begins by assuming that the random variables of interest are sampled from a specific, predefined family of probability distributions (population model), that depend on our parameter of interest. Next, we compute the likelihood function -- a function of both the parameters and the observed samples -- that quantifies the probability, or probability density, of observing the given samples. The parameters that maximize this function, treating the observed samples as fixed, provide the desired estimator. The estimator is then treated as a function of the observed sample.

\begin{theorem} \label{Theorem_Gaussian_MLE} The Gaussian MLE for $\Lambda_{uu}$, $\Lambda_{vv}$, $\Lambda_{uv}$ are the sample covariances:
\begin{equation}
 \label{eq_estimation_with_sample}
 \hat \Lambda_{uu}=\frac{1}{S} \U \U^\T,\qquad  \hat \Lambda_{vv}=\frac{1}{S} \V \V^\T, \qquad  \hat \Lambda_{uv}=\frac{1}{S} \U \V^\T.
\end{equation}
\end{theorem}
\begin{proof}
 The adjective ``Gaussian'' in front of MLE means that we assume $\u$ and $\v$ to be Gaussian vectors and compute the likelihood function based on the Gaussian density. Denoting  $\Lambda=\begin{pmatrix} \Lambda_{uu} & \Lambda_{uv}\\ \Lambda_{vu} & \Lambda_{vv}\end{pmatrix}$, the likelihood function is expressed as
 \begin{equation}
 \label{eq_Gaussian_likelihood}
  L(\U,\V;\Lambda)=  \left[2^{K+M}\pi^{K+M} \det(\Lambda)\right]^{-S/2} \exp\left(-\frac{1}{2}\sum_{j=1}^S (\W^{j})^{\T} \Lambda^{-1} \W^{j}\right),
 \end{equation}
 where $\W^{j}$ is the $j$th column of the matrix $\W=\begin{pmatrix}\U\\ \V \end{pmatrix}$. The estimator is obtained by maximizing \eqref{eq_Gaussian_likelihood} over all positive-definite symmetric matrices $\Lambda$. Ideas of a derivation showing that this maximization problem results in formulas \eqref{eq_estimation_with_sample} are provided in Problem Set 2 of Chapter \ref{Section_exercises}.
\end{proof}

Theorem \ref{Theorem_Gaussian_MLE} says that sample covariances are MLE for their population counterparts. Thus, in view of \eqref{eq_x1}, the same should hold for canonical correlations and variables and
Task \ref{Task_CCA_iid} can be approached by working with the sample CCA:\footnote{For a different interpretation of canonical correlations as MLE estimators (for another probability model) see \citet{bach2005probabilistic}.}
\begin{procedure} \label{Procedure_CCA}
 We estimate the canonical correlations and variables between $\u$ and $\v$ by their sample versions, i.e., by sample canonical correlations and variables between the linear spaces spanned by the rows of matrices $\U$ and $\V$. By Proposition \ref{prop_CCA_matrix_form} we estimate the $i$--th canonical variables and squared correlations $(\sum_{j=1}^K A_{ij} u^j, \sum_{j=1}^M B_{ij} v^j, c_i^2)$ by the triplet
 \begin{align}
 \label{eq_CCA_estimator} \Biggl(\, &\text{eigenvector } (\hat A_{i1}, \hat A_{i2},\dots,\hat A_{iK})\text{ of }\,  (\U \U^\T)^{-1} \U \V^\T (\V \V^\T)^{-1} \V \U^\T, \\
  \notag &\text{eigenvector } (\hat B_{i1}, \hat B_{i2},\dots,\hat B_{iM})\text{ of }\,   (\V \V^\T)^{-1} \V \U^\T (\U \U^\T)^{-1} \U \V^\T,\\
  \notag &\text{corresponding to the same }i\text{-th largest eigenvalue }\, \hat c_i^2 \qquad\qquad \quad \Biggr).
 \end{align}
\end{procedure}
Several remarks are in order. First, the canonical variables for $\u$ and $\v$ are random variables (mathematically, these are functions on an abstract probability space), so we cannot estimate them directly, instead we are estimating the coefficients $A_{ij}$ and $B_{ij}$ of their decomposition in coordinates $u^1,\dots,u^K$ and $v^1,\dots,v^M$. On the other hand, canonical correlations are deterministic numbers, so they are being estimated directly in the third line. Second, the eigenvectors in the first two lines of \eqref{eq_CCA_estimator} are defined only up to multiplication by an arbitrary real number. If we want to have some normalization, we need to additionally impose it. Note that in Theorem \ref{Theorem_canonical_bases} there was also freedom in multiplying $\u_i$, $\v_i$ by $-1$, so one natural approach is to avoid normalizations at all, thinking about $\u_i$ and $\v_i$ as directions (or one-dimensional lines), rather than vectors.

\smallskip

The preceding statistical reasoning suggests to use Procedure \ref{Procedure_CCA}, however, we do not yet have any guarantees on the quality of estimation. Is there some sense, in which the estimates $\hat A, \hat B$, and $\hat c_1,\dots,\hat c_K$ are close to the true values of $A$, $B$ and $c_1,\dots,c_K$? We start addressing this question below and continue in the next two chapters. We will see that the answer crucially depends on how large $K$, $M$, and $S$ are.

\subsection{Two independent data sets} Our next task is to understand the probability distribution of \eqref{eq_CCA_estimator}, given some assumptions on vectors $\u$ and $\v$. The simplest possible setting is to assume that $\u$ and $\v$ are \emph{independent} mean $0$ Gaussian vectors with arbitrary non-degenerate covariance matrices $\Lambda_{uu}$ and $\Lambda_{vv}$. That is, $\Lambda_{uv}$ is a zero matrix.

From the applied point of view, this setting can be used to construct tests for independence of two data sets. From the theoretical point of view, this setting reveals fruitful connections to classical ensembles of the random matrix theory and serves as a foundation for further studies of more complicated cases.

\begin{theorem} \label{Theorem_not_depend_on_cov} Suppose that $\U$ and $\V$ are $K\times S$ and $M\times S$ independent random matrices, whose columns are i.i.d.\ samples of independent mean $0$ Gaussian vectors $\u$ and $\v$. The distribution of the matrices $P_\U P_\V$ and $P_\U P_\V$ does not depend on the choice of covariance matrices $\Lambda_{uu}$ and $\Lambda_{vv}$. In particular, the same is true for their eigenvalues (=squared sample canonical correlations) and eigenvectors.\footnote{Recall that $P_\U P_\V= \U^\T (\U \U^\T)^{-1} \U \V^\T (\V \V^\T)^{-1} \V $ is slightly different from the first line of \eqref{eq_CCA_estimator}, but has the same non-zero eigenvalues.}
\end{theorem}
\begin{proof}
 Take a $K\times K$ non-degenerate matrix $F$ and a $M\times M$ non-degenerate matrix $G$. We claim that $P_{F \U} P_{G\V}=P_\U P_\V$. One way to see this is by noticing that as linear spaces $F\U=\U$ and $G\V=\V$. Alternative, one can use \eqref{eq_x3}:
 \begin{multline*}
  P_{F \U} P_{G\V}=\U^\T F^{\T} (F \U \U^\T F^{\T})^{-1} F \U \V^\T G^\T (G \V \V^\T G^\T )^{-1} G \V \\=\U^\T (\U \U^\T)^{-1} \U \V^\T (\V \V^\T)^{-1} \V = P_\U P_\V.
 \end{multline*}
 Under this transformation the correlation structure of the columns of $\U$ and $\V$ changes as
 $$
  \Lambda_{uu}\mapsto F \Lambda_{uu} F^\T,\qquad \Lambda_{vv}\mapsto G \Lambda_{vv} G^{\T}.
 $$
 Choosing $F=\Lambda_{uu}^{-1/2}$, $G=\Lambda_{vv}^{-1/2}$, we, therefore, can reduce arbitrary $\Lambda_{uu}$, $\Lambda_{vv}$ to identical matrices.
\end{proof}

\begin{corollary}
There is no loss of generality in choosing $\Lambda_{uu}=I_K$ and $\Lambda_{vv}=I_{M}$, which is what we will do in this and the next Chapter. Therefore, from now on both $\U$ and $\V$ are matrices of i.i.d.\ Gaussians $\mathcal N(0,1)$.
\end{corollary}


\subsection{Jacobi ensemble}

Miraculously, in the i.i.d.\ Gaussian setting for $\U$ and $\V$, the distribution of the sample canonical correlations is explicit, as was first noticed by \cite{hsu1939distribution}.

\begin{definition} \label{Definition_Jacobi_matrix}
Given two parameters $p,q>0$, the real Jacobi matrix ensemble $\J(N;p,q)$ is a distribution on $N\times N$ real symmetric matrices $\mathcal M$ of density proportional to
 \begin{equation}
  \label{eq_Jacobi_def}
  \det(\mathcal M)^{p-1} \det(I_N-\mathcal M)^{q-1}\, \dd\mathcal M,\qquad 0<\mathcal M< I_N,
 \end{equation}
 with respect to the Lebesgue measure, where $I_N$ is the $N\times N$ identity matrix, and $0< \mathcal M < I_N$ means that both $\mathcal M$ and $I_N-\mathcal M$ are positive definite.
\end{definition}
For $N=1$, the matrix reduces to its single matrix element $\mathcal M_{11}\in (0,1)$ and Definition \ref{Definition_Jacobi_matrix} becomes the density of Beta distribution:
$$
 \frac{1}{\mathfrak B(p,q)} x^{p-1} (1-x)^{q-1},\quad 0<x<1,\qquad\qquad \mathfrak B(p,q)=\int_0^1 x^{p-1}(1-x)^{q-1}\dd x=\frac{\Gamma(p)\Gamma(q)}{\Gamma(p+q)}.
$$
\begin{definition} \label{Definition_Jacobi_ev}
Given $p,q>0$, the real Jacobi eigenvalue ensemble $\mathcal J(N;p,q)$ is a distribution on $N$--tuples of real numbers $1>x_1>x_2>\dots>x_N>0$ of density
 \begin{equation}
  \label{eq_Jacobi_ev_def}
  \left[ N! \prod_{k=0}^{N-1} \frac{\Gamma(p+q+\frac{N+k-1}{2})\Gamma(\frac{3}{2})}{\Gamma(p+\frac{k}{2})\Gamma(q+\frac{k}{2})\Gamma(1+\frac{k+1}{2})}\right]\cdot  \prod_{1\le i<j\le N} (x_i-x_j) \prod_{i=1}^N x_i^{p-1} (1-x_i)^{q-1} \dd x_i.
 \end{equation}
\end{definition}
\begin{remark} The computation of the prefactor $[\cdot]$ in \eqref{eq_Jacobi_ev_def} which makes it a probability distribution is the famous Selberg integral evaluation, see, e.g.\ \cite{forrester2008importance}. \end{remark}

Both \eqref{eq_Jacobi_def} and \eqref{eq_Jacobi_ev_def} can be viewed as multidimensional generalizations of the Beta distribution. The two generalizations are closely related, as we explain in Theorem \ref{Theorem_two_Jacobis}. Their connection to the canonical correlation analysis is stated in Theorem \ref{Theorem_CCA_Jacobi}.

\begin{theorem}\label{Theorem_two_Jacobis}
 Eigenvalues of $\J(N;p,q)$--distributed matrix $\mathcal M$ are distributed as $\mathcal J(N;p,q)$.
\end{theorem}
\begin{theorem}\label{Theorem_CCA_Jacobi} Let $\U$ and $\V$ be two independent $K\times S$ and $M\times S$ random matrices, such that their $S$ columns are i.i.d.\ non-degenerate mean $0$ Gaussian vectors of dimensions $K$ and $M$, respectively, with arbitrary non-degenerate covariances. Suppose that $K\le M$ and $K+M\le S$. Then the squared sample canonical correlations $1\ge \hat c_1^2 \ge \hat c_2^2\ge\dots\ge \hat c_K^2$ between $\U$ and $\V$ have distribution $\mathcal J(K; \frac{M-K+1}{2},\frac{S-K-M+1}{2})$.
\end{theorem}
In $K=M=1$ case Theorem \ref{Theorem_CCA_Jacobi} states that the sample correlation coefficient between two Gaussian vectors has a Beta distribution. The $K\le M$ restriction is just for convenience and can be removed by symmetry between $\U$ and $\V$. The $K+M\le S$ restriction is more interesting: note that if $K+M>S$, then any $K$--dimensional linear subspace in $S$--dimensional space will have a $(K+M-S)$--dimensional intersection with any $M$--dimensional linear subspace in $S$--dimensional space. Hence, in this situation there are $K+M-S$ deterministic canonical correlations equal to $1$. One can still analyze the joint distribution of the remaining correlations (it will be again given by an instance of the Jacobi ensemble with appropriate parameters), yet we will not address it here: from the point of view of statistical applications the case of large $S$ is the most natural one, as we will see later.

Theorem \ref{Theorem_CCA_Jacobi} can be generalized to cases in which the columns of $\U$ and $\V$ are i.i.d.\ samples of Gaussian $\u$ and $\v$, where $\u$ and $\v$ allowed to be correlated, see \cite{constantine1963some}. The distribution becomes much more complicated than the Jacobi ensemble --- it involves so-called matrix hypergeometric functions.

The proofs of Theorems \ref{Theorem_two_Jacobis} and \ref{Theorem_CCA_Jacobi} rely on the techniques for computing matrix integrals and Jacobians for transformations between various matrices and their eigenvalues. Three encyclopedic sources for such techniques are \cite{Hua}, \cite{muirhead2009aspects}, and \cite{Forrester10}. We present the general schemes of the proofs in the next section omitting some of the details.

\subsection{Scheme of the proof of Theorem \ref{Theorem_two_Jacobis}} This is a very general computation which works for many ensembles of random matrices. The key point is to figure out the Jacobian of the map between a symmetric matrix and its eigenvalues.

Consider the map
$$
 \Phi: (O; \lambda_1,\dots,\lambda_N) \mapsto \X = O\, \mathrm{diag}(\lambda_1,\dots,\lambda_N)\, O^{\T},
$$
which takes a real orthogonal matrix $O$ and an $N$-tuple of real numbers $\lambda_1\ge\lambda_2\ge\dots\ge\lambda_N$ as an input and produces a real symmetric matrix as an output.

Note that for a \emph{generic} $\X$, $\Phi^{-1}(\X)$ has $2^N$ elements: $\lambda_1,\dots,\lambda_N$ are eigenvalues of $\X$ and columns of $O$ are normalized eigenvectors of $\X$. The latter are defined up to multiplication by $-1$, hence, $2^N$ options. Strictly speaking, this is only true when all $\lambda_i$ are distinct, as otherwise there is more freedom in choosing eigenvectors; however, the $\lambda_i$ coincide on a lower-dimensional set, which has probability $0$ with respect to distributions of interest, and, therefore, is not important for us.

We would like to compute the Jacobian of the transformation $\Phi$. In order to make this precise, note that the set of all $N\times N$ orthogonal matrices is a $N(N-1)/2$ dimensional manifold in $\mathbb R^{N^2}$: it is being cut by $N(N+1)/2$ quadratic equations expressing that the columns of the matrix $O$ are orthonormal. Hence, the space of all orthogonal matrices has a natural measure $\dd O$ inherited from the Lebesgue measure on $\mathbb R^{N^2}$. In parallel, $(\lambda_1,\dots,\lambda_N)$ are equipped with the Lebesgue measure $\dd \lambda_1\cdots \dd \lambda_N$ inherited from $\mathbb R^{N}$ and the symmetric matrices $\X$ are equipped with the Lebesgue measure $\dd X$ inherited from $\mathbb R^{N(N+1)/2}$, identified with the space of matrix elements on and above the diagonal. The following lemma explains how these measures are related under $\Phi$:

\begin{lemma} \label{Lemma_ev_Jacobian}
 Under the map $\Phi$ we have
 $$
  \prod_{i<j}(\lambda_i-\lambda_j)\, \dd \lambda_1 \cdots \dd \lambda_N \, \dd O \stackrel{\Phi}{\longrightarrow} c_N\, \dd \X,
 $$
 where $c_N$ is an omitted explicit constant, which depends on $N$, but not on $\X$.
\end{lemma}
\begin{proof}
 We show that $\dd\Phi(O;\lambda_1,\dots,\lambda_N)\sim \prod_{i<j}(\lambda_i-\lambda_j)\, \dd \lambda_1 \cdots \dd \lambda_N \, \dd O$ by computing the Jacobian of the map $\Phi$ at an arbitrary point $(O; \lambda_1\ge \lambda_2\dots\ge \lambda_N)$. Without loss of generality, we can impose two simplifying assumptions. First, the inequalities are strict $\lambda_1>\lambda_2>\dots>\lambda_N$, as the $\lambda$'s with equal coordinates have measure $0$. Second, it is sufficient to only compute the Jacobian at the identical matrix, $O=I_N$. This is because of the obvious property of $\Phi$:
 $$
  \Phi(\tilde O O; \lambda_1,\dots,\lambda_N)= \tilde O  \Phi(O; \lambda_1,\dots,\lambda_N) \tilde O^\T,
 $$
 and an observation that multiplication of $O$ by $\tilde O$ preserves the Lebesgue measure (as any orthogonal transformation), and similarly the map $\X\mapsto \tilde O \X \tilde O^\T$ preserves the Lebesgue measure. Hence, the Jacobian at any $O$ is equal to the Jacobian at $O=I_N$.

 For $O$ near the identity matrix we can introduce a natural coordinate system on the orthogonal matrices:
 \begin{align*}
  O&=\exp(R)=I_N+R+\frac{R^2}{2!}+\frac{R^2}{3!}+\dots, \\ R&=\ln(O)=\ln(I_N+ O-I_N)=(O-I_N)-\frac{(O-I_N)^2}{2}+\frac{(O-I_N)^3}{3}-\dots.
 \end{align*}
 One can check that the above correspondence is a diffeomorphism between orthogonal matrices $O$ close to $I_N$ and skew-symmetric matrices  $R$ is (i.e., $R^\T=-R$) close to $0$. Because locally near $R=0$, the coordinate system simplifies to a shift
 $O\approx I_N+R$,
 the Lebesgue measure $\dd O$ near the identity matrix becomes $\dd R$ (i.e., the Lebesgue measure on $N(N-1)/2$ matrix elements of $R$ above the diagonal) --- the formal statement is that the Jacobian of the transformation between $O$ and $R$ is $1$ at $R=0$ or $O= I_N$.

 Let us rewrite the map $\Phi$ in terms of $R=[R_{ij}]_{i,j=1}^N$:
 \begin{multline}
  \Phi(\exp(R);\lambda_1,\dots,\lambda_N)=(I_N+R)\mathrm{diag}(\lambda_1,\dots,\lambda_N)(I_N-R) + o(R)\\=
  \mathrm{diag}(\lambda_1,\dots,\lambda_N) +R\cdot \mathrm{diag}(\lambda_1,\dots,\lambda_N)-\mathrm{diag}(\lambda_1,\dots,\lambda_N)\cdot R+o(R)= o(R)
  \\+ \begin{pmatrix} \lambda_1 & R_{12}(\lambda_2-\lambda_1) & R_{13}(\lambda_3-\lambda_1) & \dots
          \\ R_{21}(\lambda_1-\lambda_2) & \lambda_2 & R_{23}(\lambda_3-\lambda_2) & \dots \\
            \vdots &&&\ddots & R_{N-1,N}(\lambda_N-\lambda_{N-1}) \\ &&& R_{N,N-1}(\lambda_{N-1}-\lambda_N) & \lambda_N
            \end{pmatrix}.
 \end{multline}
 The coordinates on symmetric matrices are their entries above and on the diagonal. Hence, we need to differentiate all of them in $R_{ij}$, $1\le i<j\le N$ and $\lambda_{i}$, $1\le i \le N$ to get the Jacobian matrix. The Jacobian matrix at $R=0$ is diagonal: the part corresponding to the partial derivatives of the above the diagonal elements with respect to $R_{ij}$ is diagonal with $(\lambda_j-\lambda_i)$ elements; the partial derivatives of the diagonal elements with respect to $\lambda_i$ are the identity matrix. We conclude that the Jacobian is $\prod_{i<j}(\lambda_j-\lambda_i)$, as claimed in the lemma.
\end{proof}

\begin{proof}[Proof of Theorem \ref{Theorem_two_Jacobis}] We represent $\mathcal M$ in Definition \ref{Definition_Jacobi_matrix} in terms of its eigenvalues and eigenvectors as $\mathcal M=\Phi(O;\lambda_1,\dots,\lambda_N)$, rewrite
 $\det(\mathcal M)^{p-1} \det(I_N-\mathcal M)^{q-1}$ in terms of $\lambda_i$ and use Lemma \ref{Lemma_ev_Jacobian} to transform \eqref{eq_Jacobi_def} into
 $$
  C_N \cdot \prod_{i<j}(\lambda_i-\lambda_j) \prod_{i=1}^N \lambda_i^{p-1}(1-\lambda_i)^{q-1} \dd \lambda_1\cdots \dd \lambda_N\, \dd O,
 $$
 where $C_N$ is a normalization constant, whose exact value we omit from this computation (it is given by the ratio of the constant in \eqref{eq_Jacobi_def} and the constant $c_N$ from Lemma \ref{Lemma_ev_Jacobian}). Integrating out $\dd O$ and renaming $\lambda_i$ into $x_i$, we arrive at \eqref{eq_Jacobi_ev_def}.
\end{proof}

\subsection{Scheme of the proof of Theorem \ref{Theorem_CCA_Jacobi}}

An important part of \eqref{eq_CCA_estimator} is the matrices $\U \U^\T$ and $\V\V^\T$. The first step is to find their distribution.
\begin{definition} \label{Definition_Wishart_matrix}
Given a parameter $p>0$, the real Wishart (also known as Laguerre) matrix ensemble $\WL(N;p)$ is a distribution on $N\times N$ real symmetric positive-definite matrices $\mathcal M$ of density (with respect to the Lebesgue measure) proportional to
 \begin{equation}
  \label{eq_Wishart_def}
  \det(\mathcal M)^{p-1}  \exp\left(-\frac{1}{2} \mathrm{Trace}(\mathcal M) \right)\, \dd\mathcal M,\qquad \mathcal M>0.
 \end{equation}
\end{definition}

\begin{lemma} \label{Lemma_Wishart}
 For $L\ge K$, let $\Z$ be a $K\times L$ matrix of i.i.d.\ $\mathcal N(0,1)$ random variables. Then the distribution of $K\times K$ symmetric matrix $\Z \Z^\T$ is $\WL(K;\frac{L-K+1}{2})$
\end{lemma}
The matrix $\Z \Z^\T$  is often called the sample covariance matrix in the literature.
\begin{proof}[Ingredients of the proof of Lemma \ref{Lemma_Wishart}] Let us start from the simplest $L=K=1$ case. $\Z$ has a unque matrix element $Z_{11}$, whose probability distribution is
$$
 \frac{1}{\sqrt{2 \pi}} \exp\left(-\frac{1}{2} (Z_{11})^2\right)\dd Z_{11}.
$$
The only matrix element of $\Z \Z^{\T}$ is $x=(Z_{11})^2$. By using the change of variables formula for the probability distributions, we compute the density of $x$ to be
\begin{equation}
\label{eq_x6}
 2 \cdot \frac{1}{\sqrt{2\pi}}   \exp\left(-\frac{1}{2} x\right)\dd (\sqrt{x})= \frac{1}{\sqrt{2\pi}} x^{-1/2}  \exp\left(-\frac{1}{2} x\right)\dd x,
\end{equation}
where the prefactor $2$ comes from the observation that the map $z\mapsto z^2$ glues two points $z$ and $-z$ together.
Clearly, \eqref{eq_x6} matches the statement of the lemma. A more complicated case $K=1$, $L\ge 1$ is covered in Problem set 1 of Chapter \ref{Section_exercises}. In the most general case, the probability density of $\Z=[Z_{ij}]$ is
\begin{equation}
\label{eq_x4}
 (2\pi)^{-\frac{KL}{2}} \exp\left(-\frac{1}{2}\sum_{i,j} (Z_{ij})^2\right)= (2\pi)^{-\frac{KL}{2}} \exp\left(-\frac{1}{2}\mathrm{Trace}(\Z\Z^\T)\right).
\end{equation}
It remains to project the measure \eqref{eq_x4} along the map $\U\mapsto \X=\U \U^\T$. We see that the second factor in \eqref{eq_Wishart_def} is the same as the second factor in \eqref{eq_x4}. The first factor accounts for the change of measure in the map. It is similar to the Jacobian of the transformation, which we have just computed in \eqref{eq_x6}, yet a technical difficulty is that we map $KL$ dimensional space into $K(K+1)/2$ dimensional space, i.e., the dimension goes down. This means that we need to introduce additional coordinates encoding $\Z$ (if we assume $\X=\Z\Z^\T$ to be known), and then integrate them out\footnote{For example, if $K=L=1$, then the additional coordinate is $\pm 1$ and it is responsible for the prefactor $2$ in \eqref{eq_x6}. If $K=1$, $L=2$, then we integrate over a circle and a natural choice of the coordinate is the angle.}. A standard way to introduce these coordinates is through the  QR--decomposition for rectangular matrices. Further details of this approach are in \cite[Chapter 3]{muirhead2009aspects} or \cite[Section 3.2.3]{Forrester10}. An alternative approach is in \cite[Section 3.1.1]{Forrester10}.
\end{proof}
Two independent Wishart matrices give rise to an instance of the Jacobi ensemble.
\begin{lemma} \label{Lemma_MANOVA}
 Suppose $K\le L$ and $K\le Q$. Let $\Z$ be $K\times L$ matrix of i.i.d.\ $\mathcal N(0,1)$ and let $\Y$ be independent $K\times Q$ matrix of i.i.d.\ $\mathcal N(0,1)$. Then the distribution of the $K\times K$ symmetric matrix
 \begin{equation}
 \label{eq_MANOVA}
  \X=(\Z\Z^\T+\Y \Y^\T)^{-1/2} \Z \Z^\T (\Z\Z^\T + \Y \Y^\T)^{-1/2}
 \end{equation}
 is $\J\bigl(K;\frac{L-K+1}{2},\frac{Q-K+1}{2}\bigr)$.
\end{lemma}
The matrix in \eqref{eq_MANOVA} is commonly referred to in the literature as the MANOVA ensemble, with the name derived from Multivariate Analysis of Variance.
\begin{proof}[Proof of Lemma \ref{Lemma_MANOVA}]
 By Lemma \ref{Lemma_Wishart}, the joint probability distribution of $A=\Z\Z^\T$ and $B=\Y\Y^\T$ is proportional to
$$
 \det(A)^{\frac{L-K-1}{2}} \det(B)^{\frac{Q-K-1}{2}} \exp\left(-\frac{1}{2}\mathrm{Trace}(A+B)\right) \dd A\, \dd B.
$$
Introducing the matrix $C=A+B$ and noting that the transformation $(A,B)\mapsto (A,A+B)$ has Jacobian $1$, the joint distribution of $A$ and $C$ is
$$
 \det(A)^{\frac{L-K-1}{2}} \det(C-A)^{\frac{Q-K-1}{2}} \exp\left(-\frac{1}{2}\mathrm{Trace}(C)\right) \dd A\, \dd C.
$$
Note that $\X=C^{-1/2} A C^{-1/2}$ and that the transformation $A\mapsto \X$ in the space of $K\times K$ positive-definite symmetric matrices has Jacobian: $\dd\A=\det(C)^{\frac{K+1}{2}}\dd \X$, see, e.g.\ \cite[(1.35)]{Forrester10}. Hence, the joint distribution of $\X$ and $C$ is
\begin{multline*}
 \det(C)^{\frac{K+1}{2}} \det(C^{1/2} \X C^{1/2})^{\frac{L-K-1}{2}} \det(C-C^{1/2} \X C^{1/2})^{\frac{Q-K-1}{2}} \exp\left(-\frac{1}{2}\mathrm{Trace}(C)\right) \dd \X\, \dd C\\=
 \det(\X)^{\frac{L-K-1}{2}} \det(I_K-\X)^{\frac{Q-K-1}{2}} \exp\left(-\frac{1}{2}\mathrm{Trace}(C)\right)  \det(C)^{\frac{K+1}{2}+\frac{L-K-1}{2}+\frac{Q-K-1}{2}} \dd \X\, \dd C.
\end{multline*}
The density has a factorized form and to get the distribution of $\X$, we collect all $\X$--dependent factors, arriving at $\J\bigl(K;\frac{L-K+1}{2},\frac{Q-K+1}{2}\bigr)$.
\end{proof}

\begin{proof}[Proof of Theorem  \ref{Theorem_CCA_Jacobi}] By Theorem \ref{Theorem_not_depend_on_cov} we can assume without loss of generality that all matrix elements $\U$ and $\V$ are i.i.d.\ $\mathcal N(0,1)$, which is the approach taken throughout the proof.

 The squared sample canonical correlations between $\U$ and $\V$ are eigenvalues of $P_\U P_\V$, which is the same as eigenvalues of $P_\U P_\V P_\V $ and the same as eigenvalues of $P_\V P_\U P_\V$. The idea of the proof is to convert $P_\V P_\U P_\V$ into the form of Lemma \ref{Lemma_MANOVA} by an orthogonal transformation of the ambient $S$--dimensional space. Choose an $S\times S$ orthogonal matrix $O$, which maps the rows of $\V$ into the first $M$ coordinate vectors, i.e.,
$$
 \V O^\T=\begin{pmatrix} 1 & 0 & \dots\\ 0 & 1 & 0 & \dots \\ & & \ddots\\   0 & \dots & 0 & 1 & 0 & \dots\end{pmatrix}.
$$
Directly from the definition of the projector (or from the matrix formulas, as in \eqref{eq_x3}), one sees that
$$
 O P_\V P_\U P_\V O^\T = P_{\V O^{\T}} P_{\U O^\T} P_{\V O^{\T}}.
$$
Since conjugations do not change eigenvalues, we conclude that the desired canonical correlations are the eigenvalues of the matrix $ P_{\V O^{\T}} P_{\U O^\T} P_{\V O^{\T}}$. The latter matrix is the top-left $M\times M$ corner of the projector $P_{\U O^{\T}}$.

Next, note that $\U O^{\T} \stackrel{d}{=} \U$. Indeed, $O$ is a function of $\V$, and therefore is independent from $\U$ and can be taken to be deterministic for the purpose of this computation. Each row of $\U$ is a vector of i.i.d.\ $\mathcal N(0,1)$, and such vectors are invariant (in distribution) under orthogonal transformations, because Gaussian vectors are uniquely determined by the covariance, and the latter changes by $I_{S}\mapsto O I_S O^{\T} = I_S$.

The conclusion from this discussion is that the positive canonical correlations between $\U$ and $\V$ have the same distribution as the positive eigenvalues of $M\times M$ corner of the matrix $P_{\U}=\U^{\T} (\U \U^\T)^{-1} \U$.

Let us split the $K\times S$ matrix $\U$ into two parts: the first $M$ columns form $K\times M$ matrix $\Z$ and the remaining $S-M$ columns form $K\times (S-M)$ matrix $\Y$. Then the $M\times M$ corner of the matrix $P_{\U}=\U^{\T} (\U \U^\T)^{-1} \U$ is the same as the $M\times M$ matrix  $\Z^{\T} (\Z\Z^{\T}+\Y \Y^{\T})^{-1} \Z$. Using the fact that for any rectangular matrices $A$ and $B$ of same dimensions the non-zero eigenvalues of $AB$ and $BA$ coincide, we convert the eigenvalues:
\begin{multline*}
 \mathrm{e.v.}\left(\Z^{\T} (\Z\Z^{\T}+\Y \Y^{\T})^{-1} \Z\right)=\mathrm{e.v.}\left( (\Z\Z^{\T}+\Y \Y^{\T})^{-1} \Z\Z^{\T}\right)\\=\mathrm{e.v.}\left( (\Z\Z^{\T}+\Y \Y^{\T})^{-1/2} \Z\Z^{\T}(\Z\Z^{\T}+\Y \Y^{\T})^{-1/2}\right).
\end{multline*}
 The last matrix is that of Lemma \ref{Lemma_MANOVA} with $L=M$ and $Q=S-M$. We conclude that this matrix has distribution
 $\J\bigl(K;\frac{L-K+1}{2},\frac{Q-K+1}{2}\bigr)=\J\bigl(K;\frac{M-K+1}{2},\frac{S-K-M+1}{2}\bigr)$. By Theorem \ref{Theorem_two_Jacobis}, this implies that the eigenvalues have distribution $\mathcal J\bigl(K;\frac{M-K+1}{2},\frac{S-K-M+1}{2}\bigr)$.
\end{proof}

\newpage

\section{Limits of the Jacobi ensemble}

In the last chapter we have shown that the distribution of the $K$ squared (sample) canonical correlations between two independent spaces spanned by i.i.d.\ Gaussian matrices of $K\times S$ and $M\times S$ sizes has density proportional to
\begin{equation}
\label{eq_CCA_Jacobi_specialized}
 \prod_{1\le i<j\le K} (x_i-x_j) \prod_{i=1}^K x_i^{\frac{M-K-1}{2}} (1-x_i)^{\frac{S-K-M-1}{2}}, \qquad 1>x_1>x_2>\dots>x_K>0.
\end{equation}
The goal for this chapter is to analyze the behavior of \eqref{eq_CCA_Jacobi_specialized} in various asymptotic regimes.

\subsection{Finite $K$ and $M$} We begin by assuming that $K$ and $M$ are fixed and small, while $S$ is large. Based on the interpretations in Section \ref{Section_CCA_as estimation}, this means that we have \emph{many} i.i.d.\ samples of low-dimensional vectors $\u$ and $\v$. The following result goes back to \cite{hsu1941limiting}.

\begin{theorem} \label{Theorem_CCA_to_Laguerre}
 Take $M\ge K$ and suppose that $\x(K,M,S)$ is a $K$--dimensional random vector distributed as \eqref{eq_CCA_Jacobi_specialized}. Then
 \begin{equation}
  \lim_{S\to\infty} \left[S\cdot \x(K,M,S)\right] \stackrel{d}{=} \y(K,M),
 \end{equation}
 where $\y(K,M)$ is a $K$--dimensional random vector with density proportional to
\begin{equation}
\label{eq_Laguerre_limit}
 \prod_{1\le i<j\le K} (y_i-y_j) \prod_{i=1}^K y_i^{\frac{M-K-1}{2}} \exp\left(-\frac{1}{2}y_i\right), \qquad y_1>y_2>\dots>y_K>0.
\end{equation}
\end{theorem}
\begin{remark}
 Using Lemma \ref{Lemma_Wishart} and repeating the proof of Theorem \ref{Theorem_two_Jacobis}, one can show that $\y(K,M)$ has the distribution of $K$ eigenvalues of the Wishart matrix $\Z \Z^\T$, where $\Z$ is $K\times M$ matrix of i.i.d.\ $\mathcal N(0,1)$.
\end{remark}
\begin{proof}
 We substitute $x_i=\frac{y_i}{S}$ into \eqref{eq_CCA_Jacobi_specialized} and take the limit as $S\to\infty$, using ${\lim_{S\to\infty}(1-\tfrac{y}{S})^{S/2}=\exp(-y/2)}$.
\end{proof}

Theorem \ref{Theorem_CCA_to_Laguerre} has two conceptual consequences. First, it shows that the squared sample canonical correlations converge as $S\to\infty$ to the true squared canonical correlations, which are all zeros in this case, as two data-sets are independent; the speed of convergence is $\frac{1}{S}$. Second, the theorem can be used to construct a \emph{statistical test} for the hypothesis ``two data-sets are independent'':

\begin{procedure} \label{Procedure_independence_test_1} Choose the confidence level\footnote{It is typical to choose $\alpha=0.9,\,\alpha=0.95$, or $\alpha=0.99$.} $\alpha\in(0,1)$. Let $q_\alpha$ be $\alpha$--quantile for the largest coordinate in $\y(K,M)$:
$$
 \mathrm{Prob}\bigl(\text{largest coordinate of }\y(K,M)<q_\alpha\bigr)=\alpha.
$$
Given two rectangular data-sets $\U$ and $\V$ of $K\times S$ and $M\times S$ sizes, compute their largest squared sample canonical correlations $\hat c_1^2$. If
$$
 S\hat c_1^2> q_\alpha,
$$
then \emph{reject} the hypothesis that $\U$ and $\V$ are independent.
\end{procedure}
Theorem \ref{Theorem_CCA_to_Laguerre} coupled with Theorem \ref{Theorem_not_depend_on_cov} implies that for $\U$ and $\V$ with i.i.d.\ Gaussian columns, the independence of $\U$ and $\V$ implies that $Sc_1^2>q_\alpha$ with probability close to $1-\alpha$ for large $S$. Hence, for $\alpha$ close to $1$, this is a very unlikely event, which is the reason for the rejection of the hypothesis. In fact, one can relax the Gaussianity assumption in Theorem \ref{Theorem_CCA_to_Laguerre} (see Problem set 3 in Section \ref{Section_exercises}), which implies the asymptotic validity of Procedure \ref{Procedure_independence_test_1} for non-Gaussian $\U$ and $\V$. On the other hand, a limitation of the procedure is that the assumption of i.i.d.\ columns can only be slightly relaxed --- uncorrelated columns with some mild technical conditions are permissible --- but cannot be entirely removed.

\subsection{Large $K$ and $M$} The next step is to investigate the case of large $K$ and large $M$. Figure \ref{Fig_Wachter_simulation} shows a histogram of all squared canonical correlations obtained from a single simulation for i.i.d.~$\mathcal{N}(0,1)$ matrix elements, $K=100$, $M=150$, $S=500$. The histogram is constructed by splitting the $[0,1]$ interval into bins, counting the squared correlations within each bin, and normalizing so that the area under the histogram equals $1$. While the histogram is random, the simulation reveals a notable feature: Even though the true canonical correlations are all zero, the sample correlations are far from being zero. We will quantify this observation in two theorems below, first by examining all canonical correlations together and then by focusing on the largest ones.

\begin{figure}[t]
\includegraphics[width=0.7\linewidth]{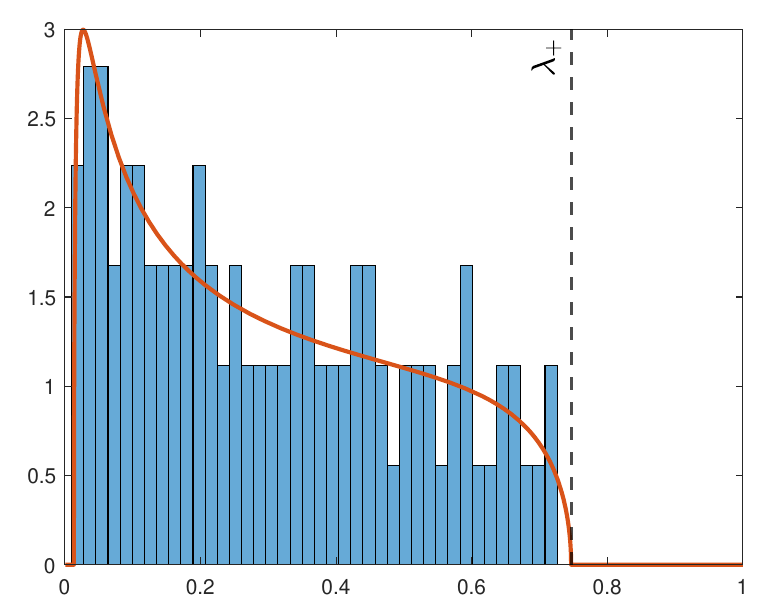 }
\caption{Histogram of simulated squared sample canonical correlations (blue columns). $K=100$, $M=150$, $S=500$. The density of the Wachter distribution is shown in orange.}
\label{Fig_Wachter_simulation}
\end{figure}

For the first theorem, we represent $K$ squared sample correlations $\hat c_i^2$ by an atomic probability measure, called the \emph{empirical measure}:
$$
 \mu_{K,M,S}=\frac{1}{K}\sum_{i=1}^K \delta_{\hat c_i^2}.
$$

The asymptotic density of $\mu_{K,M,S}$ is precisely what we see in orange in Figure \ref{Fig_Wachter_simulation}. It is typical to think of a measure through its pairings with test-functions $f$:
$$
 \mu_{K,M,S}[f]=\frac{1}{K} \sum_{i=1}^k f(\hat c_i^2).
$$
Note that since $\hat c_i^2$ are random, $\mu_{K,M,S}[f]$ is also a random variable. In particular, when $f$ is an indicator function of the interval $(-\infty,t]$, we get the \emph{empirical distribution function}:
$$
 \mu_{K,M,S}[ \mathbf 1_{(-\infty,t]}]=\frac{1}{K} \sum_{i=1}^K \mathbf 1_{c_i^2\le t} = \frac{\#\{c_i^2 \text{ which are smaller or equal than } t\}}{K}.
$$

Let us introduce the limiting object from Figure \ref{Fig_Wachter_simulation}, which is called \emph{the Wachter distribution}, paying tribute to \cite{wachter1980limiting}.
\begin{definition} \label{Definition_Wachter}
For two real parameters  $\tau_K\ge \tau_M>1$ with $\tau_K^{-1}+\tau_M^{-1}<1$, the \emph{Wachter distribution} $\omega_{\tau_K,\tau_M}$ is defined through its density
\begin{equation}
\label{eq_Wachter_density}
 \omega_{\tau_K,\tau_M}(x)\, \dd x=\frac{\tau_K}{2\pi } \frac{\sqrt{(x-\lambda_-)(\lambda_+-x)}}{x (1-x)} \mathbf 1_{[\lambda_-,\lambda_+]}\, \dd x,
\end{equation}
where the support $[\lambda_-,\lambda_+]$ of the measure is defined via
\begin{equation}
\label{eq_lambda_pm_def}
\lambda_\pm=\left(\sqrt{\tau_M^{-1}(1-\tau_K^{-1})}\pm \sqrt{\tau_K^{-1}(1-\tau_M^{-1})}  \right)^2.
\end{equation}
\end{definition}
One can check directly that $0<\lambda_-<\lambda_+<1$ for every  $\tau_K>\tau_M>1$ with $\tau_K^{-1}+\tau_M^{-1}<1$ and that \eqref{eq_Wachter_density} is a probability density, i.e., it integrates to $1$. For any integrable function $f(x)$ we denote
$$
  \omega_{\tau_K,\tau_M}[f]=\int_0^1 f(x) \omega_{\tau_K,\tau_M}(x) \dd x.
$$

\begin{theorem}\label{Theorem_CCA_LLN} Let $\U$ and $\V$ be two independent $K\times S$ and $M\times S$ random matrices, such that their $S$ columns are i.i.d.\ non-degenerate mean $0$ Gaussian vectors of dimensions $K$ and $M$, respectively, with arbitrary non-degenerate covariances.
 Let $\mu_{K,M,S}$ be the empirical measure of the correspinding squared sample canonical correlations. Suppose that $K,M,S\to \infty$, in such a way that $\frac{S}{K}\to\tau_K$, $\frac{S}{M}\to\tau_M$ with $\tau_K\ge \tau_M>1$ and $\tau_K^{-1}+\tau_M^{-1}<1$. Then
 \begin{equation}
 \label{eq_CCA_LLN}
  \lim_{K,M,S\to\infty}  \mu_{K,M,S} = \omega_{\tau_K,\tau_M}, \qquad \text{weakly, in probability,}
 \end{equation}
 which means that for each continuous function $f(x)$, we have
 \begin{equation}
  \label{eq_CCA_LLN_f}
   \lim_{K,M,S\to\infty}  \mu_{K,M,S}[f]=\omega_{\tau_K,\tau_M}[f], \qquad \text{in probability.}
 \end{equation}
\end{theorem}
\begin{remark}
 How large is the difference $\mu_{K,M,S}[f]-\omega_{\tau_K,\tau_M}[f]$? The answer depends on the smoothness of $f$. If $f$ is smooth enough (any analytic function is fine, and any function differentiable sufficient number of times is also fine), then the distance is $\frac{1}{S}$ with a Gaussian limit:
 \begin{equation}
 \label{eq_CLT_smooth}
 S\cdot \bigl(\mu_{K,M,S}[f]-\omega_{\tau_K,\tau_M}[f]\bigr)\xrightarrow[K,M,S\to\infty]{d} \xi_f, \qquad \xi_f\sim \mathcal N(m(f),\sigma^2(f)).
 \end{equation}
 The expressions for the mean $m(f)$ and covariance $\sigma^2(f)$ are explicit, see \cite{dumitriu2012global,borodin2015general}. On the other hand, suppose that $f$ is an indicator function, say $f=\mathbf 1_{[0,y]}$ for some $0\le y\le 1$. Then the scaling gets an additional logarithmic factor, see \citet{killip2008gaussian}:
 \begin{equation}
 \label{eq_CLT_indicator}
  \frac{S}{\sqrt{\ln(S)}} \cdot \biggl(\mu_{K,M,S}[\mathbf 1_{[0,y]}]-\omega_{\tau_K,\tau_M}[\mathbf 1_{[0,y]}]\biggr)\xrightarrow[K,M,S\to\infty]{d} \eta_y, \qquad \eta_y\sim \mathcal N(0,\sigma^2(y)).
 \end{equation}
 Interestingly, while in \eqref{eq_CLT_smooth}, the Gaussian $\xi_f$ are correlated as $f$ varies, in \eqref{eq_CLT_indicator} the random variables $\eta_y$ are independent over $y$.
\end{remark}

There are many approaches to proving Theorem \ref{Theorem_CCA_LLN} and analyzing fluctuations \eqref{eq_CLT_smooth}, \eqref{eq_CLT_indicator}:
\begin{itemize}
 \item The oldest approach reduces the problem to the study of the spectrum of a product of a Wishart matrix and a deterministic matrix, see \cite{wachter1980limiting}, \cite{bai2010spectral}, also \cite{bouchaud2007large}.

 \item One can produce a tridiagonal matrix model for the canonical correlations and pass to the limit using it, see \cite{killip2004matrix}, \cite{dumitriu2012global}.

 \item One can develop a version of the Fourier transform for the Jacobi ensemble and then extract asymptotic information from it using the Macdonald difference operators and their limits, see \cite{borodin2015general}, \cite{gorin2018interlacing}.

  \item Another classical approach develops formulas for marginals of the distribution of canonical correlations via Pfaffians involving Jacobi orthogonal polynomials, see \cite{Meh,Forrester10}.

 \item One can prove the existence of the limit in \eqref{eq_CCA_LLN} using the variational problem for the empirical measure, and then identify the limit and fluctuations around it, using the Dyson--Schwinger or loop equations, see \cite{guionnet2019asymptotics} for an overview.
\end{itemize}
  We sketch a proof via the last method in the next section. Our approach will crucially rely on the Gaussianity assumption, but other tools allow one to extend Theorem \ref{Theorem_CCA_LLN} to non-Gaussian matrix elements of $\U$ and $\V$.

\smallskip

  Note that while Theorem \ref{Theorem_CCA_LLN} describes all canonical correlations together, it does not say anything about their individual limits: for instance, the limits \eqref{eq_CCA_LLN}, \eqref{eq_CCA_LLN_f} will not change if we shift $\hat{c}_1^2$ by an arbitrary number. Hence, we complement Theorem \ref{Theorem_CCA_LLN} with the following theorem describing the largest and smallest canonical correlations.\footnote{While the correlations in the middle (or ``in the bulk'') are also well-understood, we will not cover them in these lectures, as they are less relevant for the statistical applications we focus on.} We define:
  \begin{equation}
    c_\pm=\frac{\tau_K}{2} \frac{\sqrt{\lambda_+-\lambda_-}}{\lambda_\pm (1-\lambda_\pm)},
  \end{equation}
   and note that
  $$
   \omega_{\tau_K,\tau_M}(x)\approx \frac{c_\pm}{\pi} \sqrt{|x-\lambda_{\pm}|}, \text{ as } x\to\lambda_{\pm}\quad \text{ inside }\quad  [\lambda_-,\lambda_+].
  $$
 \begin{theorem}\label{Theorem_CCA_extreme}
 Let $\U$ and $\V$ be two independent $K\times S$ and $M\times S$ random matrices, such that their $S$ columns are i.i.d.\ non-degenerate mean $0$ Gaussian vectors of dimensions $K$ and $M$, respectively, with arbitrary non-degenerate covariances. Let $\hat c_1^2\ge \hat c_2^2\ge \dots \ge \hat c_K^2$ be the corresponding squared sample canonical correlations. Suppose that $K,M,S\to \infty$, in such a way that $\frac{S}{K}\to\tau_K$, $\frac{S}{M}\to\tau_M$ with $\tau_K>\tau_M>1$ and $\tau_K^{-1}+\tau_M^{-1}<1$. Then, in the sense of convergence of finite-dimensional distributions:
 \begin{equation}
      \label{eq_Jacobi_up_edge}
        \lim_{K,M,S\to\infty} \left\{ K^{2/3} c_+^{2/3} \left(\hat c_i^2- \lambda_+\right)  \right\}_{i=1}^{N}= \{\aa_i\}_{i=1}^{\infty}
 \end{equation}
  and\footnote{The limiting processes $\{\aa_i\}_{i=1}^{\infty}$ arising for the largest and smallest eigenvalues are independent. If $\tau_K=\tau_M$, then \eqref{eq_Jacobi_up_edge} is still valid, but $\lambda_-=0$ and \eqref{eq_Jacobi_low_edge} changes.}
    \begin{equation}
    \label{eq_Jacobi_low_edge}
        \lim_{K,M,S\to\infty} \left\{ K^{2/3} c_-^{2/3} \left(\lambda_--\hat c_{K+1-i}^2\right)  \right\}_{i=1}^{N}= \{\aa_i\}_{i=1}^{\infty},
    \end{equation}
    where $\aa_1>\aa_2>\aa_3>\dots$ are coordinates of the Airy$_1$ random point process.
\end{theorem}
The limiting point process $\{\aa_i\}_{i=1}^{\infty}$ frequently shows up in random matrix theory: largest (or smallest) eigenvalues in many ensembles of real symmetric random matrices converge to it. The law of the largest coordinate $\aa_1$ is called the \emph{Tracy--Widom distribution} $F_1$, and it is even more widespread, appearing also in the limit theorems for interacting particles systems and stochastic PDES, such as the celebrated Kardar-Parisi-Zhang (KPZ) equation. All existing formulas for the finite-dimensional distributions of  $\{\aa_i\}_{i=1}^{\infty}$ are quite complicated. Yet, the distribution can be easily sampled and tabulated, see \cite{bornemann2009numerical,vignette_largevars}.
\smallskip

We do not provide a proof of Theorem \ref{Theorem_CCA_extreme} here. Possible approaches include:
\begin{itemize}
\item The first proof in \cite{Johnstone_Jacobi} combines the technology of Pfaffian point processes with the asymptotic analysis of Jacobi orthogonal polynomials.
\item Tridiagonal matrix model of \cite{killip2004matrix} can be combined with asymptotic analysis for the latter of \cite{ramirez2011beta} or \cite{gorin2018stochastic}.
\item The universality considerations identify the limit in \eqref{eq_Jacobi_up_edge}, \eqref{eq_Jacobi_low_edge} by comparing it with edge limits for other ensembles of random matrices, see e.g., \cite{deift2009random}, \cite{bourgade2014edge}. Comparison approaches also allow to relax the Gaussianity assumption for $\U$ and $\V$, see \cite{HanPanZhang_2016}, \cite{FanYang}.
\end{itemize}

\bigskip

Theorems \ref{Theorem_CCA_LLN} and \ref{Theorem_CCA_extreme} have two conceptual consequences.

First, they serve as a warning to practitioners applying CCA to real-world data sets: In high-dimensional setting the presence of bounded away from $0$ squared sample canonical correlations does \emph{not} imply that two data sets are dependent. Care is needed when interpreting CCA results to avoid drawing spurious conclusions.

Second, in parallel with Procedure \ref{Procedure_independence_test_1}, Theorems \ref{Theorem_CCA_LLN} and \ref{Theorem_CCA_extreme} can serve as a foundation for constructing statistical tests for the hypothesis of independence\footnote{For recent work on testing joint independence of more than two datasets see \cite{bao2024spectral}.}  of $\U$ and $\V$: one can either use all canonical correlations, relying on \eqref{eq_CLT_smooth}, \eqref{eq_CLT_indicator} or use the largest canonical correlations relying on \eqref{eq_Jacobi_up_edge}.

\subsection{Existence of the limit of empirical measure in CCA}

We begin by outlining the main ideas behind one particular proof of Theorem \ref{Theorem_CCA_LLN}. Many technical details in this approach can be reconstructed from \cite{guionnet2019asymptotics} and references therein.

The first step is to develop the \emph{variational principle}, which is based on the following rewriting of the Jacobi density $\mathcal J(K,p,q)$:

\begin{lemma} \label{Lemma_variational_form}
 Take a vector $\x=(x_1,\dots,x_K)$ and let $\mu^{\x}=\frac{1}{K} \sum_{i=1}^K \delta_{x_i}$ be its empirical measure. The Jacobi density \eqref{eq_Jacobi_ev_def} can be rewriten, up to a constant prefactor, as\footnote{$\mathfrak E$ in the notation stays for energy.}
 \begin{equation}
 \label{eq_Var_principle}
  \exp\left( K^2 \cdot \mathfrak E(\mu^{\x}) \right),\qquad\text{where}
 \end{equation}
 \begin{equation}
  \mathfrak E(\mu)=\frac{1}{2}\iint_{x\ne y} \ln|x-y|\, \mu(\dd x) \mu(\dd y) + \int_0^1\left[ \frac{p-1}{K} \ln(x)+ \frac{q-1}{K}\ln(1-x)\right] \mu(\dd x).
 \end{equation}
\end{lemma}
\begin{proof} Direct transformation \eqref{eq_Jacobi_ev_def} based on $\int f(x) \mu^\x(\dd x)=\frac{1}{K}\sum_{i=1}^K f(x_i)$. \end{proof}

\begin{lemma} \label{Lemma_LLN_var}
 Suppose that $K,M,S\to \infty$, in such a way that $\frac{S}{K}\to\tau_K$, $\frac{S}{M}\to\tau_M$ with $\tau_K\ge \tau_M>1$ and $\tau_K^{-1}+\tau_M^{-1}<1$. Then
 \begin{equation}
 \label{eq_CCA_LLN_2}
  \lim_{K,M,S\to\infty}  \mu_{K,M,S} = \tilde \omega_{\tau_K,\tau_M}, \qquad \text{weakly, in probability,}
 \end{equation}
 where the probability measure $\tilde \omega_{\tau_K,\tau_M}$ is the (unique) maximizer of the functional
 \begin{multline}
  \tilde{\mathfrak E}(\mu)
  =\frac{1}{2}\iint_{x\ne y} \ln|x-y|\, \mu(\dd x) \mu(\dd y)
  \\+ \frac{1}{2}\int_0^1\left[ \frac{\tau_M^{-1}-\tau_K^{-1}}{\tau_K^{-1}} \ln(x)+ \frac{1-\tau_K^{-1}-\tau_M^{-1}}{\tau_K^{-1}}\ln(1-x)\right] \mu(\dd x).
 \end{multline}
 over the space of all probability measures $\mu$ on the interval $[0,1]$.
\end{lemma}
\begin{proof}[Sketch of the proof] By Theorem \ref{Theorem_CCA_Jacobi}, $\mu_{K,M,S}$ is the empirical measure for the eigenvalues of $\mathcal J(K; \frac{M-K+1}{2},\frac{S-K-M+1}{2})$, i.e., it fits into the setting of Lemma \ref{Lemma_variational_form} with $p=\frac{M-K+1}{2}$ and $q=\frac{S-K-M+1}{2}$.

The first step is to prove that $\tilde{\mathfrak E}(\mu)$ has a unique maximizer --- this is based on the convexity ideas, as we deal with a quadratic functional. Next, we notice that for large $K,M,S$, the functionals $\tilde{\mathfrak E}(\mu)$ and ${\mathfrak E}(\mu)$ are almost the same. Hence, if $\mu^\x$ is close to the maximizer $\tilde \omega_{\tau_K,\tau_M}$, then $\mathfrak E (\mu^\x)$ is close to $\tilde{\mathfrak E}(\tilde \omega_{\tau_K,\tau_M})$; on the other hand, if $\mu^\x$ is away from $\tilde \omega_{\tau_K,\tau_M}$, then $\mathfrak E(\mu^\x)$ is smaller than $\tilde{\mathfrak E}(\tilde \omega_{\tau_K,\tau_M})$. Thefefore, \eqref{eq_Var_principle} implies that the probability of such $\x$ is exponentially small. Further details of this proof can be found in \cite{Benarous1997large} or \cite{guionnet2019asymptotics}.
\end{proof}

\subsection{Identification of the limit}

Our next step is to solve the variational problem ${\tilde{\mathfrak E}(\mu)\to \max}$ and show that the maximizer $\tilde \omega_{\tau_K,\tau_M}$ in Lemma \ref{Lemma_LLN_var} is the Wachter distribution $\omega_{\tau_K,\tau_M}$. We do this by writing down an algebraic equation which uniquely determins $\tilde \omega_{\tau_K,\tau_M}$ and then solve it. It is possible to deduce this equation by directly manipulating the functional $\tilde{\mathfrak E}(\mu)$ (such computations are often done in the \emph{logarithmic potential theory}, cf.\ \cite{saff2013logarithmic}). We proceed in a slightly different way and instead derive exact equations for the prelimit object --- Jacobi ensemble $\mathcal J(K,p,q)$. The equation for $\tilde{\mathfrak E}(\mu)$ is then obtained by taking a limit.

\begin{theorem} \label{Theorem_DS_equation}
 Suppose that $\x=(x_1,\dots,x_K)$ with $1>x_1>\dots>x_K>0$ is distributed as Jacobi ensemble $\mathcal J(K,p,q)$ of Definition \ref{Definition_Jacobi_ev} with $p>1$, $q>1$, and let $\mu=\frac{1}{K}\sum_{i=1}^K \delta_{x_i}$ be the corresponding empirical measure. Then for any continuously differentiable function $f(x)$, $x\in[0,1]$, we have:
 \begin{multline}
 \label{eq_DS_equation}
  \E \left( \int_{0}^1 \left[f(x)\left(\frac{p-1}{K x}+\frac{q-1}{K(x-1)}\right)\right]\mu(\dd x)  +\frac{1}{2} \int\int_{[0,1]^2}  \frac{f(x)-f(y)}{x-y} \mu(\dd x) \mu(\dd y) \right)\\=-\frac{1}{2K} \E \left( \int_{0}^1 f'(x)\mu(\dd x)\right),
 \end{multline}
 where $ \frac{f(x)-f(y)}{x-y}$ is understood as $f'(x)$ when $x=y$.
\end{theorem}
\begin{proof} The idea is to integrate by parts the expression for $\E f'(x_k)$. Omitting the normalization constant of the measure and noting that all the boundary terms vanish because of $p>1$, $q>1$, this is
\begin{multline}
 \iint\limits_{1>x_1>\dots>x_K>0} f'(x_k) \prod_{1\le i<j\le K} (x_i-x_j) \prod_{i=1}^K x_i^{p-1} (1-x_i)^{q-1} \dd x_i
 \\=-\iint\limits_{1>x_1>\dots>x_K>0} f(x_k)\left(\frac{p-1}{x_k}+\frac{q-1}{x_k-1}+\sum_{i\ne k} \frac{1}{x_k-x_i}\right) \prod_{1\le i<j\le K} (x_i-x_j) \prod_{i=1}^K x_i^{p-1} (1-x_i)^{q-1} \dd x_i.
\end{multline}
Summing over all $k=1,\dots,K$ and converting to the expectations, we get:
\begin{equation}
\label{eq_x7}
 \E_\x\left[\sum_{k=1}^K f'(x_k)\right]
 =-\E_\x\left[\sum_{k=1}^K f(x_k)\left(\frac{p-1}{x_k}+\frac{q-1}{x_k-1}\right)\right]-\E_\x \left[\sum_{k=1}^K \sum_{i\ne k} \frac{f(x_k)}{x_k-x_i}\right].
\end{equation}
We further convert the last double sum, using the notation $\frac{f(x)-f(y)}{x-y}\Bigr|_{x=y}=f'(x)$:
$$
 \sum_{k=1}^K \sum_{i\ne k} \frac{f(x_k)}{x_k-x_i}=\frac{1}{2} \sum_{1\le i\ne k \le K} \frac{f(x_k)-f(x_i)}{x_k-x_i}=\frac{1}{2} \sum_{i,k=1}^{K} \frac{f(x_k)-f(x_i)}{x_k-x_i} - \frac{1}{2}\sum_{k=1}^K f'(x_k).
$$
Therefore, \eqref{eq_x7} gets transformed to the desired form of \eqref{eq_DS_equation}:
\begin{multline}
\label{eq_x8}
 \frac{1}{2K }\E_\x\left[\frac{1}{K}\sum_{k=1}^K f'(x_k)\right]
 \\=-\E_\x\left[\frac{1}{K}\sum_{k=1}^K f(x_k)\left(\frac{p-1}{K x_k}+\frac{q-1}{K(x_k-1)}\right)\right]-\E_\x \left[\frac{1}{2 K^2}\sum_{i,k=1}^K  \frac{f(x_k)-f(x_i)}{x_k-x_i}\right].\qedhere
\end{multline}
\end{proof}
Equation \eqref{eq_DS_equation} is an instance of the Dyson--Schwinger or loop equation. It can be made a part of a larger hierarchy of similar equations, which uniquely determine the distribution of $\mathcal J(K,p,q)$. Replacing a complicated distribution with a simpler equations for expectations with respect to this distribution is a powerful approach, as we demonstrate on the example of the Gaussian distribution in Problem set 3 in Chapter \ref{Section_exercises}.

\begin{corollary}
 Assume $\tau_K>\tau_M$ and let $\tilde \omega_{\tau_K,\tau_M}$ be the limit from Lemma \ref{Lemma_LLN_var}. Then for any continuously differentiable function $f(x)$, $x\in[0,1]$, we have
 \begin{multline}
 \label{eq_Wachter_density_equation}
   \int_{0}^1 \left[f(x)\left(\frac{\tau_M^{-1}-\tau_K^{-1}}{\tau_K^{-1}}\cdot \frac{1}{x}+\frac{1-\tau_K^{-1}-\tau_M^{-1}}{\tau_K^{-1}} \cdot \frac{1}{x-1}\right)\right]\tilde \omega_{\tau_K,\tau_M}(\dd x)  \\ + \int\int_{[0,1]^2}  \frac{f(x)-f(y)}{x-y} \tilde \omega_{\tau_K,\tau_M}(\dd x) \tilde \omega_{\tau_K,\tau_M}(\dd y)=0.
 \end{multline}
\end{corollary}
\begin{proof}
 We plug $p=\frac{M-K+1}{2}$, $q=\frac{S-K-M+1}{2}$ into \eqref{eq_DS_equation} and then send $K,M,S\to\infty$ using \eqref{eq_CCA_LLN_2}. Certain care is required near $x=0$ and near $x=1$, as the integrand might have a singularity there. We omit this justification, noting that eventually $0$ and $1$ lie outside the support of $\tilde \omega_{\tau_K,\tau_M}$, a fact we will use without proof. It can be derived by directly analyzing the form of the functional $\tilde{\mathfrak E}(\mu)$).
\end{proof}

One can find an explicit formula for $\tilde \omega_{\tau_K,\tau_M}(\dd x)$ from \eqref{eq_Wachter_density_equation} by using the Stieltes tranform of the measure, defined as
$$
 G(z)=\int_0^1 \frac{1}{z-x} \tilde \omega_{\tau_K,\tau_M}(\dd x),\qquad z\in \mathbb C\setminus \text{support}(\tilde \omega_{\tau_K,\tau_M}).
$$

\begin{theorem} $G(z)$ satisfies a quadratic equation:
 \begin{equation}
 \label{eq_Wachter_density_Stieltjes_equation}
   \frac{1}{z(z-1)}\left(\frac{\tau_K^{-1}-1}{\tau_K^{-1}} \right) +   \left(\frac{\tau_M^{-1}-\tau_K^{-1}}{\tau_K^{-1}}\cdot \frac{1}{z}+\frac{1-\tau_K^{-1}-\tau_M^{-1}}{\tau_K^{-1}} \cdot \frac{1}{z-1}\right) G(z)   + \bigl[G(z)\bigr]^2=0.
 \end{equation}
\end{theorem}
\begin{proof} We plug $f(x)=\frac{1}{z-x}$ into \eqref{eq_Wachter_density_equation} and use the identities:
$$
 \frac{1}{z-x} \cdot \frac{1}{x}= \frac{1}{z} \left(\frac{1}{z-x} + \frac{1}{x}\right), \qquad \frac{1}{z-x} \cdot \frac{1}{x-1} = \frac{1}{z-1} \left( \frac{1}{z-x}+\frac{1}{x-1}\right),
$$
$$
\frac{\frac{1}{z-x}- \frac{1}{z-y}}{x-y}=\frac{1}{z-x} \cdot \frac{1}{z-y}.
$$
We obtain
 \begin{multline}
 \label{eq_Wachter_density_equation_modified}
   -\frac{\tau_M^{-1}-\tau_K^{-1}}{\tau_K^{-1}}\cdot \frac{G(0)}{z}-\frac{1-\tau_K^{-1}-\tau_M^{-1}}{\tau_K^{-1}} \cdot \frac{G(1)}{z-1}\\ +  \left(\frac{\tau_M^{-1}-\tau_K^{-1}}{\tau_K^{-1}}\cdot \frac{1}{z}+\frac{1-\tau_K^{-1}-\tau_M^{-1}}{\tau_K^{-1}} \cdot \frac{1}{z-1}\right) G(z)   + \bigl[G(z)\bigr]^2=0.
 \end{multline}
So far $G(0)$ and $G(1)$ are not known to us. However, we can identify them by noting that $G(z)$ has the asymptotic expansion at $\infty$ of the form $G(z)=\frac{1}{z}+ \frac{a_2}{z^2}+\frac{a_3}{z^3}+\dots$. Plugging into \eqref{eq_Wachter_density_equation_modified}, we get
 \begin{multline}
 \label{eq_Wachter_density_equation_modified_2}
   -\frac{\tau_M^{-1}-\tau_K^{-1}}{\tau_K^{-1}}\cdot \frac{G(0)}{z}-\frac{1-\tau_K^{-1}-\tau_M^{-1}}{\tau_K^{-1}} \cdot \frac{G(1)}{z-1}\\ +  \left(\frac{\tau_M^{-1}-\tau_K^{-1}}{\tau_K^{-1}}\cdot \frac{1}{z^2}+\frac{1-\tau_K^{-1}-\tau_M^{-1}}{\tau_K^{-1}} \cdot \frac{1}{z(z-1)}\right) +\frac{1}{z^2} =o\left(\frac{1}{z^2}\right).
 \end{multline}
Expanding all terms in $\frac{1}{z}$, we conclude
 \begin{multline}
 \label{eq_Wachter_density_equation_modified_2}
   -\frac{\tau_M^{-1}-\tau_K^{-1}}{\tau_K^{-1}}\cdot \frac{G(0)}{z}-\frac{1-\tau_K^{-1}-\tau_M^{-1}}{\tau_K^{-1}} \cdot G(1)\left(\frac{1}{z}+\frac{1}{z^2}\right)\\ +  \left(\frac{\tau_M^{-1}-\tau_K^{-1}}{\tau_K^{-1}}\cdot \frac{1}{z^2}+\frac{1-\tau_K^{-1}-\tau_M^{-1}}{\tau_K^{-1}} \cdot \frac{1}{z^2}\right) +\frac{1}{z^2} =0.
 \end{multline}
The coefficients of $\frac{1}{z}$ and $\frac{1}{z^2}$ give two linear equations on $G(0)$ and $G(1)$:
$$
 \begin{cases} -\frac{\tau_M^{-1}-\tau_K^{-1}}{\tau_K^{-1}}\cdot G(0)-\frac{1-\tau_K^{-1}-\tau_M^{-1}}{\tau_K^{-1}} \cdot G(1)=0,\\
 -\frac{1-\tau_K^{-1}-\tau_M^{-1}}{\tau_K^{-1}} \cdot G(1) +  \frac{\tau_M^{-1}-\tau_K^{-1}}{\tau_K^{-1}}+\frac{1-\tau_K^{-1}-\tau_M^{-1}}{\tau_K^{-1}} +1 =0. \end{cases}
$$
Solving for $G(0)$, $G(1)$, and then plugging into \eqref{eq_Wachter_density_equation_modified}, we get \eqref{eq_Wachter_density_Stieltjes_equation}.
\end{proof}

Directly solving equation \eqref{eq_Wachter_density_equation} for $G(z)$, we get the explicit formula:
\begin{corollary} \label{Corollary_Wachter_Stieltjes} Using the $\lambda_\pm$ notation \eqref{eq_lambda_pm_def}, we have:
 \begin{equation}
\label{eq_Wachter_Stieltjes}
G(z)= \frac{\tau_M^{-1}+\tau_K^{-1}-z+ \sqrt{ (z-\lambda_-)(z-\lambda_+)}}{2 \tau_K^{-1} z(z-1)} + \frac{1}{z}.
\end{equation}
\end{corollary}
\begin{remark}
 For the function $\sqrt{ (z-\lambda_-)(z-\lambda_+)}$, we should choose a continuous branch in $\mathbb C\setminus [\lambda_-,\lambda_+]$, such that $\sqrt{ (z-\lambda_-)(z-\lambda_+)}\sim z$ for large $z$. This (rather than the one with minus sign) branch is uniquely determined by the $G(z)\sim \frac{1}{z}$ condition.
\end{remark}

We are now ready to finish the proof of Theorem \ref{Theorem_CCA_LLN}.
\begin{theorem} \label{Theorem_Jacobi_LLN_final}
 For any $\tau_K>\tau_M$ the measure $\tilde \omega_{\tau_K,\tau_M}$ from Lemma \ref{Lemma_LLN_var} coincides with $\omega_{\tau_K,\tau_M}$, i.e.,
 $$\tilde \omega_{\tau_K,\tau_M}=\frac{\tau_K}{2\pi } \frac{\sqrt{(x-\lambda_-)(\lambda_+-x)}}{x (1-x)} \mathbf 1_{[\lambda_-,\lambda_+]}\, \dd x.$$
\end{theorem}
\begin{proof}
 We already know the Stieltjes transform of $\tilde \omega_{\tau_K,\tau_M}$  and it remains to connect it to the density of the measure. For that we notice that for $\eps>0$, we have
 $$
  \frac{1}{\pi} \mathrm{Im}\, G(x_0-\ii \eps)=\int_0^1 \frac1{\pi}\mathrm{Im}\left[\frac{1}{x_0-x-\ii\eps}\right] \tilde \omega_{\tau_K,\tau_M}(\dd x)= \int_0^1 \left[\frac{1}{\pi}\, \frac{\eps }{(x_0-x)^2+\eps^2}\right] \tilde \omega_{\tau_K,\tau_M}(\dd x).
 $$
 The function $\frac{1}{\pi}\, \frac{\eps }{(x_0-x)^2+\eps^2}$ is positive, has total integral $1$ and sharply concentrates near $x_0$ as $\eps\to 0+$. Hence, its integral picks up the density of $\tilde \omega_{\tau_K,\tau_M}$:
 $$
  \lim_{\eps\to 0+} \frac{1}{\pi} \mathrm{Im}\, G(x_0-\ii \eps)=  \tilde \omega_{\tau_K,\tau_M}(x_0).
 $$
 Plugging the expression for $G(z)$ from \eqref{eq_Wachter_Stieltjes}, we get the desired formula for the density.
\end{proof}

The case $\tau_K=\tau_M$ in Theorem \ref{Theorem_Jacobi_LLN_final} can be treated similarly, but one should be slightly more careful when handling the $x=0$ point: in this case $\lambda_-=0$ and also we might have $p=1$ in Theorem \ref{Theorem_DS_equation}, so that a new boundary term appears in integration by parts.

\newpage

\section{Signal plus noise}

In the previous chapter we studied the asymptotic properties of the sample canonical correlations in the situation when no signal exists: two data sets are independent, so all canonical correlations are zeros. I.e., there was no information to be inferred about the correlation structure. Nevertheless, the asymptotic behavior turned out to be quite rich. In this chapter we make the next step and allow for some signal in the data. That is, we discuss the situation when some of the canonical correlations $c_i$ are not zeros.

\subsection{Finite $K$ and $M$} \label{Section_finite_K_M_limit} As previously, we start by assuming that $K$ and $M$ are fixed and small, while $S$ is large. I.e., we work with matrices $\U$ and $\V$ that are composed of many i.i.d.\ samples of low-dimensional vectors $\u$ and $\v$. We start by analyzing canonical correlations and then move to canonical variables.

\begin{theorem} \label{Theorem_CCA_consistency}Suppose that the dimensions $K$ and $M$ are fixed, and consider two mean-zero random vectors\footnote{These vectors need not be Gaussian.} $\u=(u^1,\dots,u^K)^\T$ and $\v=(u^1,\dots,u^M)^\T$. Suppose that the covariance matrices $\E \u\u^\T$ and $\E \v \v^\T$ are non-degenerate and let $1\ge c_1\ge c_2\ge \dots\ge c_K\ge 0$ be the canonical correlations between $\u$ and $\v$. Suppose that $\U$ and $\V$ are matrices, whose $S$ columns are i.i.d.\ samples from $(\u^\T,\v^\T)^\T$ and let $1\ge \hat c_1(S) \ge \hat c_2(S)\ge\dots\ge c_K(S)\ge 0$ be sample canonical correlations of $\U$ and $\V$. Then,
\begin{equation} \label{eq_consistency_CCA}
 \lim_{S\to\infty} \hat c_i(S)= c_i, \qquad \text{in probability, for all }1\le i \le K.
\end{equation}
\end{theorem}
\begin{remark}
 The differences $\hat c_i(S)-c_i$ are of order $S^{-1/2}$ and their rescaled limit depends on whether some of the $c_i$ are equal to each other and on whether they are equal to zero or one, as first noticed in \cite{hsu1941limiting}. The simplest case is when all $c_i$ are positive, less than one, and distinct --- in this situation $\{S^{1/2}(\hat c_i(S)-c_i)\}_{i=1}^{K}$ are asymptotically Gaussian and independent over $1\le i\le K$. Note that there is no contradiction with $S^{-1}$ magnitude of the fluctuations in Theorem \ref{Theorem_CCA_to_Laguerre}, because that theorem deals with $\hat c_i^2$, rather than $\hat c_i$, which changes the scaling in the situation $c_i=0$ (and does not change if $c_i>0$).
\end{remark}
From the statistics perspective, Theorem \ref{Theorem_CCA_consistency} tells that $\hat c_i(S)$ are consistent estimates of $c_i$: they approximate the true (population) values of the parameters of the model as $S\to\infty$.

\begin{proof}[Proof of Theorem \ref{Theorem_CCA_consistency}] Recalling Procedure \ref{Procedure_CCA}, $\hat c_i^2(S)$ are eigenvalues of the $K\times K$ matrix $(\U \U^\T)^{-1} \U \V^\T (\V \V^\T)^{-1} \V \U^\T$. Note that the dimensions of the four matrices
$$
 \frac{1}{S}\U \U^\T, \qquad \frac{1}{S}  \U \V^\T, \qquad \frac{1}{S} \V \V^\T,\qquad \frac{1}{S} \V \U^\T
$$
do not grow with $S$ and each matrix element is $\frac{1}{S}$ times a sum of $S$ i.i.d.\ random variables. Hence, by the Law of Large Numbers, the matrix elements converge in probability to their expectations as $S\to\infty$. We conclude that
\begin{multline}\label{eq_x9}
 \lim_{S\to\infty} \left[ (\U \U^\T)^{-1} \U \V^\T (\V \V^\T)^{-1} \V \U^\T\right]=\lim_{S\to\infty} \left[ \left(\frac{\U \U^\T}{S}\right)^{-1} \frac{\U \V^\T}{S} \left(\frac{\V \V^\T}{S}\right)^{-1} \frac{\V \U^\T}{S}\right]\\= (\E \u \u^\T)^{-1} \E \u \v^\T (\E \v\v^\T)^{-1} \E \v \u^{\T}.
\end{multline}
Convergence of matrix elements of finite-dimensional matrices  implies convergence of their eigenvalues, see, e.g.,  \citet[Chapter 9, Theorem 6]{Lax} or  \citet[Corollary 6.3.8]{horn2012matrix}. The eigenvalues of the last matrix formed by expectation are precisely $c_1^2,\dots,c_K^2$.
\end{proof}

We now turn to canonical variables and seek an appropriate metric to assess how closely the variables derived from the ``sample'' $\U$ and $\V$ approximate those from the ``population'' $\u$ and $\v$. In the finite $K$ and $M$ setting, all proximity metrics are somewhat equivalent; however, selecting the right metric will become crucial in the next subsection, where we allow for large $K$ and $M$.

Let $c_i$ be the $i$--th largest canonical correlation between $\u$ and $\v$ and let $\ba=(\alpha_1,\alpha_2,\dots,\alpha_K)^\T$ and $\bb=(\beta_1,\beta_2,\dots,\beta_M)^\T$ be the corresponding canonical vectors, i.e., eigenvectors of the matrices
$$
 (\E \u \u^\T)^{-1} \E \u \v^\T (\E \v\v^\T)^{-1} \E \v \u^{\T}\quad \text{ and }\quad (\E \v \v^\T)^{-1} \E \v \u^\T (\E \u\u^\T)^{-1} \E \u \v^{\T}, \qquad \text{respectively.}
$$
As shown in Proposition \ref{prop_CCA_matrix_form}, the random variables $\ba^{\T}\u$ and $\bb^\T\v$ are canonical variables, corresponding to $c_i$. An advantage of looking at canonical vectors $\ba$ and $\bb$ instead of the canonical variables is that $\ba$ and $\bb$ are non-random and, therefore, can be directly compared with their sample versions $\hat \ba (S)$ and $\hat \bb(S)$, where $\hat \ba (S)$ and $\hat \bb(S)$ are the eigenvectors of
$$
(\U \U^\T)^{-1} \U \V^\T (\V \V^\T)^{-1} \V \U^\T \quad \text{ and }\quad (\V \V^\T)^{-1} \V \U^\T (\U \U^\T)^{-1} \U \V^\T, \qquad \text{respectively,}
$$
corresponding to the $i$th largest eigenvalue or squared sample canonical correlation $\hat c_i^2(S)$.

\begin{definition} \label{Definition_angles}
 The angle $0\le \theta^\u_i\le \pi/2$ is defined as the angle between $S$--dimensional vectors $\U^\T \ba$ and $\U^\T \widehat \ba$. Similarly, $\theta^\v_i$ is the angle between $\V^{\T}\bb$ and $\V^\T\widehat \bb$. We have
 $$
  \sin^2 \theta^\u_i=1- \frac{\la \U^\T \ba, \U^\T \widehat \ba\ra^2}{\la \U^\T \ba, \U^\T \ba\ra \la \U^\T \widehat \ba, \U^\T \widehat \ba\ra},\qquad
  \sin^2 \theta^\v_i=1- \frac{\la \V^\T \bb, \V^\T \widehat \bb\ra^2}{\la \V^\T \bb, \U^\T \bb\ra \la \V^\T \widehat \bb, \V^\T \widehat \bb\ra}.
 $$
\end{definition}
Let us emphasize two features of Definition \ref{Definition_angles}. First, the eigenvectors are only defined up to multiplication by an arbitrary real constant (or up to multiplication by $-1$ if we normalize them), hence, it is tricky to compare them directly; the use of angles avoids this problem.\footnote{The limits of individual coordinates of $\ba$ and $\bb$ are also discussed in the literature, see \cite{anderson1999asymptotic} and references therein.} Second, instead of measuring the angles between $\widehat \ba$ and $\ba$ (or $\widehat \bb$ and $\bb$), we use the angles between  $\U^\T \ba$ and $\U^\T \widehat \ba$ (or between $\V^{\T}\bb$ and $\V^\T\widehat \bb$). The reason behind this is that the former angles depend on the choice of units of measurement: the angle changes if we multiply one of the coordinates of $\ba$ by a constant (equivalently, divide a component of $\u$ or a row in $\U$ by the same constant). In practice, there is generally no default normalization of real data, and this dependence on measurement units implies that, asymptotically, the angles are influenced not only by canonical correlations but also by the covariances of $\u$ and $\v$. By focusing on  $\theta^\u_i$ and $ \theta^\v_i$, we eliminate the issue of normalization, ensuring that the results in the theorems of the next subsection are independent of the covariance matrix of the vector $\u$ or of the covariance matrix of the vector $\v$.


\begin{theorem}\label{Theorem_CCA_consistency_variables} In the setting of Theorem \ref{Theorem_CCA_consistency}, suppose that $c_i$ is distinct from $0$, $1$, and all other $c_j$'s. Then, using Definition \ref{Definition_angles}:
\begin{equation} \label{eq_consistency_CCA_angles}
 \lim_{S\to\infty}\theta^\u_i= \lim_{S\to\infty}\theta^\v_i=0 \quad \text{ in probability}.
\end{equation}
\end{theorem}
\begin{remark}\label{remark_nonuniqueness}
 The requirement for $c_i$ to be distinct from others is related to the uniqueness of canonical correlations discussed at the end of Chapter 1. Imagine that $c_i=c_{i+1}$. In this situation, there is a two-dimensional space of eigenvectors of $(\E \u \u^\T)^{-1} \E \u \v^\T (\E \v\v^\T)^{-1} \E \v \u^{\T}$ with eigenvalue $c_i^2$ and there is no natural way to single out two canonical variables out of this space. Hence, we can only discuss the distance between this space and the sample canonical variables; one can still show that a properly defined distance between such spaces tends to $0$ as $S\to\infty$, but we leave this out of the scope of the survey.
\end{remark}
\begin{proof}[Proof of Theorem \ref{Theorem_CCA_consistency_variables}] We follow the proof of Theorem \ref{Theorem_CCA_consistency} and use  \eqref{eq_x9}. The eigenvectors are found as solutions of a system of homogeneous linear equations. The matrix of this system is non-degenerate when the multiplicity of the corresponding eigenvalue is $1$, and the coefficients of the solution (defined up to simultaneous multiplication by a constant) smoothly depend on the matrix of the equation. The conclusion is that eigenvectors continuously depend on the matrix, cf.\ \citet[Chapter 9, Theorem 8]{Lax}. Hence, \eqref{eq_x9} implies convergence of eigenvectors, up to a proportionality. For the eigenvectors corresponding to the $i$--th largest eigenvalue this can be summarized as $\lim_{S\to\infty} \widehat \ba = q_1\cdot \ba$, $\lim_{S\to\infty} \widehat \bb = q_2\cdot \bb$ for deterministic constants $q_1,q_2\ne 0$. We plug the limit into Definition \ref{Definition_angles} and use the Law of Large Numbers:
\begin{multline*}
 \lim_{S\to\infty} \frac{1}{S} \la \U^\T \ba, \U^\T \widehat \ba\ra= \lim_{S\to\infty} \sum_{k,l=1}^K \alpha_k \widehat \alpha_l  \frac{1}{S}\bigl\la k\text{th column of } \U^\T ,\, l\text{th column of } \U^\T \bigr\ra\\= q_1 \sum_{k,l=1}^K \alpha_k \alpha_l \E \bigl[u^k u^l\bigr], \qquad \text{in probability.}
\end{multline*}
Repeating the same computation two more times, we get
$$
  \lim_{S\to\infty} \sin^2 \theta^\u_i=\lim_{S\to\infty}\left[1- \frac{\la \U^\T \ba, \U^\T \widehat \ba\ra^2}{\la \U^\T \ba, \U^\T \ba\ra \la \U^\T \widehat \ba, \U^\T \widehat \ba\ra}\right]= 1- \frac{(q_1)^2}{1\cdot (q_1)^2}=0.
$$
The argument for $\theta^\v_i$ is the same and we get \eqref{eq_consistency_CCA_angles}.
\end{proof}

\subsection{Linearly growing $K$ and $M$} \label{Section_growing_KM} For the case of growing $K$ and $M$ we introduce the ``signal plus noise'' or ``spiked random matrix'' model, with general terminology going back to \cite{johnstone2001distribution}. In this model, independent vectors $\u$ and $\v$ (or their sample version, arrays $\U$ and $\V$) with all canonical correlations $c_i=0$ represent pure noise, meaning no signal is present in the data. Conversely, when some $c_i$ are positive, they and their corresponding canonical variables are treated as signals. Our goal is to separate signal from noise, reconstructing signals from the observed noisy data.


The natural first step toward this goal is to assume that only few $c_i$ are positive, possibly just one. The results from Chapter 3 and, in particular, Figure \ref{Fig_Wachter_simulation} indicate that inferring a non-zero correlation from observed data is challenging: the presence of non-zero sample correlations is not sufficient to conclude that even a single $c_i$ is non-zero.


Let us assume that there is only one signal: $c_1=r>0$ and all other $c_i$ equal zero. Thus, $r^2$ equals the single nonzero eigenvalue of the $K\times K$ matrix
 $(\E \u \u^\T)^{-1} (\E\u \v^\T) (\E\v \v^\T)^{-1} (\E\v \u^\T)$ and the single nonzero eigenvalue of the ${M\times M}$ matrix $(\E\v \v^\T)^{-1} (\E\v \u^\T) (\E\u \u^\T)^{-1} (\E\u \v^\T)$. We let $\ba$ and $\bb$ be the corresponding eigenvectors of the former and the latter matrices, respectively. Note that  $(\u^\T\ba, \v^\T\bb, r)$ is the single non-trivial triplet of canonical variables and correlations for $\u$ and $\v$.

We further assume that as $S\to\infty$, $r=r(S)$ tends to a constant $\rho>0$. The reason for distinguishing between $r$ and $\rho$ is that the vectors $\u$ and $\v$ are now necessarily $S$--dependent, because their dimensions $K$ and $M$ grow with $S$, thus, their correlation structure, captured by $r$, is allowed to change with $S$.

We recall the definition of $\lambda_+$ from \eqref{eq_lambda_pm_def} and introduce three additional functions of $\rho$:


\begin{align}
\label{eq_lambda_plus} \lambda_+&=\left(\sqrt{\tau_M^{-1}(1-\tau_K^{-1})}+ \sqrt{\tau_K^{-1}(1-\tau_M^{-1})}  \right)^2, \\
\label{eq_zrho}  z_\rho&=\frac{\bigl(  (\tau_K-1)\rho^2  + 1 \bigr) \bigl(  (\tau_M-1) \rho^2 + 1\bigr)}{\rho^2 \tau_K \tau_M },\\
\label{eq_sx} \s_\u&= \frac{(1-\rho^2)(\tau_K-1)}{(\tau_M-1)(\tau_K-1) \rho^2 - 1 } \cdot \frac{ (\tau_M-1) \rho^2 + 1 }{(\tau_K-1) \rho^2 + 1},\\
\label{eq_sy} \s_\v&= \frac{(1-\rho^2)(\tau_M-1)}{(\tau_M-1)(\tau_K-1) \rho^2 - 1 } \cdot \frac{ (\tau_K-1) \rho^2 + 1 }{(\tau_M-1) \rho^2 + 1}.
\end{align}

\begin{theorem} \label{Theorem_CCA_one_spike}
 In the one signal situation outlined above, assume that $\u$ and $\v$ are jointly Gaussian and $\W=\begin{pmatrix} \U\\ \V\end{pmatrix}$ is a $(K+M)\times S$ matrix with i.i.d.\ columns distributed as $\begin{pmatrix} \u \\ \v\end{pmatrix}$. Let $S$ tend to infinity and $K\le M$ depend on it in such a way that the ratios $S/K$ and $S/M$ converge to $\tau_K>1$ and $\tau_M>1$, respectively, and $\tau_M^{-1}+\tau_K^{-1}<1$. Simultaneously, suppose that $\lim_{S\to\infty} r^2=\rho^2$. Let $\hat c_1\ge \hat c_2\ge \dots$ be the canonical correlations between $\U$ and $\V$ and let $\theta^\u$, $\theta^\v$ be the angles corresponding to $\hat c_1$, as in Definition \ref{Definition_angles} with $i=1$. Then:
 \begin{enumerate}
  \item[\bf 1.] If $\rho^2> \tfrac{1}{\sqrt{(\tau_M-1)(\tau_K-1)}}$, then $z_\rho>\lambda_+$, and
  \begin{equation}
  \label{eq_canonical_limit}
   \lim_{S\to\infty} \hat c_1^2=z_\rho, \qquad \lim_{S\to\infty} \hat c_2^2=\lambda_+  \qquad \text{in probability},
  \end{equation}
  \begin{equation}
  \label{eq_vector_limit1}
   \lim_{S\to\infty} \sin^2\theta^\u=\s_\u \qquad \text{in probability},
  \end{equation}
  \begin{equation}
  \label{eq_vector_limit2}
   \lim_{S\to\infty} \sin^2\theta^\v=\s_\v \qquad \text{in probability}.
  \end{equation}
  \item[\bf 2.] If instead $\rho^2\le \tfrac{1}{\sqrt{(\tau_M-1)(\tau_K-1)}}$, then
  \begin{equation}
  \label{eq_no_spike}
   \lim_{S\to\infty} \hat c_1^2=\lambda_+ \qquad \text{in probability}.
  \end{equation}
 \end{enumerate}
\end{theorem}
\begin{remark} \label{Remark_Wachter_perturbation}
 One can show that the matrix $(\U \U^\T)^{-1} \U \V^\T (\V \V^\T)^{-1} \V \U^\T$ is a small rank perturbation of a similar matrix of the pure noise setting studied in the previous two chapters, see, e.g., \cite[Lemma A.6]{BG3}. Hence, the convergence to Wachter distribution of Theorem \ref{Theorem_CCA_LLN} continues to hold. Therefore, one can not distinguish the presence of one positive $c_i$ by looking at the limit of the empirical measures.
\end{remark}
The part of Theorem \ref{Theorem_CCA_one_spike} about the limit of $\hat{c}_i^2$ is proved in \cite{bao2019canonical} and is further extended to non-Gaussian $\U$ and $\V$ in \cite{yang2022limiting}; another approach to the Gaussian case can be found in \cite{hou2023spiked}.
These articles also show that in case $\mathbf{1}$ the rescaled difference $S^{1/2}(\hat c_1^2-\lambda_+)$ is asymptotically Gaussian with explicit variance. The part about the limit of the angles $\theta^\u,\,\theta^\v$ is proved in \cite{BG3}.

\begin{figure}[t]
\begin{subfigure}{.49\textwidth}
  \centering
  \includegraphics[width=1.0\linewidth]{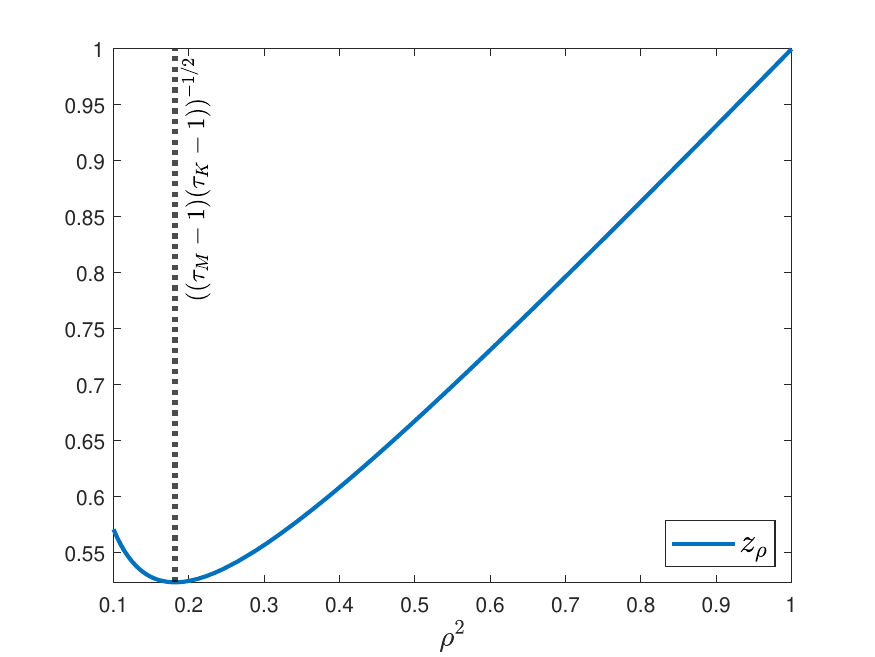}
  \caption{$z_{\rho}$ as a function of $\rho^2$.}
  \label{rho2_zrho}
\end{subfigure}%
\begin{subfigure}{.49\textwidth}
  \centering
  \includegraphics[width=1.0\linewidth]{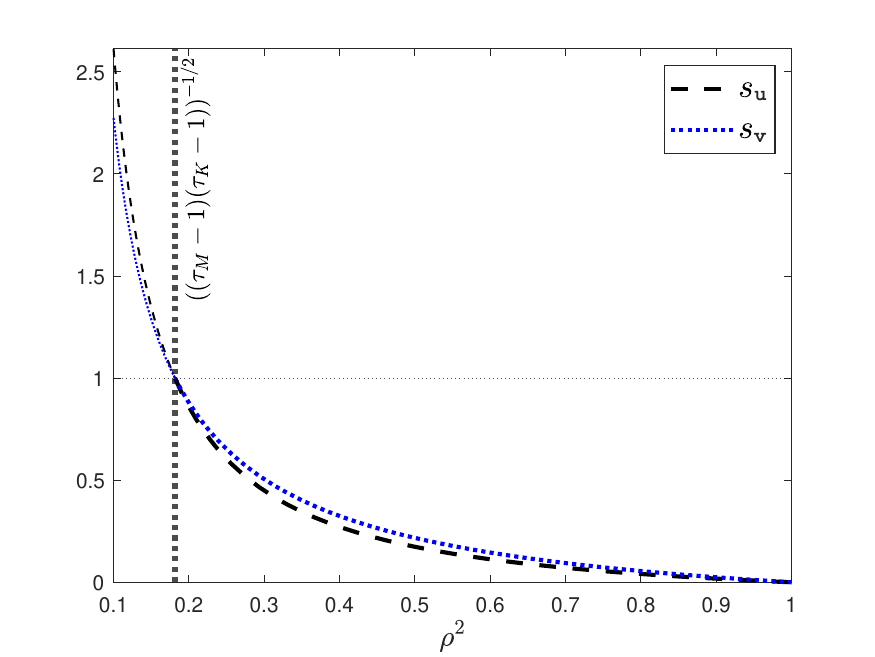}
  \caption{$\s_\u$ and $\s_\v$ as functions of $\rho^2$.}
  \label{rho2_sxsy}
\end{subfigure}
\caption{Functions in \eqref{eq_zrho}, \eqref{eq_sx}, \eqref{eq_sy} for $K=1000,\,M=1500,\,S=8000$.}
\label{fig_rho2_depend}
\end{figure}

\begin{figure}[t]
    \centering
        \includegraphics[width=0.6\textwidth]{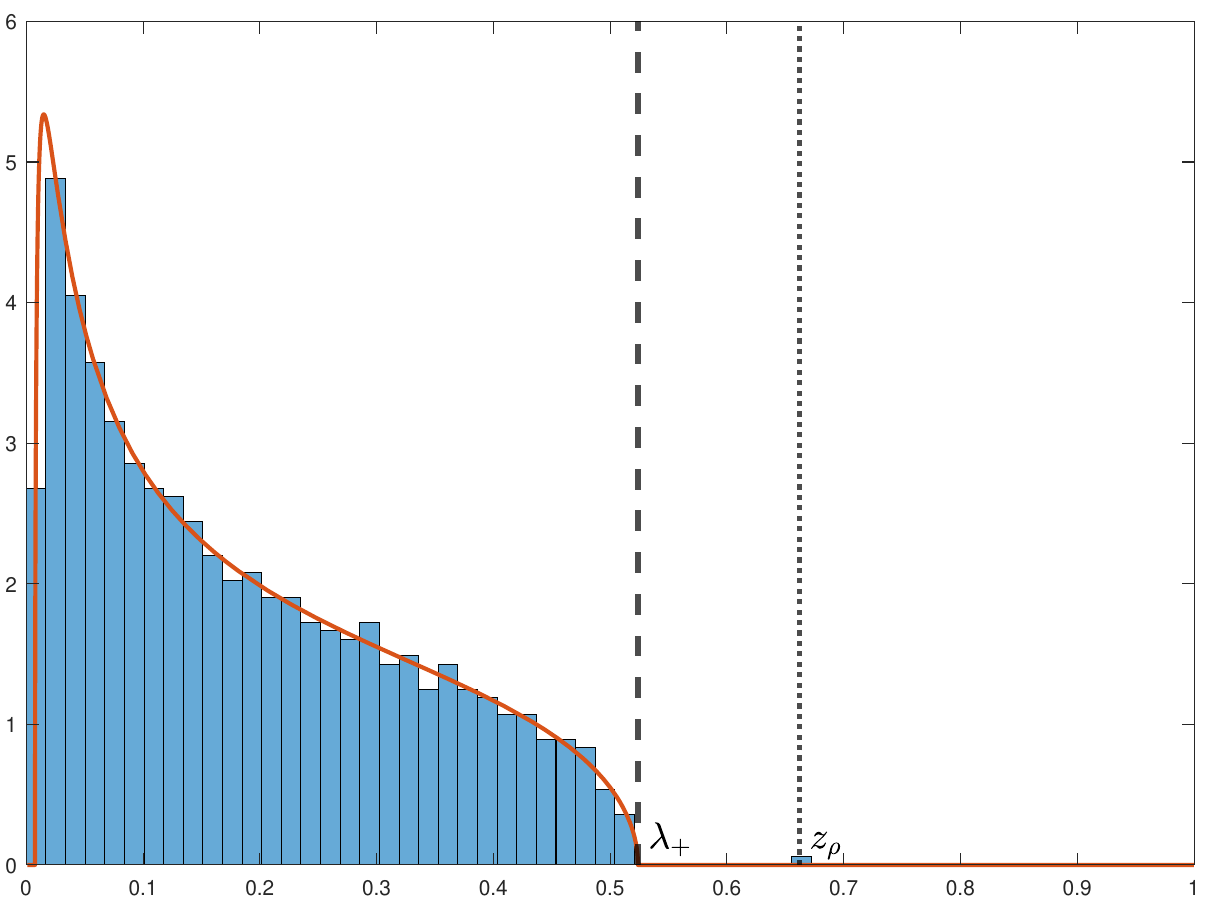}
    \caption{Histogram of simulated squared sample canonical correlations (blue columns) for a model with 1 signal, $K=1000$, $M=1500$, $S=8000$, $r^2=0.49$. A single spike is located approximately at $z_\rho$. The density of the Wachter distribution is shown in orange.} \label{Fig_spike}
\end{figure}

\begin{figure}[t]
    \centering
        \includegraphics[width=0.4\textwidth]{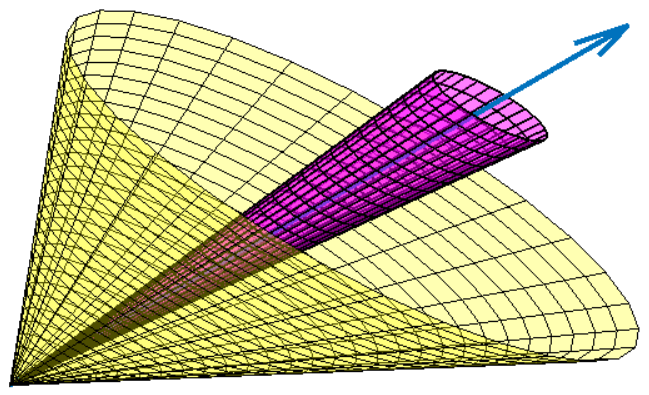}
    \caption{The estimated canonical variable lies on a cone whose axis aligns with the true direction, shown by the blue arrow. If $\sin^2\theta$ is small, then the cone is narrow, as shown in purple; if $\sin^2\theta$ is large, then the cone is wide, as shown in yellow. \label{Fig_cone}}
\end{figure}

Figure \ref{fig_rho2_depend} plots the dependence of $z_\rho$, $\s_u$, and $\s_\v$ on $\rho^2$ for a particular choice of $\tau_K$ and $\tau_M$. Analyzing the formulas \eqref{eq_zrho}, \eqref{eq_sx}, and \eqref{eq_sy} one can deduce many features of this dependence. In particular, for $\rho^2<1$, we get $z_\rho>\rho^2$, so that the largest sample canonical correlation always overestimates the true correlation. In addition, $\s_\u>0$ and $\s_\v>0$, which means that the estimates for the canonical variables are not consistent, but rather inclined by certain angles toward the desired true direction, as illustrated in Figure \ref{Fig_cone}. However, whenever either $\tau_K,\tau_M\to\infty$ or $\rho^2\to 1$, both $z_\rho$ approaches $\rho^2$ and the angles tend to $0$. Therefore, the consistency is restored in such a limit, matching the result of Theorems \ref{Theorem_CCA_consistency} and \ref{Theorem_CCA_consistency_variables}.

When we are working with real data and have no knowledge of the true value of $\rho$, the results of Theorem \ref{Theorem_CCA_one_spike} should be applied in the following way: if  the model matches the data, then most of the squared canonical correlations should belong to the $[\lambda_-,\lambda_+]$ interval. Further, if there is a gap between the largest sample canonical correlation $\lambda_1$ and $\lambda_+$, as in Figure \ref{Fig_spike}, then using Eq.~\eqref{eq_canonical_limit}, we take $\lambda_1$ as an approximation of $z_\rho$. Treating $z_\rho$ as known and approximating $\tau_K$, $\tau_M$ with $S/K$, $S/M$, Eq.~\eqref{eq_zrho} becomes a quadratic equation in $\rho^2$. Solving it and using $ \tfrac{1}{\sqrt{(\tau_M-1)(\tau_K-1)}}\le \rho^2\le 1$ to choose the correct root out of the two, we obtain an estimate for $\rho^2$. Further plugging into \eqref{eq_sx}, \eqref{eq_sy} and using  \eqref{eq_vector_limit1}, \eqref{eq_vector_limit2}, we obtain an estimate for the angles between sample and population canonical variables.

\subsection{Some ingredients of the proof of Theorem \ref{Theorem_CCA_one_spike}} We briefly outline an approach towards proving Theorem \ref{Theorem_CCA_one_spike} from \cite{BG3}. The idea is to develop CCA perturbation theory and then to produce an equation explaining how the canonical correlations and variables change when we add a vector to each of the subspaces.

Suppose that we are given two subspaces, $\widetilde \U$ and $\widetilde \V$ in $S$--dimensional space and let $\dim(\widetilde \V)={M-1}\ge {K-1}=\dim(\widetilde \U)$. In addition, we are given two vectors $\u^*$ and $\v^*$, that we add to spaces $\widetilde \U$ and $\widetilde \V$. Define
$$
 \U=\mathrm{span}(\u^*, \widetilde \U), \qquad \V=\mathrm{span}(\v^*, \widetilde \V).
$$
Both spaces $\widetilde \U$, $\widetilde \V$ and vectors $\u^*$, $\v^*$ are arbitrary, either deterministic or random, and no assumptions are made about them. Our task is to connect the canonical correlations and variables between $\U$ and $\V$ to those between $\widetilde \U$ and $\widetilde \V$.

Let $\{\widetilde \u_i\}_{i=1}^{K-1}$, $\{\widetilde \v_j\}_{j=1}^{M-1}$ and $\{\widetilde c_i\}_{i=1}^{K-1}$ be the canonical correlations and variables between $\widetilde \U$ and $\widetilde \V$, so that:
$$
  \langle \widetilde \u_i, \widetilde \u_j\rangle = \delta_{i=j}, \qquad \langle \widetilde \v_i, \widetilde \v_j\rangle = \delta_{i=j},\qquad \langle \widetilde \u_i, \widetilde \v_j\rangle = \widetilde c_i \delta_{i=j}.
$$
We also set $\widetilde c_j=0$ for $K\le j \le M-1$ for convenience of the notations.

\medskip

Now take two vectors $\widehat \ba=(\alpha_0,\alpha_1,\dots,\alpha_{K-1})$ and $\widehat \bb=(\beta_0,\beta_1,\dots,\beta_{M-1})$, such that $ \alpha_0\u^*+\sum_{i=1}^{K-1} \alpha_i \widetilde \u_i$ and $\beta_0 \v^*+\sum_{j=1}^{M-1} \beta_j \widetilde \v_j$ is a pair of canonical variables for $\U$ and $\V$ corresponding to a squared canonical correlation denoted $z$. We normalize the vectors so that
\begin{equation}
\label{eq_normalization}
 \left\|\alpha_0 \u^*+ \sum_{i=1}^{K-1} \alpha_i \widetilde \u_i\right\|^2=\left\|\beta_0 \v^*+ \sum_{j=1}^{M-1} \beta_j \widetilde \v_j\right\|^2=1.
\end{equation}
Then $z=\left\la \alpha_0 \u^*+ \sum_{i=1}^{K-1} \alpha_i \widetilde \u_i,\, \beta_0 \v^*+ \sum_{j=1}^{M-1} \beta_j \widetilde \v_j\right\ra^2$.

Note that previously, in Sections \ref{Section_finite_K_M_limit} and \ref{Section_growing_KM}, we wrote canonical variables in the bases of rows of the matrices $\U$ and $\V$; this is how canonical vectors were defined there. This time, the basis is slightly different: for $\U$ we are expanding the canonical variables as combinations of $\u^*$ and $\widetilde \u_1,\dots,\widetilde \u_{K-1}$; for $\V$ we are expanding the canonical variables as combinations of $\v^*$ and $\widetilde \v_1,\dots,\widetilde \v_{M-1}$. Up to this detail, the heuristic meaning of $\widehat \ba$ and $\widehat \bb$ is the same as before: these are still coefficients or coordinates of the expansion of canonical variables.

\begin{theorem} \label{Theorem_CCA_master_equation}
For each such triplet $\widehat \ba$, $\widehat \bb$, $z$, we have:
\begin{multline}
\label{eq_CCA_master}
 \left[\langle \u^*, \v^*\rangle + \sum_{j=1}^{M-1}  \frac{\langle\u^*, \widetilde \v_j\rangle(\widetilde c_j \langle \v^*, \widetilde \u_j\rangle -z \langle \v^*, \widetilde\v_j\rangle)}{z-\widetilde c_j^2} -z  \sum_{i=1}^{K-1}    \frac{\langle\u^*, \widetilde\u_i\rangle(\langle \widetilde\v^*, \widetilde\u_i\rangle - \widetilde c_i \langle \v^*, \widetilde\v_i\rangle)}{z - \widetilde c_i^2}   \right]^2\\= z \left[-\langle\u^*,\u^*\rangle+\sum_{j=1}^{M-1}\frac{  \langle \u^*, \widetilde \v_j \rangle^2 -2 \widetilde c_j  \langle\u^*, \widetilde\v_j\rangle  \langle \u^*, \widetilde\u_j\rangle }{z-\widetilde c_j^2} +z\sum_{i=1}^{K-1}   \frac{\langle \u^*, \widetilde\u_i\rangle^2}{z - \widetilde c_i^2} \right] \\ \times \left[-\langle\v^*,\v^*\rangle+\sum_{i=1}^{K-1}  \frac{\langle \v^*, \widetilde \u_i\rangle^2 - 2  \widetilde c_i \langle \v^*, \widetilde \u_i\rangle \langle \v^*, \widetilde \v_i\rangle}{z -\widetilde c_i^2} +z\sum_{j=1}^{M-1} \frac{\langle\v^*,\widetilde \v_j\rangle^2}{z - \widetilde c_j^2}\right].
\end{multline}
\end{theorem}
\begin{remark}
 Supplementing \eqref{eq_CCA_master} with additional equations, we can also completely determine $\widehat \ba$ and $\widehat \bb$, see \cite[Section A.2]{BG3} for the full details and proofs.
\end{remark}
\begin{proof}[Sketch of the proof of Theorem \ref{Theorem_CCA_master_equation}] The proof involves heavy computations, and a toy version (based on the same ideas) is included in Problem set 4 in Section \ref{Section_exercises}.

 According to one of the equivalent algorithms for finding canonical correlations and variables, we seek for a pair of vectors $\widehat \ba=(\alpha_0,\alpha_1,\dots,\alpha_{K-1})$ and $\widehat \bb=(\beta_0,\beta_1,\dots,\beta_{M-1})$, which represent critical points of the function
$$
 f(\widehat \ba,\widehat \bb)=\left\langle \alpha_0 \u^*+ \sum_{i=1}^{K-1} \alpha_i \widetilde \u_i,\, \beta_0 \v^*+ \sum_{j=1}^{M-1} \beta_j \widetilde \v_j\right\rangle,
$$
subject to normalization constraints \eqref{eq_normalization}. Introducing the Lagrange multipliers $a$ and $b$ corresponding to two normalizations, we are led to the Lagrangian function
\begin{multline}
 g(\widehat \ba,\widehat \bb,a,b)= \alpha_0 \beta_0 \langle \u^*, \v^*\rangle + \alpha_0 \sum_{j=1}^{M-1} \beta_j \langle \u^*,\widetilde \v_j\rangle +\beta_0 \sum_{i=1}^{K-1} \alpha_i \langle \v^*, \widetilde \u_i\rangle + \sum_{i=1}^{K-1} \alpha_i \beta_i \widetilde c_i
 \\- a\left( \alpha_0^2 \langle\u^*,\u^*\rangle + 2\alpha_0 \sum_{i=1}^{K-1} \alpha_i \langle \u^*, \widetilde \u_i\rangle +\sum_{i=1}^{K-1} \alpha_i^2 - 1\right)
   \\- b\left( \beta_0^2 \langle\v^*,\v^*\rangle + 2 \beta_0 \sum_{j=1}^{M-1} \beta_j \langle \v^*, \widetilde \v_j\rangle +\sum_{j=1}^{M-1} \beta_j^2 -1\right).
\end{multline}
We need to find critical points of this function. Differentiating with respect to $\alpha_i$ and $\beta_j$, we get a system of $K+M$ homogeneous linear equations on coordinates of $\widehat \ba$ and $\widehat \bb$:

\begin{equation} \label{eq_CCA_equations}
\begin{dcases}
 0=\frac{\partial g}{\partial \alpha_0}= \beta_0 \langle \u^*, \v^*\rangle + \sum_{j=1}^{M-1} \beta_j \langle \u^*, \widetilde \v_j\rangle - 2 a\alpha_0 \langle\u^*,\u^*\rangle -2 a \sum_{i=1}^{K-1} \alpha_i \langle \u^*, \widetilde \u_i\rangle,\\
 0=\frac{\partial g}{\partial \alpha_i}= \beta_0 \langle \v^*, \widetilde \u_i\rangle + \beta_i \widetilde c_i -2a \alpha_0 \langle \u^*, \widetilde \u_i \rangle - 2 a \alpha_i, & 1\le i < K,\\
 0=\frac{\partial g}{\partial \beta_0} = \alpha_0 \langle \u^*, \v^*\rangle +\sum_{i=1}^{K-1} \alpha_i \langle \v^*, \widetilde \u_i\rangle  - 2b \beta_0 \langle\v^*,\v^*\rangle - 2b \sum_{j=1}^{M-1} \beta_j \langle \v^*, \widetilde \v_j\rangle,\\
  0=\frac{\partial g}{\partial \beta_j}= \alpha_0 \langle \u^*, \widetilde \v_j\rangle + \delta_{j< K} \cdot \alpha_j \widetilde c_j - 2 b \beta_0 \langle \v^*, \widetilde \v_j\rangle - 2b \beta_j, & 1 \le j <M.\\
\end{dcases}
\end{equation}
In the matrix form, equations \eqref{eq_CCA_equations} can be rewritten as
\begin{equation} \label{eq_CCA_equations_matrix}
\begin{cases}
 S_{uv} \widehat \bb=2a S_{uu} \widehat \ba,\\
 S_{vu} \widehat \ba= 2b S_{vv} \widehat \bb,
\end{cases}
\end{equation}
for appropriate matrices $S_{uv}$, $S_{uu}$, $S_{vu}$, $S_{vv}$.
Comparing with the definition of squared canonical correlations as eigenvalues, one can deduce that $z=4ab$. Hence, it remains to solve \eqref{eq_CCA_equations}. For that we notice that the 2nd and 4th sets of equations split into $2\times 2$ blocks which can be used to express $\alpha_i$ and $\beta_i$ through $\alpha_0$ and $\beta_0$:
\begin{equation}
\label{eq_2equations}
\begin{cases}
 - 2 a \alpha_i + \widetilde c_i\beta_i  &=\,\,2a \alpha_0 \langle \u^*, \widetilde \u_i \rangle -\beta_0 \langle \v^*, \widetilde \u_i\rangle, \\
  \,\,\,\,\,\widetilde c_i\alpha_i - 2b \beta_i &=\,\,- \alpha_0 \langle \u^*, \widetilde \v_i\rangle  + 2 b \beta_0 \langle \v^*, \widetilde \v_i\rangle.
\end{cases}
\end{equation}
After plugging \eqref{eq_2equations} back into the first and third equations in \eqref{eq_CCA_equations}, we arrive at two homogeneous linear equations on $\alpha_0$ and $\beta_0$.  A non-degenerate system of two homogenous linear equations with two variables does not have non-zero solutions. Hence, in order for such a solution to exist, the system needs to be degenerate, i.e.\ the determinant of its matrix of coefficients should be equal to zero.  This condition is precisely \eqref{eq_CCA_master}.
\end{proof}

\begin{lemma} \label{Lemma_Wachter_spike_answer} In the setting of Theorem \ref{Theorem_CCA_one_spike}, the squared canonical correlations $z$, which are larger than $\lambda_+$, asymptotically solve an equation, which is a limit of \eqref{eq_CCA_master}:
\begin{equation} \label{eq_z_lim_eq}
  \frac{\displaystyle z \left[1-2\tau_K^{-1}-\frac{1}{z} \cdot (\tau_M^{-1}-\tau_K^{-1})-(1-z) \tau_K^{-1} G(z) \right]  \left[1-\tau_K^{-1}-\tau_M^{-1}-(1 - z) \tau_K^{-1} G(z) \right]}{\displaystyle  \left[1 - \tau_M^{-1}-z\tau_K^{-1}  -  z(1-z)  \tau_K^{-1} G(z)  \right]^2}
  =\rho^2,
\end{equation}
where $G$ is the Stieltjes transform of the Wachter distribution in Corollary \ref{Corollary_Wachter_Stieltjes}. If $0\le \rho^2\le\frac{1}{\sqrt{(\tau_M-1)(\tau_K-1)}}$, then there is no $z>\lambda_+$ satisfying \eqref{eq_z_lim_eq}. But if
\begin{equation}
 \label{eq_Wachter_cutoff}
  \frac{1}{\sqrt{(\tau_M-1)(\tau_K-1)}}<\rho^2 \le 1,
\end{equation}
then there is a unique $z>\lambda_+$ satisfying \eqref{eq_z_lim_eq}, denoted $z_\rho$. In this situation the relationship between $\rho$ and $z_\rho$ is:
\begin{align}
\label{eq_z_Wachter} z_\rho&=\frac{\bigl(  (\tau_K-1)\rho^2  + 1 \bigr) \bigl(  (\tau_M-1) \rho^2 + 1\bigr)}{\rho^2 \tau_K \tau_M } \qquad \text{ or, equivalently, }\\
\label{eq_z_inverse_Wachter} \rho^2&= \frac{z_\rho-\tau_M^{-1}-\tau_K^{-1}+2\tau_M^{-1}\tau_K^{-1} +\sqrt{(z_\rho-\lambda_-)(z_\rho-\lambda_+)}}{2(1-\tau_M^{-1})(1-\tau_K^{-1})}.
\end{align}
\end{lemma}
\begin{proof}[Sketch of the proof.] First, we make linear transformations of $\u$ and $\v$ vectors, so that $\ba$ turns into $(1,0,\dots,0)$ and $\bb$ turns into $(1,\dots,0)$ and all the remaining coordinates of $\u$ and $\v$ (with indices larger than one) are independent with each other and with the first coordinates.

Now we are in the context of Theorem \ref{Theorem_CCA_master_equation} where $\u^*$ is the first row of $\U$, $\v^{*}$ is the first row of $\V$, $\tilde \U$ is the span of the second, third,\dots, $K$th rows of $\U$, and $\tilde \V$ is the span of the second, third, \dots, $M$th rows of $\V$. Hence, \eqref{eq_CCA_master} holds and we need to study its large $K,M,S$ limit.

For that we further observe that $\la\u^*, \widetilde \u_i\ra$, $i=1,\dots,K-1$, are i.i.d.\ Gaussians, because $\u^*$ has i.i.d.\ Gaussian coordinates and $\u_i$ are orthonormal. Similarly, $\la\u^*, \widetilde \v_j\ra$ are i.i.d.\ Gaussians, and so are  $\la\v^*, \widetilde \u_i\ra$ and  $\la\v^*, \widetilde \v_j\ra$. Hence, we can condition on the values of $\widetilde  c_i$ and then apply the law of large numbers to the sums in \eqref{eq_CCA_master}, replacing various sums of scalar products with $\u^*$ and $\v^*$ by their (conditional) expectations. After that we get sums of expressions involving $\frac{1}{z-\widetilde c_i^2}$. Noting that $\widetilde \U$ and $\widetilde \V$ are independent, we can apply Theorem \ref{Theorem_CCA_LLN}, and recognize the appearance of $G(z)$ in the limit of the sums. In this way we arrive at \eqref{eq_z_lim_eq}.

It remains to simplify \eqref{eq_z_lim_eq} to the expressions of \eqref{eq_z_Wachter}, \eqref{eq_z_inverse_Wachter} using the explicit formula of Corollary \ref{Corollary_Wachter_Stieltjes}, which is a cumbersome, but straightforward computation, see \cite{BG3} for the further details.
\end{proof}

\subsection{Multiple signals} Theorem \ref{Theorem_CCA_one_spike} covers a situation with a single non-zero canonical correlation $c_i$ in the population and one can ask what happens when there are more. The case of finitely many non-zero $c_i$ is well understood, whereas for the case of infinitely many non-zero $c_i$ many questions remain open. In the finite case, when some of the non-zero canonical correlations $c_i$ coincide, the non-uniqueness of canonical variables is again relevant, see, e.g., Remark \ref{remark_nonuniqueness}. In addition, the centered distributional limits of $\hat c_i$ become non-Gaussian, see \cite[Theorem 1.10]{bao2019canonical}. If the canonical correlations are all distinct then essentially the same Theorem \ref{Theorem_CCA_one_spike} continues to hold for each $c_i$ individually and independently, see \cite[Theorem 1.10]{bao2019canonical} and \cite[Theorem 3.7]{BG3}. The histogram in Figure \ref{Fig_two_stocks_hist}, referenced in Chapter 1, illustrates multiple signals --- three, to be specific--- using real data. Let us explain that picture in more detail.

We calculate canonical correlations between consumer cyclical and non-cyclical (consumer defensive) stocks. The former are known to follow the state of the economy, while the latter are not related to business cycles and are useful during economic slowdowns. Knowing their correlations can be helpful for portfolio allocation. Hence, we apply CCA and search for maximally correlated combinations of cyclical and non-cyclical stocks. We use weekly returns (which are widely believed to be uncorrelated across time) for the $80$ largest cyclical and $80$ largest non-cyclical stocks over ten years ($01.01.2010-01.01.2020$), which gives us $521$ observations across time; the exact list of stocks is reported in \cite[Appendix C]{BG3}. We can see that the overall shape of the histogram of the correlations matches the Wachter distribution. In view of Theorem \ref{Theorem_CCA_LLN} from Chapter 3 and Remark \ref{Remark_Wachter_perturbation}, this indicates that the data matches our modelling assumptions and we can use the formulas of Theorem \ref{Theorem_CCA_extreme}. At the same time,  three largest correlations are clearly separated from the rest. Interestingly, the presence of three correlations is resonant with the celebrated three factors of \citet{fama1992cross}.

\bigskip

When a linear proportion of the population canonical correlations $c_i$ are nonzero, several new questions arise. In this regime, the Wachter distribution of Theorem~\ref{Theorem_CCA_LLN} no longer describes the limiting empirical distribution of the squared sample canonical correlations. A more intricate limiting law, which depends on the empirical distribution of the population canonical correlations, has recently been derived in \cite{bai2022limiting, hou2023spiked, zhang2025limiting}. These works provide explicit formulas for the limiting empirical measure of the sample canonical correlations as a functional of the population counterpart.

From a statistical perspective, the corresponding inverse problem is also of interest: given the observed sample canonical correlations and variables, one aims to recover their population analogues in the setting where many $c_i$ are nonzero. Theorem~\ref{Theorem_CCA_one_spike} suggests that consistent recovery is not possible, yet optimal estimation procedures can still be developed. This question has been extensively studied in the context of Principal Component Analysis (PCA), where such procedures are known as \emph{cleaning} or \emph{shrinkage}; see \cite{ledoit2022power} for a review. Closer to the CCA setting, optimal cleaning methods for cross-covariance matrices were developed in \cite{benaych2023optimal}. For CCA itself, however, we are not aware of any works establishing analogous optimal shrinkage or cleaning schemes.

\subsection{Subcritical signals and regularizations}

Theorem \ref{Theorem_CCA_one_spike} requires $\rho^2> \tfrac{1}{\sqrt{(\tau_M-1)(\tau_K-1)}}$, i.e., the correlation of the signal to be large enough, in order to be useful. When $\rho^2$ is smaller, by \eqref{eq_no_spike} the histogram of canonical correlations does not have spikes anymore and the theorem is of no use. The situation $\rho^2<\tfrac{1}{\sqrt{(\tau_M-1)(\tau_K-1)}}$ is called \emph{subcritical} in the literature, and one can be wondering: is there any way to infer asymptotic information on the signal, i.e., the non-zero canonical correlation and corresponding variables in this case? Is it even possible to detect the presence of a subcrticial signal?

If one sticks to spectral methods, then the most likely answer is ``no''. In this direction, assuming the data to be Gaussian, \cite{johnstone2020testing} computed the asymptotics of the ratio of the probability densities of the canonical correlations under the alternative of a single subcritical signal and under the null of no signal. They show that the ratio neither grows to $\infty$, nor decays to $0$, but rather converges to a non-trivial Gaussian process, parameterized by the value of subcritical $\rho^2$. This is an indication of impossibility of the detection with probability close to $1$ of subcritical $\rho^2$. The transition to subcritical regime at $\rho^2 \approx \tfrac{1}{\sqrt{(\tau_M-1)(\tau_K-1)}}$ for the largest eigenvalues is explored in detail in \cite{bykhovskaya2025weak}.

For more general class of methods it might be possible to improve the $\tfrac{1}{\sqrt{(\tau_M-1)(\tau_K-1)}}$ threshold and detect some subcritical $\rho^2$ when the noise is non-Gaussian. The details have not yet been figured out in the CCA setting, however, in the related PCA setup improvement can be achieved by applying a certain function (which depends on the distribution on the noise and is the identical function for the Gaussian noise --- leading to no improvements in the latter case) to each matrix element before using spectral methods, see, e.g., \citet{perry2018optimality} and references therein.

An alternative approach to overcoming the problem of the deterioration of the quality of CCA estimates as the ratio $S/K$ and $S/M$ become smaller is by imposing additional restrictions. If we have extra information about the nature of canonical variables, then we can use it by introducing \emph{regularizations} into CCA. Some of the popular approaches require the canonical vectors to be sparse (i.e.,  only a few coordinates of canonical variables are non-zero) or smooth (i.e., the coordinates depend on their labels in a continuous or smooth way). A way to achieve these regularizing requirements is by using the variational formula of Section \ref{subsection_maximiz_problem} for the canonical correlations and variables and then adding extra terms penalizing undesired features to the functions which we maximize. This is a very flexible approach, as various penalties can be used to guarantee various features of regularized canonical variables. However, a drawback is that one needs to rely on the assumptions imposed in constructing the regularization, whose validity might be hard to check in real-world data sets.

There is a large literature on various modifications of CCA and as a possible starting point for explorations of numerous methods, one can consider
\citet{gao2015minimax,gao2017sparse,uurtio2017tutorial,yang2019survey,tuzhilina2023canonical} and references therein. 

\newpage

\section{Cointegration}
Beyond its primary objective of identifying correlated components across different data sets, CCA can also be applied to analyze various properties of time series. The idea is to manually construct two rectangular data sets out of a single data set and then  implement CCA to infer some features of the data. In particular, in this section we explore how CCA can be used to test for the presence of non-exploding, or cointegrating, linear combinations within an exploding multidimensional time series.

In time series analysis data is typically allowed to be correlated across the $S$ dimension, which in this case represents time and is frequently denoted as $T$. Thus, we now depart from the i.i.d.\ column setting for  $\U$ and $\V$ used in the previous three chapters. Going forward, the matrices $\U$ and $\V$ will exhibit strong dependencies, with the columns of one matrix also being interdependent.

\subsection{Vector autoregression of order $\mathbf{1}$}

The simplest yet nontrivial model for time evolution $X=(X_0,X_1,X_2,\dots,X_T)$ is obtained by considering the vector autoregression model (without an intercept) of order $1$, VAR($1$). The model is based on a sequence of i.i.d.~mean zero vector errors $\{\eps_t\}$ with non-degenerate covariance matrix $\Lambda$. The $K$--dimensional process $X_t$ is initialized at a fixed $X_0$ and solves an evolution equation:
\begin{equation}\label{eq_var_1}
\Delta X_t=\Pi X_{t-1}+\eps_t,\qquad t=1,\ldots,T,
\end{equation}
where $\Delta X_t:=X_t-X_{t-1}$. Eq.~\eqref{eq_var_1} is called \emph{error correction} form of a vector autoregression. Eq.~\eqref{eq_var_1} has two (unknown, in principle) deterministic parameters: $K\times K$ positive definite covariance matrix $\Lambda$ and $K\times K$ coefficient matrix $\Pi$. Throughout this chapter we assume that the errors $\eps_t$ are Gaussian, however, many results are either proven or expected to extend to the non-Gaussian setting.

The properties of $X_t$ as a function of $t$ can be very different depending on the matrix $\Pi$. Let us illustrate this in the case $K=1$, so that the evolution equation \eqref{eq_var_1} becomes
\begin{equation}\label{eq_ar_1}
 x_t=\theta x_{t-1} + \eps_t, \quad t=1,2,\dots,
\end{equation}
where $\theta$ is a constant corresponding to $I+\Pi$ in \eqref{eq_var_1}. Possible behaviors of \eqref{eq_ar_1} are:
\begin{itemize}
\item If $\theta=0$, then $x_t=\eps_t$ is an i.i.d.\ sequence with no time correlations or dependencies.
\item More generally, if $|\theta|<1$, then $x_t$ has correlations in $t$, but they are short-range: the correlation coefficient between $x_t$ and $x_{t+s}$ decays exponentially in $s$. In addition, $x_t$ converges as $t\to\infty$ in distribution to a non-degenerate limiting random variable. This regime is called \emph{stationary} --- the terminology comes from the fact that one can choose an initial condition $x_0$, which would make $x_t$ a stationary process\footnote{We call a stochastic process $x_t$, $t=0,1,\dots$ stationary, if the joint distributions of the shifted process $\{x_{t+s}\}_{t=0,1,2,\dots}$ do not depend on the choice of $s=0,1,2,\dots$.} in $t$; in fact the law of this special initial condition coincides with $\lim_{t\to\infty} x_t$ for a solution started from \emph{any} initial condition. Some authors refer to such processes as $I(0)$ or integrated of order $0$.
\item If $|\theta|>1$, then $x_t$ grows exponentially in $t$.
\item In the boundary case $\theta=1$, the variance of $x_t$ grows linearly with $t$. Noting that $x_t=x_0+\eps_1+\eps_2+\dots+\eps_t$, the process can be rescaled so that\footnote{This is an application of the Functional Central Limit Theorem, see, e.g., \citet[Theorem 8.1.4]{durrett2019probability}.}
    $$
     \lim_{T\to\infty} \frac{1}{\sqrt{T}} x_{\lfloor \tau T\rfloor}= B_\tau,
    $$
    where $B_\tau$, $\tau\geq0$, is the standard Brownian motion multiplied by a standard deviation of $\eps_t$. In particular, the correlations of $x_t$ are very wide: the covariance $\E(x_t x_{t+s})$ does not decay with $s$, while the correlation between $x_t$ and $x_{t+s}$ decays as $s^{-1/2}$. This regime is called the \emph{unit root} case in the literature --- the terminology comes from association to each autoregression an evolution of a certain polynomial, which is $1-z$ in this case and therefore has a root $1$. Some authors refer to such processes as $I(1)$ or integrated of order $1$.

\item Another boundary case $\theta=-1$ is reduced to the previous one by $x_t\mapsto (-1)^t x_t$.
\end{itemize}
The autoregression model \eqref{eq_ar_1} and its various generalizations are widely adopted as a workhorse for time series modeling, see, e.g., \citet{hamilton2020time}. The reasons for that are simplicity of the linear model on one hand and diverse possible behaviors on the other hand. The $K$--dimensional extension \eqref{eq_var_1} has even richer behaviors. In particular, if $\Pi=0$, then $X_t$ is a multidimensional random walk, as in the $\theta=1$ case above. On the other hand, if $\Pi=-I_K$, then $X_t=\eps_t$ is just a sequence of i.i.d.\ random variables, as in the $\theta=0$ case. As we will see below, many intermediate cases are important for applications and have non-trivial theoretical properties.

\subsection{Cointegration} We start with an informal definition. A $K$--dimensional time-series $X_t$, $t=0,1,2,\dots,$ is said to be \emph{cointegrated}, if its individual components
have statistical properties of a unit root  --- like the $\theta=1$ case in \eqref{eq_ar_1} ---  but there exists a linear combination of components which is stationary. The number of such linearly independent stationary combinations is called the \emph{cointegration rank}.

The prefix ``co-'' in cointegration refers to a sense of ``jointly'', as in co-movements. The intuition is that all coordinates follow a long-run equilibrium defined by their stationary linear combination.  This is particularly intuitive in the $K=2$ case, where cointegration indicates that the first and second coordinates of $X_t$ tend to move together over time.

The importance of cointegration stems from the celebrated papers \cite{granger1981} and \cite{engle_granger1987}. For example, they argue that monthly rates on $1$-month and $20$-year treasury bonds are cointegrated. A lot of other variables in macroeconomics and finance, such as price level, consumption, output, trade flows, and interest rates are non-stationary, and, thus, are potentially subject to cointegration. Figure \ref{Fig_Tbills} illustrates this point with interest rates on U.S.~Treasury bills of three-month and one-year maturities. The time series on the left show wide correlations and can be modeled as unit roots, while the difference of the two series, called spread, on the right exhibits much shorter correlations, consistent with a stationary process.

\begin{figure}[t]
  \begin{subfigure}{.47\textwidth}
  \centering
  \includegraphics[width=1.0\linewidth]{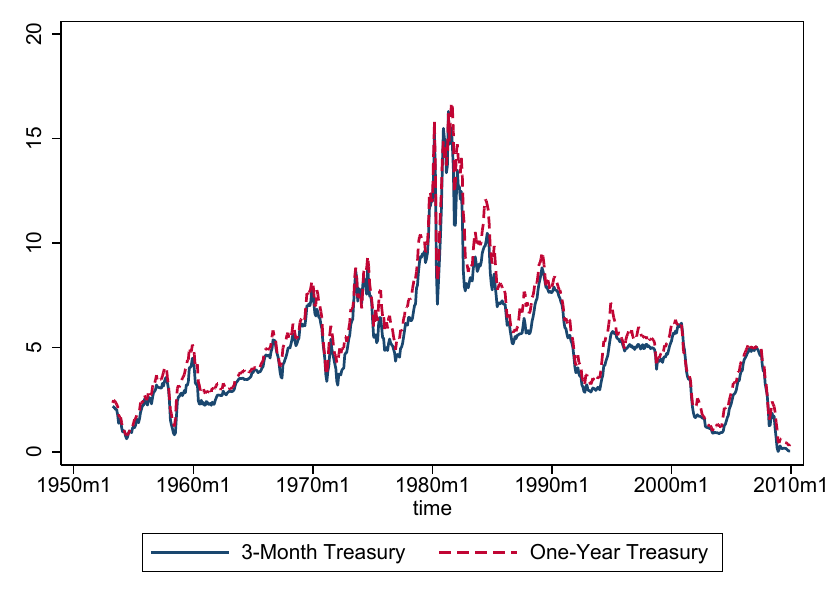}
  \end{subfigure}%
  \begin{subfigure}{.47\textwidth}
  \includegraphics[width=1.0\linewidth]{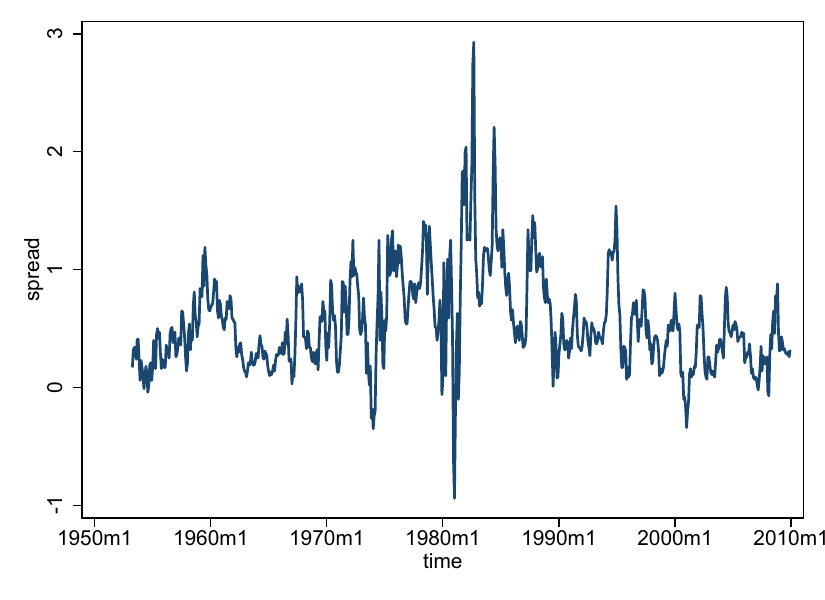}
\end{subfigure}
  \caption{Left panel: interest rates for US Treasury bills of two different maturities. Right panel: spread, i.e., the difference between the rates.}\label{Fig_Tbills}
\end{figure}

Going back to the VAR(1) model \eqref{eq_var_1}, we can restate the presence of cointegration as a property of the matrix $\Pi$.

\begin{theorem}[Granger representation theorem] \label{Theorem_Granger_representation} Suppose that the parameters $\Lambda$ and $\Pi$ in VAR(1) model \eqref{eq_var_1} are such that:
\begin{enumerate}
 \item There exists a (random) initial condition $X_0$, for which the time increments $\Delta X_t$ form a stationary process.
 \item There are exactly $r$ cointegrating relationships. This means: For a deterministic $K\times r$ matrix $\B$ of rank $r$, there exists a (random) initial condition $X_0$, for which $\B^{\T} X_t$ is an $r$--dimensional stationary process; however, there should be no such $K\times (r+1)$ matrix $\B$ of rank $r+1$
\end{enumerate}
Then the rank of the matrix $\Pi$ is $r$.
\end{theorem}
\begin{remark}
 The same statement holds if we replace stationarity under a specific initial condition with other characteristics of the $|\theta| < 1$ case in \eqref{eq_ar_1}, such as exponentially decaying correlations or the existence of a finite limit as $t \to \infty$.
\end{remark}
\begin{remark}
 If we would like each individual component of $X_t$ to either be stationary or behave like a random walk, then the stationarity requirement for the time increments $\Delta X_t$ is automatically true. This requirement excludes processes which grow faster than random walk, see $|\theta|>1$ case in \eqref{eq_ar_1} and the last exercise in Problem set 5 of Section \ref{Section_exercises} for examples.
\end{remark}
\begin{remark}
Reversing the direction and starting from a matrix $\Pi$ of rank $r$ requires additional conditions on the parameters to ensure that properties $(1)$ and $(2)$ hold. See \cite{johansen_book} for a detailed and precise analysis.
\end{remark}
\begin{proof}[Proof of Theorem \ref{Theorem_Granger_representation}]
 Let the rank of the matrix $\Pi$ be $\hat r$. Then it can be decomposed as
 $\Pi= A B,$ with a $K\times \hat r$ matrix $A$  and  a $\hat r\times K$ matrix $B$, such that $\mathrm{rank}(A)=\mathrm{rank}(K)=\hat r$.
 We plug into \eqref{eq_var_1} to get
 \begin{equation}\label{eq_var_1_rewritten}
\Delta X_t=A B X_{t-1}+\eps_t
\end{equation}
 and note that under appropriate choice of the initial condition $X_0$, both $\Delta X_t$ and $\eps_t$ are stationary. Hence, so is the remaining term $A B X_{t-1}$. Since $A$ has the maximal rank $\hat r$, this implies stationarity of $B X_{t-1}$, meaning there are $\hat r$ cointegrating relationships. Thus, $r\ge \hat r$.

 For the opposite inequality, we choose a $(K-\hat r) \times K$ matrix $C$ of rank $K-\hat r$ and such that $CA=0$. Multiplying \eqref{eq_var_1_rewritten} by $C$ on the left, we get
 $$
  C \Delta X_t= C\eps_t, \qquad
  C X_t= C X_0  + C\sum_{t'=1}^t \eps_{t'}.
 $$
 Hence, using the non-degeneracy of the covariance matrix $\Lambda$ of $\eps_t$, we conclude that there can not be stationary linear combinations of $X_t$ in the linear space spanned by the $K-\hat r$ rows of $C X_t$. Therefore, there can not be more than $\hat r$ cointegrating relationships and $r\le \hat r$.
\end{proof}

Note that for $r\ge 1$, there is no uniqueness in the choice of the cointegrating relationships, only the linear space spanned by them is uniquely determined: the matrix $\B$ in Theorem \ref{Theorem_Granger_representation} can be replaced with $Q \B$ for any non-degenerate $r\times r$ matrix $Q$. Similarly, in the decomposition $\Pi=AB$ we can replace $A\mapsto A Q^{-1}$, $B\mapsto Q B$.

\subsection{Cointegration testing} For a non-stationary data set, one might ask whether it is cointegrated and, if so, what is the rank of cointegration. That is, following Theorem \ref{Theorem_Granger_representation}, we can test if the rank of $\Pi$ is zero (i.e., $\Pi=0$), which would imply the absence of cointegration.



Similarly to Task \ref{Task_CCA_iid}, we can approach this new question via Gaussian Maximum Likelihood. Namely, assuming that the $\eps_t$ in \eqref{eq_var_1} are Gaussian and given observations $X=(X_0,X_1,\dots,X_T)$, we evaluate the likelihood function $L(X; \Lambda,\Pi)$, which computes the probability density at $X$. Next, we maximize this function over $\Lambda$ and $\Pi$ under the assumption $\mathrm{rank}(\Pi)\le r$, thus getting
\begin{equation}
\label{eq_Lr}
 L_r(X)=\max_{\begin{smallmatrix} \text{positive definite }\Lambda, \\ \Pi\text{ with }\mathrm{rank}(\Pi)\le r \end{smallmatrix}} L(X; \Lambda,\Pi).
\end{equation}
The definition readily implies the monotonicity:
$$
 L_0(X)\le L_1(X)\le \dots\le L_K(X).
$$
We expect that if $\Pi=0$ already fits the data well, maximizing over additional choices of $\Pi$ should not result in a substantial increase in $L_r(X)$ as a function of $r$. On the other hand, if the true rank of $\Pi$ is non-zero, a significant increase in $L_r(X)$ is anticipated. This leads us to the following statistical procedure, known as the \emph{likelihood ratio test}.

\begin{procedure} \label{Procedure_likelihood_ratio} Evaluate the logarithm of the likelihood ratio:
\begin{equation}
\label{eq_likelikhood_ratio}
 \ln \left(\frac{L_0(X)}{L_r(X)}\right)
\end{equation}
 If it is large and negative, then reject the hypothesis $\Pi=0$ in favor of the alternative $\Pi\ne 0$, $\mathrm{rank}(\Pi)\le r$.
\end{procedure}
The implementation of Procedure \ref{Procedure_likelihood_ratio} involves two steps. The first step is to compute \eqref{eq_likelikhood_ratio} as a function of $X$. This explicit computation goes back to \cite{Anderson} and \cite{johansen1988}, and we derive it in Theorem \ref{Theorem_likelihood_ratio_evaluation} below. The second step requires understanding the typical magnitude of \eqref{eq_likelikhood_ratio} to be able to interpret when it becomes significantly large.

\begin{theorem} \label{Theorem_likelihood_ratio_evaluation} For $X=(X_0,X_1,\dots,X_{T})$, the expression \eqref{eq_likelikhood_ratio} equals
 \begin{equation}
 \label{eq_likelikhood_ratio_eval}
   \ln\left( \frac{L_0(X)}{L_r(X)}\right)=\frac{T}{2}\sum_{i=1}^r \ln(1-\lambda_i),
 \end{equation}
 where $1\ge \lambda_1\ge\lambda_2\ge\dots\ge \lambda_K\ge 0$ are squared sample canonical correlations between $K$--dimensional subspaces $\U$ and $\V$ of $T$--dimensional vector space, defined as follows:  $\U$ is the span of $K$ rows of $(\Delta X_1, \Delta X_2,\dots,\Delta X_T)$, $\V$ is the span of $K$ rows of $(X_0,X_1,\dots,X_{T-1})$.
\end{theorem}

To understand why CCA might be relevant to identifying the rank of $\Pi$, let’s revisit \eqref{eq_var_1}. Suppose that $\eps_t$ is very small, while $\Pi$ has rank $1$. Then  $\Delta X_t=\Pi X_{t-1}+\eps_t$ implies that there is a linear combination of rows of $X_{t-1}$, which is nearly proportional to each row of $\Delta X_t$, resulting in a canonical correlation close to $1$. Similarly, if $\Pi$ has rank $r$ (still with $\eps_t$ small), there would be $r$ canonical correlations close to $1$. This suggests that the number of large canonical correlations corresponds to the rank of $\Pi$. Theorem \ref{Theorem_likelihood_ratio_evaluation} formalizes this observation and extends it to cases where $\eps_t$ is not negligible.

\begin{proof}[Proof of Theorem \ref{Theorem_likelihood_ratio_evaluation}] We solve for the maximum of \eqref{eq_Lr} to find $L_r(X)$. We represent matrix $\Pi$ of rank $r$ as $\Pi=A B$ with $K\times r$ matrix $A$ and $r\times K$ matrix $B$. Using Gaussian density from Problem 2 in Problem set 1 of Chapter \ref{Section_exercises},  the likelihood $L$ as a function of $\Lambda,A,B,X$ is proportional to
\begin{equation}
\label{eq_x10}
 (\det \Lambda)^{-T/2} \exp\left(-\frac{1}{2} \sum_{t=1}^T (\Delta X_t- A B X_{t-1})^{\T}\Lambda^{-1}(\Delta X_t- A B X_{t-1})\right).
\end{equation}
We need to maximize \eqref{eq_x10} over $\Lambda$, $A$, and $B$. We first maximize over $\Lambda$ by noticing that the logarithm of \eqref{eq_x10} can be rewritten as
\begin{equation}
\label{eq_x11}
 -\frac{T}{2}\ln \det \Lambda-\frac{1}{2} \mathrm{Trace}\left(\Lambda^{-1}\sum_{t=1}^T (\Delta X_t- A B X_{t-1}) (\Delta X_t- A B X_{t-1})^{\T} \right).
\end{equation}
We use results of Problem 2 in Problem set 2 of Chapter \ref{Section_exercises} with $\gamma=T$, $G=\Lambda$, and $D=\sum_{t=1}^T (\Delta X_t- A B X_t) (\Delta X_t- A B X_t)^{\T}$ to find the maximum of \eqref{eq_x11} and conclude that the maximum is achieved at
\begin{equation}
\label{eq_x13}
 \Lambda= \frac{1}{T} \sum_{t=1}^T (\Delta X_t- A B X_{t-1}) (\Delta X_t- A B X_{t-1})^{\T},
\end{equation}
and the maximum value is
\begin{equation}
 \label{eq_x12}
 -\frac{T}{2}\ln \det  \Lambda+\frac{K T}{2}\ln T-\frac{KT}{2}.
\end{equation}
In addition, we need to maximize \eqref{eq_x10} over $A$, which is easy because its logarithm is a quadratic function of $A$. The derivatives with respect to the matrix elements of $A$ should be zeros, which means that the maximum is achieved at parameters satisfying
\begin{equation}
 \Lambda^{-1} \sum_{t=1}^T\left(\Delta X_t-AB X_{t-1}\right) (B X_{t-1})^\T=0.
\end{equation}
Using this identity, we express $AB$ through other parameters and simplify \eqref{eq_x13}:
\begin{multline}
 \frac{1}{T}\sum_{t=1}^T \Bigl(\Delta X_t (\Delta X_t)^\T - A B X_{t-1}(\Delta X_t)^\T - \Delta X_t (B X_{t-1})^\T A^\T+ A B X_{t-1} (B X_{t-1})^\T A^\T \Bigr)\\= \frac{1}{T}\sum_{t=1}^T \Bigl(\Delta X_t (\Delta X_t)^\T - A B X_{t-1}(\Delta X_t)^\T \Bigr)=\frac{1}{T}S_{00} \cdot \Bigl( I_K-S_{00}^{-1} S_{01} S_{11}^{-1} S_{10} \Bigr),
\end{multline}
where $S_{\cdot}$ are sums of products of $\Delta X_t$ and $B X_{t-1}$:
$$
 S_{00}=\sum_{t=1}^T \Delta X_t (\Delta X_t)^\T, \quad S_{01}= \sum_{t=1}^T \Delta X_t (B X_{t-1})^\T,  \quad S_{10}=S_{01}^\T,\quad S_{11}=\sum_{t=1}^T B X_{t-1} (B X_{t-1})^\T.
$$
Plugging back into \eqref{eq_x12}, we should maximize over $B$ the expression
\begin{equation}
 \label{eq_x14}
 -\frac{T}{2}\ln \det  \Bigl( I_K-S_{00}^{-1} S_{01} S_{11}^{-1} S_{10} \Bigr)-\frac{T}{2}\ln \det S_{00}+\frac{K T}{2}\ln T-\frac{KT}{2},
\end{equation}
in which only the first term actually depends on $B$ (through $S_{01}$, $S_{10}$, $S_{11}$). Recalling that canonical correlations can be computed as eigenvalues of the product of four matrices \eqref{eq_x16}, we identify the first term with
\begin{equation}
\label{eq_x15}
 -\frac{T}{2}\sum_{i=1}^r \ln(1-\lambda_i^B),
\end{equation}
Where $\lambda_1^B,\lambda_2^B,\dots,\lambda_r^B$ are squared canonical correlations between the $K$ rows of $(\Delta X_t)_{t=1}^T$ and the $r$ rows of $(B X_{t-1})_{t=1}^T$ in $T$--dimensional space. Using the variational meaning of the canonical correlations, i.e., that they can be obtained through recursive maximizations, as in Section \ref{subsection_maximiz_problem}, and noting that the linear combinations of the rows of $(B X_{t-1})_{t=1}^T$ are also linear combinations of rows of $(X_{t-1})_{t=1}^T$, we conclude that \eqref{eq_x15} is maximized when $\lambda_i^B$ is the $i$--th largest squared sample canonical correlations between $K$--dimensional subspaces $\U$ and $\V$ claimed in the theorem. Subtracting from $\ln (L_0(X))$, which is given by the same formula \eqref{eq_x14}, but without the first term (because $B=0$ for $r=0$), we arrive at \eqref{eq_likelikhood_ratio_eval}.
\end{proof}

\medskip

For the second step in implementing Procedure \ref{Procedure_likelihood_ratio}, we draw on the ideas that partly parallel the developments in Chapter 3. Specifically, we analyze the asymptotic behavior of the sum $\sum_{i=1}^r \ln(1-\lambda_i)$ in two regimes: for fixed $K$ and large $T$, and for $K$ and $T$ being simultaneously large.

\subsection{Finite $K$, large $T$ asymptotics} We start from the classical asymptotic regime where the time series dimension $K$ is small and fixed, while the number of observations $T$ is large and is assumed to tend to infinity. Here, we aim to analyze $\sum_{i=1}^r \ln(1 - \lambda_i)$ under the null hypothesis $\Pi = 0$, which helps define ``large'' in Procedure \ref{Procedure_likelihood_ratio}, and under alternatives with non-trivial $\Pi$, to assess the asymptotic power of the test. Let us start with the null hypothesis case.


\begin{theorem}\label{Theorem_Johansen_test_small_K}
 Fix $K$ and let $\lambda_i$ be the squared sample canonical correlations  of Theorem \ref{Theorem_likelihood_ratio_evaluation}.
 Then for $X_t$ evolving according to \eqref{eq_var_1} with $\Pi=0$, we have
 \begin{equation} \label{eq_Johansen_small_K}
  \lim_{T\to\infty} (T\lambda_1,T\lambda_2,\dots, T \lambda_K)\stackrel{d}{=} (\nu_1,\nu_2,\dots,\nu_K),
 \end{equation}
 where $(\nu_i)_{i=1}^K$ are eigenvalues of the matrix $\mathcal C \mathcal V^{-1} \mathcal C^{\T}$,
 with the $K\times K$ matrices $\mathcal C$ and $\mathcal V$ defined in terms of $K$ independent standard Brownian motions $B^{i}$, $1\le i \le K$:
 \begin{equation}
   \mathcal C[i,j]=\int_0^1 B^{j}(\tau) \dd B^{i}(\tau), \qquad \mathcal V[i,j]=\int_0^1  B^{i}(\tau)  B^{j}(\tau)  \dd \tau.
 \end{equation}
\end{theorem}
\begin{remark}
$\mathcal C[i,j]$ is an Ito integral, see e.g.\ \cite{oksendal2013stochastic} for an introduction to stochastic integration. For $i=j$ it can be simplified to  $\mathcal C[i,i]=\frac{1}{2} \left( B^{i}(1) \right)^2- \frac{1}{2}$. In general, this object should be defined as a limit of integral sums:
\begin{equation}
\label{eq_Ito}
 \int_0^1 B^{j}(\tau) \dd B^{i}(\tau)=\lim_{n\to\infty} \sum_{\ell=1}^n B^{j}\left(\frac{\ell-1}{n}\right) \left[B^{i}\left(\frac{\ell}{n}\right)-B^{i}\left(\frac{\ell-1}{n}\right)\right].
\end{equation}
For $i\ne j$, conditionally on $B^{j}$, the expression under the limit is a sum of independent mean $0$ Gaussian random variables. Hence, in the limit $n\to\infty$ we get a mixed normal $\mathcal N\left(0, \int_0^1 \left(B^{j}(\tau)\right)^2 \dd \tau\right)$.
\end{remark}
\begin{proof}[Proof of Theorem \ref{Theorem_Johansen_test_small_K}]  Directly solving \eqref{eq_var_1} in the case $\Pi=0$, we get
\begin{equation}
\Delta X_t=\eps_t,\qquad X_t=X_0+\sum_{s=1}^t \eps_s,\qquad t=1,\ldots,T.
\end{equation}
The $T$ columns of the matrix $\U$ from Theorem \ref{Theorem_likelihood_ratio_evaluation}  are
$\eps_t$, $t=1,2,\dots,T$,
and the $T$ columns of the matrix $\V$ are
$X_0+\sum_{s=1}^{t-1} \eps_s$, $t=1,2,\dots,T$. Recall that $\eps_t$ are i.i.d.\ $\mathcal N(0,\Lambda)$. Using the same argument as in Theorem \ref{Theorem_not_depend_on_cov}, we see that the distribution of the squared canonical correlations $\lambda_1,\dots,\lambda_K$ does not depend on the choice of $\Lambda$. Hence, we set $\Lambda=I_K$ without loss of generality.

We compute $\lambda_i$ as the eigenvalues of $(\U \U^\T)^{-1} \U \V^\T (\V \V^\T)^{-1} \V \U^\T$, and use the Law of Large Numbers and the Functional Central Limit Theorem (for which see, e.g.,  \citet[Theorem 8.1.4]{durrett2019probability}). to deduce the joint distributional convergences:
\begin{equation}
 \lim_{T\to\infty} \frac{1}{T} (\U \U^\T)=I_K, \qquad \lim_{T\to\infty} \frac{1}{T^2}(\V \V^\T)=\mathcal V
\end{equation}
For $\U\V^\T$ we have, using the notation $\eps^i_t$ for the $i$--th coordinate of $\eps_t$:
$$
 \frac{1}{T} \U\V^{\T}[i,j]=\frac{1}{T}\sum_{t=1}^T \eps_t^i X_0^j + \frac{1}{T}\sum_{1\le \tau<t\le T}\eps_t^i \eps_\tau^j.
$$
Using the Law of Large Numbers, the first term tends to $0$ as $T\to\infty$. For the Gaussian errors $\eps_t$, the second sum is exactly the same as the right-hand side of  \eqref{eq_Ito}, by identifying $n=T$ and matching $\eps_t$ with the time-increments of the Brownian motions.\footnote{If $\eps_t$ were non-Gaussian, we would have to use an appropriate version of the Functional Central Limit Theorem here, however, the distributional limit would remain the same.} Hence, we obtain the distributional convergence to $\mathcal C[i,j]$ as $T\to\infty$. Therefore,
$$
 T (\U \U^\T)^{-1} \U \V^\T (\V \V^\T)^{-1} \V \U^\T \xrightarrow[T\to\infty]{d} \mathcal C \mathcal V^{-1} \mathcal C^{\T}.
$$
Convergence of matrix elements implies convergence of eigenvalues, see, e.g.,  \citet[Chapter 9, Theorem 6]{Lax} or  \citet[Corollary 6.3.8]{horn2012matrix}.
\end{proof}
\begin{corollary}\label{Corollary_Johansen_small_K}
  Fix $K$ and let $\lambda_i$ be the squared canonical correlations  of Theorem \ref{Theorem_likelihood_ratio_evaluation}.
 Then for $X_t$ evolving according to \eqref{eq_var_1} with $\Pi=0$, we have
 $$
  \lim_{T\to\infty} \frac{T}{2}\sum_{i=1}^r \ln(1-\lambda_i)\stackrel{d}{=}-\frac{1}{2} \sum_{i=1}^r \nu_i,
 $$
 where $\nu_i$ are as in Theorem \ref{Theorem_Johansen_test_small_K}.
\end{corollary}
\begin{proof} We apply Theorem \ref{Theorem_Johansen_test_small_K} and notice that $\ln(1-\lambda)\approx -\lambda$ for small $\lambda$.\end{proof}
Corollary \ref{Corollary_Johansen_small_K} clarifies the term ``large'' in Procedure \ref{Procedure_likelihood_ratio}: when $K$ is small and $T$ is large, the likelihood ratio \eqref{eq_likelikhood_ratio} should be compared with quantiles of the random variable $-\frac{1}{2} \sum\limits_{i=1}^r \nu_i$.

The procedure’s effectiveness --- rejecting the null of no cointegration when $\Pi \neq 0$ --- relies on the fact that asymptotically $I(0)$ (stationary) and $I(1)$ (unit root) variables are uncorrelated: The covariance between them remains bounded, $I(0)$ variables have bounded variance, while $I(1)$ variables' variance grows to infinity over time. When $\Pi=0$, no stationary linear combination $\beta^{\T}X_{t-1}$ exists, and we are left with correlations between various $I(0)$ and $I(1)$ variables, which causes all sample canonical correlations to vanish asymptotically. In contrast, when $\Pi\neq0$ a stationary linear combination exists and the corresponding sample canonical correlation does not vanish.


Analyzing the procedure under  $\Pi\neq0$ corresponds to assessing its power. Theorem \ref{Theorem_Johansen_test_small_K} and Corollary \ref{Corollary_Johansen_small_K} can be extended to allow for this scenario. Additionally, extensions accommodate higher-order autoregressions (i.e., including $X_{t-1},X_{t-2},\dots,X_{t-k}$ rather than only $X_{t-1}$ in \eqref{eq_var_1}) and inclusion of additional deterministic terms, see \cite{johansen1988,johansen1991}.
\subsection{Large $K$ regime}
Soon after Theorem \ref{Theorem_Johansen_test_small_K} and its extensions were discovered, practitioners noticed that $K$ needs to be much smaller than $T$, in order for the result to be useful. The distance between finite sample and asymptotic limit is too large for intermediate values of $K$, e.g., for $K=10$, $T=100$, see, for example, \cite{ho_sorensen1996} and \citet{gonzalo_pitarakis}. One possible way to address this distortion is by considering a different asymptotic regime. The development of alternative asymptotics based on $K$ and $T$ growing to infinity proportionally to each other is recent, see \cite{onatski_ecta}, \cite{BG1}, \cite{BG2}. It turns out that the asymptotic behavior is governed by the same objects as in Theorems \ref{Theorem_CCA_LLN} and \ref{Theorem_CCA_extreme} of Chapter 3: Wachter distribution and Airy$_1$ random point process. However, the parameters change.

\begin{theorem}\label{Theorem_Cointegration_LLN} Suppose that $X_t$ evolves according to \eqref{eq_var_1} with $\Pi=0$ and arbitrary finite initial condition $X_0$. Let $\lambda_i$ be the squared sample canonical correlations  of Theorem \ref{Theorem_likelihood_ratio_evaluation} and let $\mu_{K,T}=\frac{1}{K}\sum_{i=1}^K \delta_{\lambda_i}$ be the empirical measure of $\lambda_i$. Suppose that $K,T\to\infty$ in such a way that $T/K\to\tau>2$. Then, for the Wachter distribution $\omega$, as in Definition \ref{Definition_Wachter}, we have
 \begin{equation}
 \label{eq_Coint_LLN}
  \lim_{K,T\to\infty}  \mu_{K,T} = \omega_{1+\tau,\, (1+\tau)/2}, \qquad \text{weakly, in probability,}
 \end{equation}
 which means that for each continuous function $f(x)$, we have
 \begin{equation}
  \label{eq_Coint_LLN_f}
   \lim_{K,T\to\infty}  \mu_{K,T}[f]=\omega_{1+\tau,\, (1+\tau)/2}[f], \qquad \text{in probability.}
 \end{equation}
\end{theorem}
We illustrate Theorem \ref{Theorem_Cointegration_LLN} using a real data set. Figure \ref{Fig_SP_Wachter_data} in Chapter 1 shows a histogram of the squared sample canonical correlations $\lambda_i$, calculated for weekly stock returns in S\&P$100$ over ten years, so that $K=92$, $T=521$, see \cite[Section 9.6]{BG1} for the data description. The histogram shows a striking alignment with the density of Wachter distribution $\omega_{1+\tau,\,(1+\tau)/2}$ (orange curve), where $\tau=521/92\approx 5.66$.


Note that one should be careful with applying Theorem \ref{Theorem_Cointegration_LLN} to $\sum_{i=1}^r \ln(1-\lambda_i)$ needed for Procedure \ref{Procedure_likelihood_ratio}. The reason is that $\ln(1-x)$ explodes at $x=1$ and a priory, Theorem \ref{Theorem_Cointegration_LLN} only implies a one-sided bound on $\lim_{K,T\to\infty} \frac{1}{K}\sum_{i=1}^{r}\ln(1-\lambda_i)$ in the regime of $r$ growing linearly with $K,T$, cf.\ \cite[(6)]{onatski_ecta} and \cite[Corollary 4]{BG2}. In order to obtain more precise asymptotics, one needs better control on $\lambda_1$; this control was achieved in \cite{BG1} (see also  \cite{onatski2019extreme}) at the expense of slightly modifying the procedure for producing $\lambda_i$.

\begin{procedure} \label{Procedure_modified_coint} Given $K$-dimensional times series $X_0,X_1,\dots,X_T$, the modified likelihood ratio statistic is $\frac{T}{2}\sum_{i=1}^r \ln(1-\tilde \lambda_i)$, where $\tilde \lambda_i$ are defined the following way:
\begin{enumerate}
\item De-trend the data and define
$$
 \tilde X_t = X_{t-1} - \frac{t-1}{T} (X_T-X_0).
$$
\item  De-mean the data and define the $T$ columns of the matrix $\tilde \U$ as
 $$
 \Delta X_{t}-\frac{1}{T}\sum_{s=1}^T \Delta X_{s}, \qquad t=1,2,\dots,T,
 $$
and the $T$ columns of the matrix $\tilde \V$ as
$$
 \tilde X_{t}-\frac{1}{T}\sum_{s=1}^T \tilde X_{s}, \qquad t=1,2,\dots,T.
$$
\item Let $\tilde \lambda_1\ge \dots\ge \tilde \lambda_K$ be the squared sample canonical correlations between $\tilde \U$ and $\tilde \V$.
\end{enumerate}
\end{procedure}
This modified procedure is discussed in detail in \cite{BG1}. Conceptually, new detrending and demeaning steps remove potential non-zero means in the noise terms $\eps_t$ making the approach quite natural (in practice there is no guarantee that the true mean is zero). Technically, these steps enable us to establish an asymptotic theorem:

\begin{theorem} \label{Theorem_Cointegration_Airy} Suppose that $X_t$ evolves according to \eqref{eq_var_1} with $\Pi=0$ and arbitrary finite initial condition $X_0$. Let $\tilde \lambda_i$ be the squared sample canonical correlations of Procedure \ref{Procedure_modified_coint}. Then for each finite $r=1,2,\dots$, we have convergence in distribution for the largest correlations:
 \begin{equation}
	\label{eq_statistic_limit}
	 \frac{\sum_{i=1}^{r} \ln(1-\tilde \lambda_i)- r \cdot c_1(K,T)}{ K^{-2/3}  c_2(K,T)}  \, \xrightarrow[T,K\to\infty]{d} \sum_{i=1}^r \aa_i,
 \end{equation}
	where
	\begin{equation}
	c_1\left(K,T\right)=\ln\left(1-\lambda_+\right), \qquad
	c_2\left(K,T\right)=-\frac{2^{2/3} \lambda_+^{2/3}}{(1-\lambda_+)^{1/3} (\lambda_+-\lambda_-)^{1/3}} \left(\tau+1\right)^{-2/3}  <0,
	\end{equation}	
   \begin{equation}\label{eq_lpm_def}
	 \lambda_\pm=\frac{1}{(\tau+1)^2}\left[\sqrt{2\tau}\pm \sqrt{\tau-1}  \right]^2.
	\end{equation}
\end{theorem}
\begin{remark}
The theorem is stated for $\sum_{i=1}^r\ln(1-\lambda_i)$, but a similar result holds for individual $\lambda_i$ with $i$ not growing with $K,T$. The constants $\lambda_\pm$ are precisely the endpoints of the Wachter distribution  $\omega_{1+\tau,\, (1+\tau)/2}$, cf.\ \eqref{eq_lambda_pm_def}.
\end{remark}
Again, \eqref{eq_statistic_limit} can be immediately applied to clarify the term ``large'' in Procedure \ref{Procedure_modified_coint} (modified version of Procedure \ref{Procedure_likelihood_ratio} based on $\tilde\lambda_i$). One should use the quantiles of the random variables $\sum_{i=1}^r \aa_i$ to deduce what values are atypically large. The densities of these random variables for $r\le 3$ are plotted in Figure \ref{Fig_airy_density} and tables of distribution functions for $r\le 10$ are contained in \cite{vignette_largevars}.

\begin{figure}[t]
	\centering
	{\scalebox{0.75}{\includegraphics{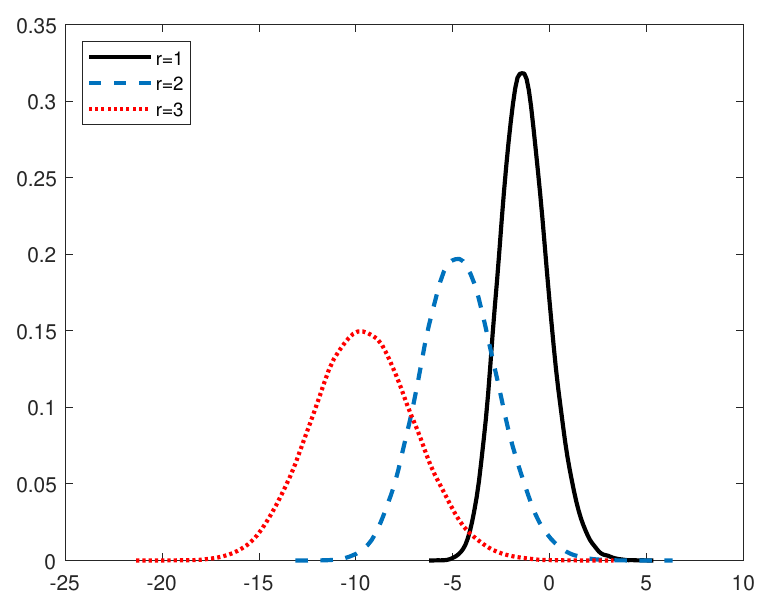}}}
	\caption{The probability density for the random variables $\sum\limits_{i=1}^r\aa_i$.}
	\label{Fig_airy_density}
\end{figure}

Theorem \ref{Theorem_Cointegration_LLN} was first proved in \cite{onatski_ecta} by using delicate manipulations with resolvents to obtain a system of four equations for the Stieltjes transform of $\lim_{K,T\to\infty}  \mu_{K,T}$ and then identifying the solution as the Wachter distribution. In this approach the Gaussianity of $\eps_t$ is not important, and it allows slight modifications of the procedure (e.g., works both for $\lambda_i$ and $\tilde \lambda_i$). An alternative proof of Theorem \ref{Theorem_Cointegration_LLN}, as well as the proof of Theorem \ref{Theorem_Cointegration_Airy} was obtained in \cite{BG1} exploring a different set of ideas, for which both Gaussianity of $\eps_t$ and the precise details of Procedure \ref{Procedure_modified_coint} are important\footnote{We expect Theorem \ref{Theorem_Cointegration_Airy} to extend to much wider generality, but this was not proven as of 2024.}. The latter approach builds on a connection to Jacobi ensemble:

\begin{theorem}
\label{Theorem_main_approximation}
In the setting of Theorem \ref{Theorem_Cointegration_Airy} suppose that $T,K\to\infty$ in such a way that $T>2K$ and the ratio $T/K$ remains bounded away from $2$ and $+\infty$. One can couple (i.e.~define on the same probability space) the $\tilde \lambda_1\ge \tilde \lambda_2\ge \dots\ge \tilde \lambda_K$ of Procedure \ref{Procedure_modified_coint} and eigenvalues $x_1\ge \dots\ge x_K$ of the Jacobi ensemble $\mathcal J(K;\frac{K}{2}, \frac{T-2K}{2})$ of Definition \ref{Definition_Jacobi_ev} in such a way that  for each $\epsilon>0$ we have
 $$
   \lim_{T,K\to\infty} \mathrm{Prob}\left( \max_{1\le i \le K} |\tilde \lambda_i-x_i|< \frac{1}{K^{1-\epsilon}}\right)=1.
 $$
\end{theorem}
Let us emphasize the differences between Theorem \ref{Theorem_main_approximation} and the previous appearance of the Jacobi ensemble in CCA in Theorem \ref{Theorem_CCA_Jacobi}. In the latter theorem, the rectangular matrices $\U$ and $\V$ are independent; in contrast, in the former theorem one matrix is a function of another: up to minor corrections, $\tilde \U$ of Procedure \ref{Procedure_modified_coint} is the time-increment of $\tilde \V$. Moreover, Theorem \ref{Theorem_CCA_Jacobi} is an exact match with the Jacobi ensemble, while Theorem \ref{Theorem_main_approximation} is only an approximation (and no exact match is known or expected). The idea of the proof of Theorem \ref{Theorem_main_approximation} is to replace the time-shift operator implicitly entering Procedure \ref{Procedure_modified_coint} by a uniformly random orthogonal operator and check that this replacement does not change the canonical correlations much, see \cite[Theorem 4]{BG1} for details.

\begin{figure}[t]
	\centering
	{\scalebox{0.75}{\includegraphics{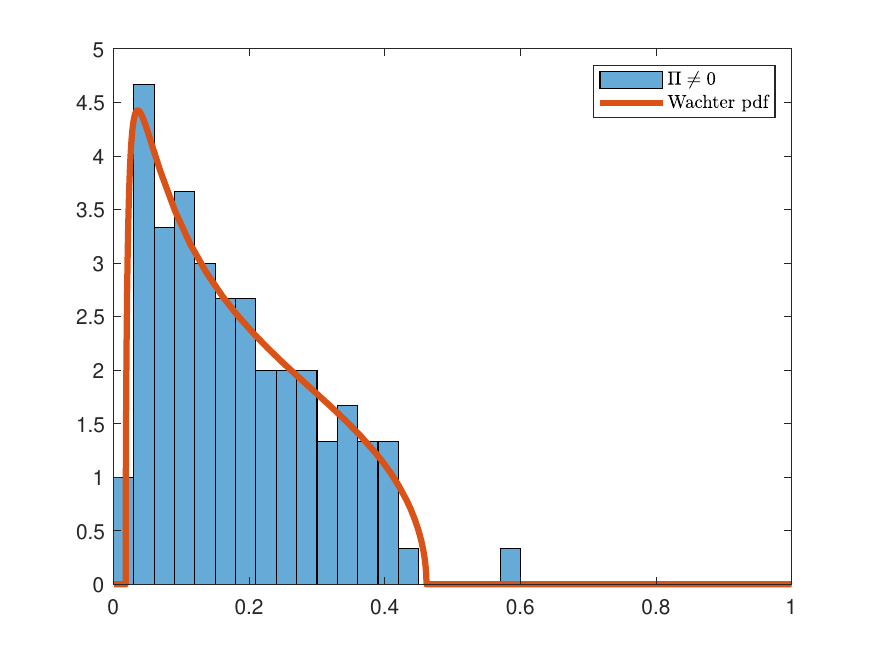 }}}
	\caption{Histogram of $\tilde \lambda_i$ for $\Pi$ of rank $1$ ($\Pi$ has single $-1$ in top-left corner), $K=100$, $T=1000$.}
	\label{Fig_cointegration_spike}
\end{figure}

Once Theorem \ref{Theorem_main_approximation} is established, Theorems \ref{Theorem_Cointegration_LLN} and \ref{Theorem_Cointegration_Airy} become applications of the asymptotic results for the Jacobi ensemble which we discussed in Chapter 3.

\subsection{Outlook: signal plus noise for cointegration} So far we only explored the large $K$, $T$ asymptotics of the objects related to cointegration tests under the null hypothesis of no cointegration, $\Pi=0$. Drawing from the ideas in Chapter 4, one might hope to develop a signal-plus-noise framework for cases where $\Pi\ne 0$. As of 2024, such a theory has yet to be fully established. Nevertheless, Monte Carlo simulations (see Figure \ref{Fig_cointegration_spike}) are encouraging and suggest that this theory may soon emerge. Looking at the real data, the presence of three correlations extending to the right of the support of the Wachter distribution for cryptocurrencies in Figure \ref{Fig_cryptovar1} of Chapter 1 likely signals the presence of three cointegrating relationships.


\newpage
\section{Exercises} \label{Section_exercises}

\subsection{Problem set 0 --- before Chapter 1} \textit{A refresher on the linear algebra}

\bigskip

\noindent {\bf Problem 1.} Consider two pairs of lines in $\mathbb R^n$: $(\u,\v)$ and $(\u',\v')$. Suppose that the angle between $\u$ and $\v$ is $0\le \phi\le \frac{\pi}{2}$ and the angle between $\u'$ and $\v'$ is $0\le \phi'\le \frac{\pi}{2}$. Prove that $\phi=\phi'$ if and only if there exists an orthogonal transformation of $\mathbb R^n$ (i.e., $n\times n$ orthogonal matrix $O$), which maps $\u$ to $\u'$ and $\v$ to $\v'$.

\bigskip

\noindent {\bf Problem 2.} Let $A$ be a $k\times n$ real matrix. Treat the rows of this matrix as vectors in $\mathbb R^{n}$, let $\U\subset R^n$ be the linear space spanned by the rows of $A$, and let $P_{\U}$ be the orthogonal projection on $\U$. Viewing $P_\U$ as a $n\times n$ matrix, find an expression for it in terms of $A$.

\smallskip

\noindent {\bf Remark.} For simplicity, feel free to assume $k\le n$ and that $A$ has rank $k$. If you want to challenge yourself, you can also try to develop a formula for the case when rank of $A$ is smaller than $k$.

\bigskip

\noindent {\bf Problem 3.} Consider two lines in $\mathbb R^n$, $\u$ and $\v$ with angle $\phi$ between them. Let $P_\u$ be orthogonal projection on $\u$ and let $P_\v$ be orthogonal projection on $\v$. Prove that $P_\v P_\u$ has a unique non-zero eigenvalue and this eigenvalue equals $\cos^2\phi$.

\bigskip

\noindent {\bf Problem 4.} Let $\xi_1,\xi_2,\dots,\xi_n$ be a random vector. Assume that $\E \xi_i=0$ and $\E \xi_i^2<\infty$ for all $1\le i \le n$. Prove that there exists a deterministic $n\times n$ matrix $A$ and a random vector $\eta_1,\eta_2,\dots,\eta_n$, such that:
\begin{itemize}
 \item $\displaystyle \begin{pmatrix} \xi_1\\ \xi_2\\ \vdots \\ \xi_n \end{pmatrix}= A \begin{pmatrix} \eta_1\\ \eta_2\\ \vdots \\ \eta_n \end{pmatrix}$\quad  and
 \item $\eta_i$, $1\le i \le n$, are mean $0$, uncorrelated, and have variance $1$, i.e.,
 $$
  \E \eta_i=0,\qquad \E [\eta_i \eta_j] = \delta_{i=j}, \qquad 1\le i,j\le n.
 $$
\end{itemize}

\newpage

\subsection{Problem set 1 --- between Chapters 1 and 2} \textit{Transformations of Gaussian vectors}

\bigskip

\noindent {\bf Problem 1.} Compute $\displaystyle \int_{-\infty}^{\infty} e^{-x^2}\, \dd x$.

\noindent {\bf Hint.} One possible approach is to square the expression and treat it as a 2-dim integral.

\bigskip

\noindent {\bf Problem 2.} Given a $n\times n$ positive-definite matrix $\Phi$, we define an $n$--dimensional mean $0$ Gaussian vector $(\xi_1,\xi_2,\dots,\xi_n)$ through its probability density with respect to the Lebesgue measure:
$$
 \rho(x_1,\dots,x_n)= \frac{1}{Z(\Phi)} \exp\left( - \frac{1}{2} \begin{pmatrix} x_1 & x_2& \dots & x_n \end{pmatrix} \Phi \begin{pmatrix} x_1\\ \vdots \\ x_n \end{pmatrix} \right)
$$

\begin{enumerate}

\item[a.] Compute the value of $Z(\Phi)$ which makes $\rho$ a probability measure.

\item[b.] Prove that $\E \xi_i=0$.

\item[c.] Let $\Lambda$ be the covariance matrix of $\xi$, i.e., $\Lambda_{i,j}=\E [\xi_i \xi_j]$. Express $\Lambda$ via $\Phi$.

\noindent {\bf Hint.} Problem 4 from the previous set of exercises might help.

\end{enumerate}

\smallskip

\noindent{\bf Notation.} Such a Gaussian vector is denoted $\mathcal N(0,\Lambda)$, where $0$ indicates that the mean is $0$ (and alternatively, if we add a deterministic vector $\mu$, then the mean would be $\mu$), and $\Lambda$ is the covariance matrix. In particular, in the one-dimensional case $\mathcal N(0,1)$ is the standard Gaussian random variable of mean $0$ and variance $1$.

\bigskip

\noindent {\bf Problem 3.} Let $(\xi_1,\dots,\xi_n)$ be a random vector with i.i.d.\ $\mathcal N(0,1)$ components. Find the probability density function of $\xi_1^2+\xi_2^2+\dots+\xi_n^2$.

\smallskip

\noindent {\bf Hint.} The answer for $n=2$ is particularly simple.

\bigskip

\noindent {\bf Problem 4.} Let $(\xi_1,\dots,\xi_n)$ be a vector with i.i.d.\ $\mathcal N(0,\sigma^2_{\xi})$ components and let $(\eta_1,\dots,\eta_n)$ be another independent vector with i.i.d.\ $\mathcal N(0,\sigma^2_{\eta})$ components. Find the probability distribution (its density function) of the squared correlation coefficient between $\xi$ and $\eta$, which is
$$
 \frac{(\xi_1 \eta_1 + \xi_2 \eta_2+\dots +\xi_n \eta_n)^2}{(\xi_1^2+\xi_2^2+\dots+\xi_n^2) \cdot (\eta_1^2+\eta_2^2+\dots+\eta_n^2)}.
$$

\smallskip

\noindent {\bf Hint.} You might try to \emph{rotate $\mathbb R^{n}$} and show first that the distribution of this squared correlation coefficient is the same as the distribution of the squared correlation coefficient between vector $\xi$ and deterministic vector $(1,0,0,\dots,0)$.

\newpage

\subsection{Problem set 2 --- between Chapters 2 and 3}\textit{ MLE and matrix algebra}

\bigskip

\noindent {\bf Problem 1.} Let $\xi_1,\xi_2,\dots,\xi_S$ be i.i.d.\ $\mathcal N(0,\sigma^2)$, where $\sigma^2$ is unknown parameter, which we would like to estimate. Write down the joint density of $\xi_1,\dots,\xi_S$ as a function $L(x_1,\dots,x_S;\sigma^2)$. Find the Gaussian MLE $\hat \sigma^2$ by maximizing
$$
 \hat \sigma^2=\mathrm{argmax}_{\sigma^2}\left[ L(\xi_1,\dots,\xi_S;\sigma^2)\right].
$$
Use the Law of Large Numbers to prove that $\lim_{S\to\infty} \hat \sigma^2 = \sigma^2$.

\bigskip

\noindent {\bf Problem 2.} Suppose that $D$ is a positive-definite symmetric $N\times N$ matrix and let $\gamma>0$. Consider a function
$$
 f(G)=-\gamma\ln\det(G)-\mathrm{Trace}(G^{-1} D).
$$
Prove that the maximum of $f(G)$ over all $N\times N$ positive-definite symmetric $N\times N$ matrices occurs at $G=\frac{1}{\gamma}D$.

\smallskip

\noindent {\bf Hint.} First, consider $D=I_N$ and represent trace and determinant through eigenvalues. Then reduce general $D$ to $D=I_N$ by a change of variables.

\bigskip

\noindent {\bf Problem 3.}  Let $\xi_1,\xi_2,\dots,\xi_S$ be i.i.d.\ $N$-dimensional vectors $\mathcal N(0,\Lambda)$, where $\Lambda$ is an unknown covariance matrix, which we would like to estimate. Write down the joint density of $\xi_1,\dots,\xi_S$ as a function $L(x_1,\dots,x_S; \Lambda)$. Find the Gaussian MLE $\hat \Lambda$ by maximizing
$$
 \hat \Lambda=\mathrm{argmax}_{\Lambda}\left[ L(\xi_1,\dots,\xi_S; \Lambda)\right].
$$
Use the Law of Large Numbers to prove that $\lim_{S\to\infty} \hat \Lambda = \Lambda$.

\smallskip

\noindent {\bf Hint.} Use Problem 2 to find $\mathrm{argmax}$.

\bigskip

\noindent {\bf Problem 4.} Let $R$ be a $N\times N$ matrix. Define
$$
 O=\exp(R)=I_N + R + \frac{R^2}{2!} + \frac{R^3}{3!} + \dots.
$$
\begin{enumerate}
  \item Prove that if $R$ is skew-symmetric, then $O$ is orthgonal.
  \item Prove that there exists an open neighborhood $\mathcal U_1$ of $0$ in the space of all skew-symmetric matrices and an open neighborhood $\mathcal U_2$ of $I_N$ in the space of all orthogonal matrices, such that $R\mapsto \exp(R)$ is a diffeomorphism between $\mathcal U_1$ and $\mathcal U_2$.
\end{enumerate}



\newpage

\subsection{Problem set 3 --- between Chapters 3 and 4} \textit{Limits under independence}

\bigskip

\noindent {\bf Problem 1.} Let $\xi$ and $\eta$ be two independent mean $0$ random variables with finite fourth moments. Let $\xi_1,\xi_2,\dots$ and $\eta_1,\eta_2,\dots$ be i.i.d.\ samples of $\xi$ and $\eta$. Let $\rho^2(S)$ denote the squared sample correlation coefficient between $(\xi_1,\dots,\xi_S)$ and $(\eta_1,\dots,\eta_S)$. Using LLN and CLT investigate the behavior of $\rho^2(S)$ as $S\to\infty$.

\smallskip

\noindent {\bf Hint.} You should get the same answer as in the Gaussian case discussed in Chapter 3.

\bigskip

\noindent {\bf Problem 2.} Recall that we showed in Chapter 2 that in the Gaussian case the single squared sample canonical correlation  $\hat{\mathfrak c}^2$ between $1\times S$ matrix $\U$ and $M\times S$ matrix $\V$ is distributed as  $\mathcal J(1; \frac{M}{2},\frac{S-M}{2})$, i.e., with Beta-density proportional to
\begin{equation*}
  x^{\frac{M-2}{2}} (1-x)^{\frac{S-M-2}{2}}, \qquad 0<x<1.
\end{equation*}
Suppose that $M,S\to\infty$ in such a way that $S/M\to \tau_M>1$. Prove that there exist constants $a$ and $b$ such that
$$
 \lim_{M,S\to\infty} \left[\frac{\hat{\mathfrak c}^2 - a}{b}\sqrt{S}\right]\stackrel{d}{=} \mathcal N(0,1).
$$
Find formulas for $a$ and $b$ in terms of $\tau_M$.

\smallskip

\noindent {\bf Hint.} Taylor-expand the density near its maximum.

\bigskip

\noindent {\bf Problem 3.} Extend the asymptotics of Problem 2 to the case when $\U$ is $K\times S$ and $K$ stays fixed as $S\to\infty$.

\bigskip

\noindent {\bf Problem 4.} Let $\xi$ be a random variable $\mathcal N(m,\sigma^2)$. Let $\mathbf {Poly}$ denote the class of all continuously differentiable functions $f(x)$ on $\mathbb R$, which grow at most polynomially as $x\to\infty$.

\begin{enumerate}
\item Prove that for each $f\in \mathbf {Poly}$ we have
\begin{equation*}
 \E \bigl[f'(\xi)\bigr]=\E \left[\frac{\xi-m}{\sigma^2} f(\xi)\right].
\end{equation*}
\item Prove that the last identity uniquely characterizes the Gaussian distribution. In other words, if $\xi$ satisfies such an identity for all $f\in\mathbf{Poly}$, then $\xi\stackrel{d}{=}\mathcal N(m,\sigma^2)$.

\smallskip

\noindent {\bf Hint.} You can choose $f$ from your favorite class of functions determining a distribution: $\{x^n\}_{n=0}^{\infty}$, $\{\frac{1}{z-x}\}_{z\in\mathbb C\setminus \mathbb R}$, $\{\exp(\ii t x)\}_{t\in\mathbb R}$ --- each of them works, although the argument using each family would be slightly different.
\end{enumerate}

\newpage

\subsection{Problem set 4 --- between Chapters 4 and 5} \textit{Spiked GOE} (simpler version of Theorem \ref{Theorem_CCA_one_spike})

\bigskip

\noindent{\bf Problem 1.} Take $N\times N$ real symmetric matrix $\B$ with distinct eigenvalues $\lambda_1>\lambda_2>\dots>\lambda_N$ and normalized eigenvectors $\u_i$, $1\le i \le N$, satisfying $\langle \u_i, \u_j \rangle =\delta_{i=j}$. In addition, take a column-vector $\u^*$. For a real constant $\theta$ define
$$
 \A=\theta\cdot  \u^* (\u^*)^\T + \B.
$$
Suppose that $\langle \u_i, \u^*\rangle\ne 0$ for all $i$. Prove that each eigenvalue $a$ of $\A$ satisfies an equation
$$
 \frac{1}{\theta }=\sum_{i=1}^N \frac{  \langle \u_i, \u^*\rangle^2}{a-\lambda_i}.
$$

\noindent {\bf Hint.} Search for an eigenvector of the form $\sum_{i=1}^N \alpha_i \u_i$.

\bigskip

\noindent{\bf Problem 2.} As continuation of Problem 1, suppose that $\u^*$ has i.i.d.\ $\mathcal N(0,\frac{1}{N})$ components. Prove that $ \langle \u_i, \u^*\rangle$, $i=1,2,\dots,N$ are also i.i.d.\ $\mathcal N(0,\frac{1}{N})$.

\bigskip

\noindent{\bf Problem 3.} Suppose that $\B$ is the rescaled GOE, which means that $\B=\frac{1}{2\sqrt{N}}(X+X^*)$, where $X$ is a matrix of i.i.d.\ $\mathcal N(0,1)$. You can use without proof the semicircle law in the form: if $\lambda_1\ge \dots\ge \lambda_N$ are eigenvalues of $\B$ and $G_N(z)=\frac{1}{N}\sum_{i=1}^N \frac{1}{z-\lambda_i}$, then:
$$\lim_{N\to\infty} \lambda_1=2, \qquad \lim_{N\to\infty}\lambda_N = -2, \qquad \lim_{N\to\infty} G_N(z)=\frac{1}{2}\bigl(z-\sqrt{z^2-4}\bigr).$$
Suppose $\theta>1$. Prove that (for $\u^*$ as in Problem 2) as $N\to\infty$, the relation between $\theta$ and $a>2$ of Problem $1$ becomes $a=\theta+\frac{1}{\theta}$.

\bigskip

\noindent{\bf Problem 4.} In the setting of Problem 3 compute the limit as $N\to\infty$ of the angle between $\u^*$ and the eigenvector of $\A$ with the largest eigenvalue $a\approx \theta + \frac{1}{\theta}$. How does it behave as $\theta\to 1$? As $\theta\to+\infty$?

\newpage

\subsection{Problem set 5 --- after Chapter 5} \textit{Examples and properties of autoregressions}
\bigskip

\noindent{\bf Problem 1.} Let $\eps_t$ be i.i.d.\ $\mathcal N(0,1)$. Fix a deterministic parameter $\theta$ and let $x_t$, $t=0,1,2,\dots$ be a
solution to
$$
 x_t=\theta x_{t-1}+\eps_t,\qquad t=1,2,\dots,\qquad x_0=0.
$$
For all values of $\theta$, compute explicitly the correlation coefficient between $x_t$ and $x_{t+s}$ as a function of $t$ and $s$. Prove that the coefficient:
\begin{enumerate}
 \item Decays exponentially as $s\to\infty$, if $|\theta|<1$.
 \item Decays polynomially as $s\to\infty$, if $|\theta|=1$.
\end{enumerate}

\bigskip

\noindent{\bf Problem 2.} In the setting of previous problem with $|\theta|<1$, find out what should be the distribution of (random) $x_0=\xi$, so that $x_t$ would become a stationary process in $t$. Prove that for any other initial condition $x_0$, one has
$$
 \lim_{t\to\infty} x_t\stackrel{d}{=} \xi.
$$

\bigskip

\noindent{\bf Problem 3.} Suppose that a two-dimensional vector $X_t$ solves an evolution equation:
$$
 \Delta X_t=\begin{pmatrix} -1/2 & -1/2 \\ -1/2 & -1/2 \end{pmatrix} X_{t-1} + \eps_t, \qquad \Delta X_t = X_t- X_{t-1}, \qquad t=1,2,\dots,
$$
where $\eps_t$ are i.i.d.\ with independent $\mathcal N(0,1)$ components. Show that $X_t$ is cointegrated, i.e.,
\begin{itemize}
 \item Both components of $X_t$ have linearly growing variance as $t\to\infty$.
 \item There exists a linear combination of two components of $X_t$ which is stationary in $t$ (under an appropriate choice of the initial condition $X_0$).
\end{itemize}

\bigskip

\noindent{\bf Problem 4.} Suppose that a two-dimensional vector $X_t$ solves an evolution equation:
$$
 \Delta X_t=\begin{pmatrix} 0 & 1 \\ 0 & 0 \end{pmatrix} X_{t-1} + \eps_t, \qquad \Delta X_t = X_t- X_{t-1}, \qquad t=1,2,\dots,
$$
where $\eps_t$ are i.i.d.\ with independent $\mathcal N(0,1)$ components. Compute the speed of growth of the variances of each of the two components of $X_t$ as $t\to\infty$.

\medskip

\noindent{\bf Remark.} We say that one of the components of $X_t$ is $I(1)$ process, while another one is $I(2)$ process; the same terminology refers to the stationary processes as $I(0)$. The value $d$ in $I(d)$ refers to how many times one needs to apply $\Delta$ to $X_t$ to transform it to a stationary process.

\newpage

\bibliographystyle{abbrvnat}
\bibliography{CCA_lectures_biblio}

\end{document}